\documentclass[twoside]{article}
\usepackage{fancyhdr}

\usepackage[top=2in, bottom=1.5in, left=1.5in, right=1.5in]{geometry}

\usepackage{latexsym}
\usepackage{xspace}
\usepackage{alltt}
\usepackage{color}
\usepackage{graphics}
\usepackage{enumitem}
\usepackage{amsmath}
\usepackage{amsthm}
\usepackage{amssymb}
\usepackage{url}
\usepackage{framed}

\begin{document}

\title{\textbf{Taming the Infinite Chase}:\\ \textbf{Query Answering}\\
  \textbf{under Expressive Relational Constraints}\footnote{This is a pre-print
    of the paper appearing in the Journal of Artificial Intelligence Research,
    vol.~48, pages~115--174, 2013, available at
    \texttt{http://www.jair.org/papers/paper3873.html}}}

\author{Andrea Cal\`\i$^{2,4}$, Georg Gottlob$^{1,4}$ and Michael Kifer$^3$\\[2em]
%         \texttt{andrea{@}dcs.bbk.ac.uk} \texttt{georg.gottlob{@}cs.ox.ac.uk}
%
 \begin{minipage}[t]{49mm}
   \small \centering
   $^1$Dept.~of Computer Science\\
   University of Oxford, UK\\
 \end{minipage}
\hfill
 \begin{minipage}[t]{52mm}
   \small \centering
   $^2$Dept.~of Computer Science\\
   Birkbeck, University of London, UK\\
 \end{minipage}\\[1,5em]
\begin{minipage}[t]{50mm}
  \small \centering
  $^3$Dept.~of Computer Science \\
  Stony Brook University, USA
\end{minipage}\hfill
\begin{minipage}[t]{62mm}
  \small \centering
  $^4$Oxford-Man Inst.~of Quantitative Finance\\
  University of Oxford, UK\\
\end{minipage}\\[1.5em]
\small \texttt{andrea{@}dcs.bbk.ac.uk}\\[-1mm]
\small \texttt{georg.gottlob{@}cs.ox.ac.uk}\\[-1mm]
\small \texttt{kifer{@}cs.stonybrook.edu} }

\date{2013}

\maketitle
\thispagestyle{empty}

\sloppy

%%%%%%%%%%%%%%%%%%%%%%%%%% General Math

\newcommand{\A}{\mathcal{A}} \newcommand{\B}{\mathcal{B}}
\newcommand{\C}{\mathcal{C}} \newcommand{\D}{\mathcal{D}}
\newcommand{\E}{\mathcal{E}} \newcommand{\F}{\mathcal{F}}
\newcommand{\G}{\mathcal{G}} \renewcommand{\H}{\mathcal{H}}
\newcommand{\I}{\mathcal{I}} \newcommand{\J}{\mathcal{J}}
\newcommand{\K}{\mathcal{K}} \renewcommand{\L}{\mathcal{L}}
\newcommand{\M}{\mathcal{M}} \newcommand{\N}{\mathcal{N}}
\renewcommand{\O}{\mathcal{O}} \renewcommand{\P}{\mathcal{P}}
\newcommand{\Q}{\mathcal{Q}} \newcommand{\R}{\mathcal{R}}
\renewcommand{\S}{\mathcal{S}} \newcommand{\T}{\mathcal{T}}
\newcommand{\U}{\mathcal{U}} \newcommand{\V}{\mathcal{V}}
\newcommand{\W}{\mathcal{W}} \newcommand{\X}{\mathcal{X}}
\newcommand{\Y}{\mathcal{Y}} \newcommand{\Z}{\mathcal{Z}}

%%%%%%%%%%%%%%%%%%%%%%%%%% Abbreviations

\newcommand{\setone}[2][1]{\set{#1\cld #2}}
\newcommand{\eset}{\emptyset}
\newcommand{\ol}[1]{\overline{#1}}                % overline
\newcommand{\ul}[1]{\underline{#1}}               % underline
\newcommand{\uls}[1]{\underline{\raisebox{0pt}[0pt][0.45ex]{}#1}}
%% ul with space between text and line

\newcommand{\ra}{\rightarrow}
\newcommand{\Ra}{\Rightarrow}
\newcommand{\la}{\leftarrow}
\newcommand{\La}{\Leftarrow}
\newcommand{\lra}{\leftrightarrow}
\newcommand{\Lra}{\Leftrightarrow}
\newcommand{\lora}{\longrightarrow}
\newcommand{\Lora}{\Longrightarrow}
\newcommand{\lola}{\longleftarrow}
\newcommand{\Lola}{\Longleftarrow}
\newcommand{\lolra}{\longleftrightarrow}
\newcommand{\Lolra}{\Longleftrightarrow}
\newcommand{\ua}{\uparrow}
\newcommand{\Ua}{\Uparrow}
\newcommand{\da}{\downarrow}
\newcommand{\Da}{\Downarrow}
\newcommand{\uda}{\updownarrow}
\newcommand{\Uda}{\Updownarrow}

%%%%%%%%%%%%%%%%%%%%%%%%%% Relations

\newcommand{\incl}{\subseteq}
\newcommand{\imp}{\rightarrow}
\newcommand{\deq}{\doteq}
\newcommand{\dleq}{\dot{\leq}}                   % dotted less equal

%%%%%%%%%%%%%%%%%%%%%%%%%% Spaces

\newcommand{\per}{\mbox{\bf .}}                  % period

\newcommand{\cld}{,\ldots,}                      % ,...,
\newcommand{\ld}[1]{#1 \ldots #1}                 % #1 ... #1
\newcommand{\cd}[1]{#1 \cdots #1}                 % #1 ... #1
\newcommand{\lds}[1]{\, #1 \; \ldots \; #1 \,}    % _#1_..._#1_
\newcommand{\cds}[1]{\, #1 \; \cdots \; #1 \,}    % _#1_..._#1_

\newcommand{\dd}[2]{#1_1,\ldots,#1_{#2}}             % x1,...,xn (da da)
\newcommand{\ddd}[3]{#1_{#2_1},\ldots,#1_{#2_{#3}}}  % xi1,...,xin (da da down)
\newcommand{\dddd}[3]{#1_{11}\cld #1_{1#3_{1}}\cld #1_{#21}\cld #1_{#2#3_{#2}}}
%%                                x_11,...,x_1n,...,x_m1,...,x_mn_m

\newcommand{\ldop}[3]{#1_1 \ld{#3} #1_{#2}}   % x1 #3...#3 xn
\newcommand{\cdop}[3]{#1_1 \cd{#3} #1_{#2}}   % x1 #3...#3 xn
\newcommand{\ldsop}[3]{#1_1 \lds{#3} #1_{#2}} % x1 _#3_..._#3_ xn
\newcommand{\cdsop}[3]{#1_1 \cds{#3} #1_{#2}} % x1 _#3_..._#3_ xn

%%%%%%%%%%%%%%%%%%%%%%%%%% Delimiters

\newcommand{\quotes}[1]{{\lq\lq #1\rq\rq}}
\newcommand{\set}[1]{\{#1\}}                      % set
\newcommand{\Set}[1]{\left\{#1\right\}}
\newcommand{\bigset}[1]{\Bigl\{#1\Bigr\}}
\newcommand{\bigmid}{\Big|}
\newcommand{\card}[1]{|{#1}|}                     % cardinality of a set
\newcommand{\Card}[1]{\left| #1\right|}
\newcommand{\cards}[1]{\sharp #1}
\newcommand{\sub}[1]{[#1]}
\newcommand{\tup}[1]{\langle #1\rangle}            % tuple
\newcommand{\Tup}[1]{\left\langle #1\right\rangle}

%%% Local Variables: 
%%% mode: latex
%%% TeX-master: "main"
%%% save-place: t
%%% End: 

%%%%%%%%%% LOCAL MACROS %%%%%%%%%%

\newcommand{\nop}[1]{}

%%%%%%%%%%%%%%%%%%%%%%%%%%%
%% MACRO ADDED BY GEORG
 
\newcommand{\bool}{\mathit{bool}} 
%%%%%%%%%% MACROS FOR THE GTGDS SECTION

\newcommand{\uu}{\vett{u}} % vector u
\newcommand{\UU}{\vett{U}} % vector u
\newcommand{\vv}{\vett{v}} % vector v
\newcommand{\VV}{\vett{V}} % vector v
\newcommand{\ww}{\vett{w}} % vector w
\newcommand{\WW}{\vett{W}} % vector W
\renewcommand{\ss}{\vett{s}} % vector s
\renewcommand{\SS}{\vett{S}} % vector S
\newcommand{\vz}{\vett{0}} % vector 0
\newcommand{\vo}{\vett{1}} % vector 1

\newcommand{\whead}{\mathit{head}}
\newcommand{\zero}{\mathit{zero}}
\newcommand{\one}{\mathit{one}}
\newcommand{\initc}{\mathit{init}}
\newcommand{\next}{\mathit{next}}
\newcommand{\config}{\mathit{config}}
\newcommand{\accept}{\mathit{accept}}
\newcommand{\accepting}{\mathit{accepting}}
\newcommand{\blankp}{\mathit{blank}}
\newcommand{\succp}{\mathit{succ}}
\newcommand{\start}{\mathit{start}}
\newcommand{\state}{\mathit{state}}
\newcommand{\existential}{\mathit{existential}}
\newcommand{\universal}{\mathit{universal}}
\newcommand{\zeroone}{\mathit{zeroone}}

%%%%%%%%%% GENERAL

\newcommand{\datalogpm}{Datalog$^\pm$}

\newcommand{\DB}{\mathit{DB}} 
\newcommand{\wrt}[0]{with respect to}

\renewcommand{\emptyset}{\varnothing} 
\renewcommand{\leq}{\leqslant} 
\renewcommand{\geq}{\geqslant} 

\newcommand{\ifdirection}{\textsl{``If''}. {}}
\newcommand{\onlyifdirection}{\textsl{``Only if''.} {}}

\newcommand{\mgu}[1]{\mathrm{mgu}(#1)}

% \newenvironment{proofsk}{\textit{Proof~(sketch).}}% 
% {$\Box$}

%%%%%%%%%% COMPUTATIONAL PROBLEMS

\newcommand{\problem}[1]{\textsf{#1}}
\newcommand{\bcqans}{\problem{BCQAns}}

%%%%%%%%%% COMMENTS BY AUTHORS IN THE DRAFT

% \newcommand{\andrea}[1]{\textcolor{blue}{#1}}
% \newenvironment{andrea}{\begin{leftbar}}{\end{leftbar}}
\newenvironment{andrea}{}{}

\newcommand{\noteac}[1]
{\noindent\framebox{\parbox{.96\columnwidth}{\textbf{Andrea:}
        #1}}} 

\newcommand{\notegg}[1]
{\noindent\framebox{\parbox{.96\columnwidth}{\textbf{Georg:}
        #1}}} 

\newcommand{\notemk}[1]
{\noindent\framebox{\parbox{.96\columnwidth}{\textbf{Michael:}
        #1}}}

%%%%%%%%%% DATABASES

\newcommand{\rel}[1]{\mathsf{#1}}
\newcommand{\attr}[1]{\mathit{#1}}
\newcommand{\const}[1]{\mathit{#1}}
\newcommand{\vett}[1]{\mathbf{#1}}

% extension of a relation
\newcommand{\ext}[2]{#1^{#2}}

%%%%%%%%%% INCOMPLETE DATABASES

\newcommand{\sol}[2]{\mathrm{sol}(#1,#2)} % sol(\dep,D)

%\newcommand{\tup}[1]{\langle #1 \rangle}

% Operator that returns the relation to which a tuple belong
% \newcommand{\Rel}[1]{\mathit{Rel}(#1)}

%%%%%%%%%% DOMAINS

\newcommand{\dom}{\Delta} 
\newcommand{\freshdom}{\Delta_N}
\newcommand{\variables}{\Delta_V} 
\newcommand{\candom}[1]{\Delta_{#1}}

\newcommand{\adom}[1]{\mathit{dom}(#1)} % active domain}
\newcommand{\hb}{\mathit{HB}}

%%%%%%%%%% DEPENDENCIES

\newcommand{\dep}{\Sigma}
\newcommand{\tdep}{\Sigma_T}
\newcommand{\edep}{\Sigma_E}
\newcommand{\fdep}{\Sigma_F}
\newcommand{\kdep}{\Sigma_K}

\newcommand{\flldep}{\Sigma_\mathit{FLL}}
\newcommand{\flldepfull}{\Sigma_\mathit{FLL}^\mathit{full}}

\newcommand{\key}[1]{\mathit{key}(#1)}
\newcommand{\isa}[1]{\mathit{ISA}}

% satisfaction of dependencies
\newcommand{\satisfy}{\models}

%%%%%%%%%% CONJUNCTIVE QUERIES

\newcommand{\vars}[1]{\mathit{vars}(#1)}
\newcommand{\varshb}[1]{\mathit{vars}^\bot(#1)}
\newcommand{\varsfresh}[1]{\mathit{vars}^+(#1)}
\newcommand{\head}[1]{\mathit{head}(#1)}
\newcommand{\body}[1]{\mathit{body}(#1)}
\newcommand{\conj}[1]{\mathit{conj}(#1)}

% predicates in conjunctive queries
\newcommand{\bodyp}[0]{\mathit{body}}  % p stands for predicate
\newcommand{\headp}[0]{\mathit{head}}

\newcommand{\ans}[3]{\mathit{ans}(#1,#2,#3)} %% ans(Q,\dep,D)

\newcommand{\ndv}{\star}

\newcommand{\atom}[1]{\underline{#1}}

\newcommand{\pred}[0]{\mathit{pred}}

%%%%%%%%%% PREDICATES

\newcommand{\data}[0]{\mathsf{data}}
\newcommand{\member}[0]{\mathsf{member}}
\newcommand{\mandatory}[0]{\mathsf{mandatory}}
\newcommand{\fsub}[0]{\mathsf{sub}}
\newcommand{\subattribute}[0]{\mathsf{subattribute}}
\newcommand{\funct}[0]{\mathsf{funct}}
\newcommand{\type}[0]{\mathsf{type}}

%%%%%%%%%% SQUID DECOMPOSITIONS AND RELATED

\newcommand{\rcover}{{$\R$}-cover}
% \newcommand{\squidd}[1]{\mathit{squidd}(#1)}

%%%%%%%%%% SETS AND SEQUENCES OF ATTRIBUTES AND VARIABLES

% \newcommand{\ins}[1]{\bar{#1}}
% \newcommand{\insA}{\ins{A}}
% \newcommand{\insB}{\ins{B}}
% \newcommand{\insC}{\ins{C}}
% \newcommand{\insD}{\ins{D}}
% \newcommand{\insX}{\ins{X}}
% \newcommand{\insY}{\ins{Y}}
% \newcommand{\insK}{\ins{K}}

% \newcommand{\vars}[1]{\bar{#1}}
% \newcommand{\varsX}{\bar{X}}
% \newcommand{\varsY}{\bar{Y}}
% \newcommand{\varsZ}{\bar{Z}}

\newcommand{\seq}[1]{\bar{#1}}

%%%%%%%%%% CHASE STUFF

% general chase
\newcommand{\chase}[2]{\mathit{chase}(#1,#2)}
% partial chase (up to step #3):
\newcommand{\pchase}[3]{\mathit{chase}^{[#3]}(#1,#2)}
% limited level chase (up to level #3):
% \newcommand{\lchase}[3]{\mathit{chase}^{[#3]}(#1,#2)}
% oblivious chase
\newcommand{\ochase}[2]{\mathit{Ochase}(#1,#2)}
% partial oblivious chase:
\newcommand{\pochase}[3]{\mathit{Ochase}^{[#3]}(#1,#2)}
% restricted chase:
\newcommand{\rchase}[2]{\mathit{Rchase}(#1,#2)}
% partial restricted chase:
\newcommand{\prchase}[3]{\mathit{Rchase}^{[#3]}(#1,#2)}
% limited level chase, up to max level
% \newcommand{\lmchase}[2]{\mathit{chase}^{[\maxlevel]}(#1,#2)}
%
% only facts of the chase with constants and vars of D
\newcommand{\chasehb}[2]{\mathit{chase}^{\bot}(#1,#2)}
% only facts of the chase with both values from D and freshly introduced ones
\newcommand{\chasefresh}[2]{\mathit{chase}^{+}(#1,#2)}
\newcommand{\blockchase}[2]{\mathit{blockchase}(#1,#2)}
\newcommand{\level}[0]{\mathrm{level}}
\newcommand{\freeze}[1]{\mathrm{freeze}(#1)}
% max level of chase segment
\newcommand{\maxlevel}[0]{\delta_{M}}
% guarded chase graph
\newcommand{\gcg}[2]{\mathrm{gcg}(#1,#2)}
% guarded chase forest
\newcommand{\gcf}[2]{\mathrm{gcf}(#1,#2)}
\newcommand{\rgcf}[2]{\mathrm{rgcf}(#1,#2)}
% treewidth
\newcommand{\tw}[1]{\mathrm{tw}(#1)}
% for special usage, just the symbol: (deprecated)
\newcommand{\GCF}[0]{\mathrm{gcf}}
% subtree of GCF (rooted at arg #1)
% \newcommand{\troot}[1]{\stackrel{\downarrow}{#1}}
\newcommand{\troot}[1]{{#1_\downarrow}}
% subtree of GCF with cloud (rooted at arg #1)
\newcommand{\trootc}[1]{\nabla#1}
% part of GCF rooted at #1 and generated using #2 as (sub)cloud
\newcommand{\ggcf}[2]{\mathrm{gcf}[#1,#2]}
% subtree of GCF with restricted cloud (rooted at arg #1)
\newcommand{\rtrootc}[1]{\nabla^r#1}

%%%%%%%%%% CLOUDS

\newcommand{\cloud}[3]{\mathit{cloud}(#1,#2,#3)}
\newcommand{\rcloud}[3]{\mathit{rcloud}(#1,#2,#3)}
\newcommand{\clouds}[2]{\mathit{clouds}(#1,#2)}
\newcommand{\cloudsplus}[2]{\mathit{clouds}^+(#1,#2)}
\newcommand{\subclouds}[3]{\mathit{subclouds}(#1,#2,#3)}
\newcommand{\subcloudsplus}[2]{\mathit{subclouds}^+(#1,#2)}

\newcommand{\can}[0]{\mathit{can}}

\newcommand{\Rel}[0]{\mathit{Rel}}

%%%%%%%%%% ALGORITHMS

\newcommand{\newatoms}[0]{\mathsf{newatoms}}
\newcommand{\oldproved}[0]{\mathsf{oldproven}}

\newcommand{\tpoint}[0]{\mathsf{Tpoint}}
\newcommand{\subst}[0]{\mathsf{subst}}

\newcommand{\atoms}[1]{\textit{atoms}(#1)}

\newcommand{\acheck}[0]{\mathsf{Acheck}}
\newcommand{\qcheck}[0]{\mathsf{Qcheck}}
\newcommand{\tcheck}[0]{\mathsf{Tcheck}}
\newcommand{\acheckbis}[0]{\mathsf{Acheck2}}
\newcommand{\qcheckbis}[0]{\mathsf{Qcheck2}}
\newcommand{\tcheckbis}[0]{\mathsf{Tcheck2}}
\newcommand{\racheck}[0]{\mathsf{rAcheck}}
\newcommand{\rqcheck}[0]{\mathsf{rQcheck}}
\newcommand{\rtcheck}[0]{\mathsf{rTcheck}}

\newcommand{\pol}[0]{\mathit{pol}}

%%%%%%%%%% PARAGRAPH

% \renewcommand{\paragraph}[1]{\textbf{#1}}
\newcommand{\ourpar}[1]{\textsc{#1}}

%%%%%%%%%% PROGRAMS AND RULES

\newcommand{\prog}{\Pi}      % generic program 

\newcommand{\guard}[1]{\Pi}      % generic program 

\newcommand{\limp}{\ra} % implies
\newcommand{\pmil}{\la} % is implied by

%%%%%%%%%% MISCELLANEOUS SYMBOLS

\newcommand{\talph}{\Lambda} % input aplhabet of a T M
\newcommand{\init}{\kappa}   % symbol for the initial configuration of a TM
\newcommand{\blank}{\flat}   % blank symbol for a TM

\newcommand{\fll}{F-Logic Lite} 

\newcommand{\select}[1]{\sigma_{\{#1\}}}

%%% Local Variables: 
%%% mode: latex
%%% TeX-master: "main"
%%% End: 

%%%%%%%%%%%%%%%%%%%%%%%%%%% ENVIRONMENTS and THEOREMS

\newcounter{cefalo}
\newcounter{sarago}
\newcounter{cefalocont}

%%% LISTS

\newenvironment{enumcont}
  {\begin{list}{(\roman{cefalo})}{\usecounter{cefalo}
     \labelwidth3em
     \itemsep0cm
     \setcounter{cefalo}{\value{cefalocont}}}}
  {\setcounter{cefalocont}{\value{cefalo}}\end{list}}

\newenvironment{rplist}
  {\begin{list}{\textit{(\roman{cefalo})}}{\usecounter{cefalo}}}
  {\end{list}}

\newenvironment{plist}
  {\begin{list}{\textit{(\arabic{cefalo})}}{\usecounter{cefalo}}}
  {\end{list}}

\newenvironment{aplist}
  {\begin{list}{\textit{(\alph{cefalo})}}{\usecounter{cefalo}}}
  {\end{list}}

%%% THEOREMS AND SIMILAR STUFF

\newtheorem{theorem}{Theorem}[section]
\newtheorem{corollary}[theorem]{Corollary}
\newtheorem{proposition}[theorem]{Proposition} 
\newtheorem{lemma}[theorem]{Lemma}

\newtheorem{definitionAux}[theorem]{Definition} 
\newenvironment{definition}{\begin{definitionAux}\upshape %\rm
}{\end{definitionAux}}

\newtheorem{exampleAux}[theorem]{Example}
\newenvironment{example}{\begin{exampleAux}\upshape}
  {\hfill\markempty\end{exampleAux}}
% \newenvironment{exampleCont}[1]{\trivlist
%   \item[\hskip \labelsep{\textbf{Example~#1 (cont.)}}]}{\endtrivlist}

% \newtheorem{examplesAux}[theorem]{Examples} 
% \newenvironment{examples}{\begin{examplesAux}\rm}{\end{examplesAux}}

% \newtheorem{constructionAux}[theorem]{Construction} 
% \newenvironment{construction}{\begin{constructionAux}\rm}
%                             {\end{constructionAux}}

%%%%%%%%%%%%
% more stuff

%%% END MARKERS

% \newenvironment{proof}{\noindent\textsc{Proof.\ }}{\markempty}
\newenvironment{proofsk}{\noindent\textit{Proof (sketch).\ }}{\hfill\markempty}

\def\qed{\hfill{\qedboxempty}      % qed with empty box
  \ifdim\lastskip<\medskipamount \removelastskip\penalty55\medskip\fi}

\def\qedboxempty{\vbox{\hrule\hbox{\vrule\kern3pt
                 \vbox{\kern3pt\kern3pt}\kern3pt\vrule}\hrule}}

\def\qedfull{\hfill{\qedboxfull}   % qed with full box
  \ifdim\lastskip<\medskipamount \removelastskip\penalty55\medskip\fi}

\def\qedboxfull{\vrule height 4pt width 4pt depth 0pt}

\newcommand{\markfull}{\qedboxfull}
\newcommand{\markempty}{\qedboxempty}

%%%%%%%%%%%%%%%%%%%%%%%%%%% BOXES

\newcommand{\fpbox}[3]{\framebox[#1]{\parbox{#2}{#3}}}

%%% Local Variables: 
%%% mode: latex
%%% TeX-master: "main"
%%% save-place: t
%%% End: 

% \renewcommand{\baselinestretch}{.99}

\begin{abstract}
  The \emph{chase} algorithm is a fundamental tool for query evaluation and for
  testing query containment under \emph{tuple-generating dependencies (TGDs)}
  and \emph{equality-generating dependencies (EGDs)}.  So far, most of the
  research on this topic has focused on cases where the chase procedure
  terminates.
  This paper introduces expressive classes of TGDs defined via syntactic
  restrictions: \emph{guarded TGDs (GTGDs)} and \emph{weakly guarded sets of
    TGDs (WGTGDs)}.  For these classes, the chase procedure is not guaranteed
  to terminate and thus may have an infinite outcome. Nevertheless, we prove
  that the problems of conjunctive-query answering and query containment under
  such TGDs are decidable.  We provide decision procedures and tight complexity
  bounds for these problems.  Then we show how EGDs can be incorporated into
  our results by providing conditions under which EGDs do not harmfully
  interact with TGDs and do not affect the decidability and complexity of query
  answering.  We show applications of the aforesaid classes of constraints to
  the problem of answering conjunctive queries in \emph{F-Logic Lite}, 
  % a recently introduced 
  an object-oriented ontology language, and in some tractable Description
  Logics.
\end{abstract}

%%% Local Variables: 
%%% mode: latex
%%% TeX-master: "main"
%%% End: 

\section{Introduction}
\label{sec:introduction}

This paper studies a simple yet fundamental rule-based language for ontological
reasoning and query answering: the language of \emph{tuple-generating
  dependencies (TGDs)}.  This formalism captures a wide variety of logics that
so far were considered unrelated to each other: the OWL-based languages
$\E\L$~\cite{BaBL05} and DL-Lite~\cite{CDLL*07,ACKZ09} on the one hand and
object-based languages like F-Logic Lite~\cite{cali-kifer-06} on the other.
The present paper is a significant extension of our earlier work \cite{CaGK08},
which has since been applied in other contexts and gave rise to the
\datalogpm~family~\cite{CaGP11} of ontology languages.  The present paper
focuses on the fundamental complexity results underlying one of the key
fragments of this family.  Subsequent work has focused on the study of various
special cases of this formalism~\cite{CaGL12}, their complexity, and extensions
based on other paradigms~\cite{CaGP12}.

Our work is also closely related to the work on query answering and query
containment~\cite{ChMe77}, which are central problems in database theory and
knowledge representation and, in most cases, are reducible to each other.  They
are especially interesting in the presence of integrity constraints---or
\emph{dependencies}, in database parlance.  In databases, query containment has
been used for query optimization and schema
integration~\cite{aho-sagiv-ullman-siam-1979,JoK84,MaLF00}, while in knowledge
representation it is often used for object classification, schema integration,
service discovery, and more~\cite{calvanese-lncs2408-2002,li03software}.

A practically relevant instance of the containment problem was first studied
in~\cite{JoK84} for functional and inclusion dependencies and later
in~\cite{calvanese98decidability}.
Several additional decidability results were obtained by focusing on concrete
applications.  For instance, \cite{CaMa10} considers constraints arising from
Entity-Relationship diagrams, while~\cite{cali-kifer-06} considers constraints
derived from a relevant subset of F-logic~\cite{flogic-new}, called F-Logic
Lite.

Some literature studies variants or subclasses of \emph{tuple-generating
  dependencies (TGDs)} for the purpose of reasoning and query answering.
%
% A TGD $\forall \vett{X} \forall \vett{Y} \Phi(\vett{X},\vett{Y}) \ra
% \exists{\vett{Z}} \Psi(\vett{X},\vett{Z})$ is a first-order formula, where
% $\Phi(\vett{X},\vett{Y})$ and $\Psi(\vett{X},\vett{Z})$ are conjunctions of
% atoms over $\R$.
%
A TGD is a Horn-like rule with existential\-ly-quantified variables in the
head.  Some early works on this subject dubbed the resulting language
\emph{Datalog with value invention}~\cite{Mail98,Cabi98}.
More formally, a TGD $\forall \vett{X} \forall \vett{Y} \Phi(\vett{X},\vett{Y})
\ra \exists{\vett{Z}} \Psi(\vett{X},\vett{Z})$ is a first-order formula, where
$\Phi(\vett{X},\vett{Y})$ and $\Psi(\vett{X},\vett{Z})$ are conjunctions of
atoms, called \emph{body} and \emph{head} of the TGD, respectively.  A TGD is
satisfied by a relational instance $B$ if
whenever the body of the TGD is satisfied by $B$ then $B$ also satisfies the
head of the TGD.  It is possible to enforce a TGD that is \emph{not} satisfied
by adding new facts to $B$ so that the head, and thus the TGD itself, will
become satisfied.  These new facts will contain \emph{labeled null values}
(short: \emph{nulls}) in the positions corresponding to variables $\vett{Z}$.
Such nulls are similar to Skolem constants.
%In the literature on chase, these Skolem constants are
%usually referred to as \emph{labeled nulls} (or simply \emph{nulls}).  The
The \emph{chase} of a database $D$ in the presence of a set $\dep$ of TGDs is
the process of iterative enforcement of all dependencies in $\dep$, until a
fixpoint is reached.  The result of such a process, which we also call
\emph{chase}, can be infinite and, in this case, this procedure cannot be used
without modifications in decision algorithms.  Nevertheless, the result of a
chase serves as a fundamental theoretical tool for answering queries in the
presence of TGDs~\cite{CaLR03,FKMP05} because it is representative of all
models of $D \cup \dep$.

% A well-known procedure that enforces the validity of a set of TGDs is the
% chase~\cite{MaMS79,JoK84}.  The chase basically applies the TGDs iteratively
% until a fixed-point is reached, that is guaranteed to satisfy all TGDs.  The
% result of the chase can be infinite in some cases.  Dealing with an infinite
% chase presents additional challenges with respect to the terminating chase.

% If a TGD is not satisfied by a database instance $D$, then it is possible to
% enforce its satisfaction by modifying $D$ and adding new atoms that satisfy
% the head.  Such new atoms may contain labeled null-values at the positions
% where existentially quantified variables are placed.

% A well-known procedure that enforces the validity of a set of TGDs is the
% \emph{chase}~\cite{MaMS79,JoK84}.  The chase procedure consists of iteratively
% applying the TGDs until a fixpoint is reached, that is guaranteed to satisfy
% all TGDs.  It is possible that the result of the chase (as well as the chase
% procedure itself) is infinite.  Dealing with an infinite chase presents
% additional challenges with respect to the terminating chase.

% this has been addressed in several works, among
% which~\cite{JoK84,CaMa10,cali-kifer-06}.

\medskip

In the present paper, we do not focus on a specific logical theory. Instead, we
tackle the common issue of the possibly non-terminating chase underlying
several of the earlier studies, including the
works~\cite{JoK84,CaMa10,cali-kifer-06}.
All these works study constraints in the language of TGDs and
\emph{equality-generating dependencies (EGDs)} using the chase technique, and
all face the problem that the chase procedure might generate an infinite
result.
We deal with this problem in a much more general way by carving out a very
large class of constraints for which the infinite chase can be ``tamed'', i.e.,
modified so that it would become a decision procedure for query answering.

\medskip

In Section~\ref{sec:decidability}, we define the notions of sets of
\emph{guarded TGDs (GTGDs)} and of \emph{weakly guarded sets of TGDs (WGTGDs)}.
A TGD is guarded if its body contains an atom called \emph{guard} that covers
all variables occurring in the body.  WGTGDs generalize guarded TGDs by
requiring guards to cover only the variables occurring at so-called
\emph{affected} positions (predicate positions that may contain some labeled
nulls generated during the chase).  Note that \emph{inclusion dependencies} (or
IDs) can be viewed as trivially guarded TGDs.  The importance of guards lies in
Theorem~\ref{theo:grid}, which shows that there is a fixed set $\dep_u$ of
GTGDs plus a single non-guarded TGD, such that query evaluation under $\dep_u$
is undecidable.
However, we show that for WGTGDs the (possibly infinite) result of the chase
has finite treewidth (Theorem~\ref{the:wgtgds-decidable}).  We then use this
result together with well-known results about the generalized tree-model
property~\cite{GoGr00,Grae99} to show that evaluating Boolean conjunctive
queries is decidable for WGTGDs (and thus also for GTGDs).  Unfortunately, this
result does not directly provide useful complexity bounds.

In Section~\ref{sec:lower}, we show lower complexity bounds for conjunctive
query answering under weakly guarded sets of TGDs.  We prove, by Turing machine
simulations, that query evaluation under weakly guarded sets of TGDs is
\textsc{exptime}-hard in case of a fixed set of TGDs, and
2\textsc{exptime}-hard in case the TGDs are part of the input.

In Section~\ref{sec:upper}, we address upper complexity bounds for query
answering under weakly guarded sets of TGDs.  Let us first remark that showing
$D \cup \dep \models Q$ is equivalent to showing that the theory $\T = D\cup
\dep \cup \{\neg Q\}$ is unsatisfiable.  Unfortunately, $\T$ is in general not
guarded because $Q$ is not and because WGTGDs are generally non-guarded
first-order sentences (while GTGDs are).  Therefore, we cannot
(as one might think at first glance) directly use known results on
guarded logics~\cite{GoGr00,Grae99} to derive complexity results for query
evaluation.  We thus develop completely new algorithms by which we
prove that the problem in question is \textsc{exptime}-complete in case of
bounded predicate arities and, even in case the TGDs are fixed,
2\textsc{exptime}-complete in general.

In Section~\ref{sec:guarded}, we derive complexity results for reasoning with
GTGDs.  In the general case, the complexity is as for WGTGDs but,
interestingly, when reasoning with a \emph{fixed} set of dependencies (which is
the usual setting in data exchange and in description logics), we get much
better results: evaluating Boolean queries is \textsc{np}-complete and
is in \textsc{ptime} in case the query is atomic.  Recall that Boolean query
evaluation is \textsc{np}-hard even in case of a simple database without
integrity constraints~\cite{ChMe77}.  Therefore, the above \textsc{np} upper
bound for general Boolean queries is optimal, i.e., there is no class of
TGDs for which query evaluation (or query containment) is more efficient.

In Section~\ref{sec:polycloud}, we describe a semantic condition on weakly
guarded sets of TGDs.
% that we call the \emph{Polynomial Clouds Criterion (PCC)}.
%
% A set of WGTDs meets this criterion, if the number of clouds it
% generates during a chase is polynomial in the size of the input database
% instance $D$, and if the cloud of each generated atom $a$ can be obtained in
% polynomial time from $a$'s parent in the chase.
We prove that whenever a set of WGTGDs fulfills this condition, answering
Boolean queries is in \textsc{np}, and answering atomic queries, as well as
queries of bounded treewidth, is in \textsc{ptime}.

Section~\ref{sec:multiple-atoms} extends our
results to the case of TGDs with multiple-atom heads.
The extension is trivial for all cases except for the case of bounded
predicate arity.

Section~\ref{sec:egds} deals with \emph{equality generating dependencies
  (EGDs)}, a generalization of functional dependencies.  Unfortunately, as
shown in~\cite{ChV85,Mitch83,JoK84,Koch02,CaLR03}, query answering and many
other problems become undecidable in case we admit both TGDs and EGDs.  It
remains undecidable even if we mix the simplest class of guarded TGDs, namely,
inclusion dependencies, with the simplest type of EGDs, namely functional
dependencies, and even key dependencies~\cite{ChV85,Mitch83,JoK84,CaLR03}.
% Some sufficient conditions for tractability are known for the case of TGDs
% plus key constraints.
In Section~\ref{sec:egds}, we present a sufficient semantic condition for
decidability of query-answering under sets of TGDs and general EGDs.  We call
EGDs \emph{innocuous} when, roughly speaking, their application (i.e.,
enforcement) does not introduce new atoms, but only eliminates atoms.  We show
that innocuous EGDs can be essentially \emph{ignored} for conjunctive query
evaluation and query containment testing.

% Finally, in Section~\ref{sec:applications}, we show how our results can be
% applied to a non-trivial set of TGDs and EGDS.  In particular, we show that
% \fll~\cite{cali-kifer-06}, a meaningful fragment of
% F-Logic~\cite{flogic-new} can be handled by our approach.

The TGD-based ontology languages in this paper are part of the larger family of
ontology languages called \datalogpm~\cite{CaGP11}.  Our results subsume % both
the main decidability and \textsc{np}-complexity result in~\cite{JoK84}, the
decidability and complexity results on F-Logic Lite in~\cite{cali-kifer-06},
and those on DL-Lite as special cases.  In fact, Section~\ref{sec:applications}
shows that our results are even more general than that.

The complexity results of this paper, together with some of their immediate
consequences, are summarized in Figure~\ref{fig:results}, where all complexity
bounds are tight.  Notice that the complexity in the case of fixed queries
\emph{and} fixed TGDs is the so-called \emph{data complexity}, i.e., the
complexity with respect to the data only, which is of particular interest in
database applications.  The complexity for variable Boolean conjunctive queries
(BCQs) and variable TGDs is called \emph{combined complexity}.  It is easy to
see (but we will not prove it formally for all classes) that all complexity
results for atomic or fixed queries extend to queries of bounded width, where
by \emph{width} we mean treewidth or even hypertree width~\cite{GoLS02}---see
also~\cite{AdGG07,GoLS01}.
% By ``bounded width'' we mean bounded treewidth or even hypertree
% width~\cite{GoLS02}.

\begin{figure}[tb]
  \centering
  \begin{tabular}{|l||c|c|c|}
    \hline
    \textbf{BCQ type} & \textbf{GTGDs} &
    \textbf{WGTGDs}\\ \hline\hline
    \textbf{general} & 2\textsc{exptime} & 2\textsc{exptime}\\ \hline
      \textbf{atomic or fixed} & 2\textsc{exptime} & 2\textsc{exptime}\\
    \hline
%
    % \textbf{fixed} & 2\textsc{exptime} & 2\textsc{exptime} & \textsc{ptime} &
    % \textsc{exptime}\\ \hline
%
    % \textbf{atomic or fixed} & 2\textsc{exptime} & 2\textsc{exptime} & \textsc{ptime}
    % & \textsc{exptime}\\ \hline
  \end{tabular}\\[.3em]
  Query answering for variable TGDs.\\[.9em]

  \begin{tabular}{|l||c|c|c|}
    \hline
    \textbf{BCQ type} & \textbf{GTGDs} &
    \textbf{WGTGDs}\\ \hline\hline
    \textbf{general} & \textsc{np} & \textsc{exptime}\\ \hline
      \textbf{atomic or fixed} & \textsc{ptime} & \textsc{exptime}\\
    \hline
%
    % \textbf{fixed} & 2\textsc{exptime} & 2\textsc{exptime} & \textsc{ptime} &
    % \textsc{exptime}\\ \hline
%
    % \textbf{atomic} & 2\textsc{exptime} & 2\textsc{exptime} & \textsc{ptime}
    % & \textsc{exptime}\\ \hline
  \end{tabular}\\[.3em]
  Query answering for fixed TGDs.\\[.9em]

  \begin{tabular}{|l||c|c|c|}
    \hline
    \textbf{BCQ type} & \textbf{GTGDs} &
    \textbf{WGTGDs}\\ \hline\hline
    \textbf{general} & \textsc{exptime} & \textsc{exptime}\\ \hline
      \textbf{atomic or fixed} & \textsc{exptime} & \textsc{exptime}\\
    \hline
%
    % \textbf{fixed} & 2\textsc{exptime} & 2\textsc{exptime} & \textsc{ptime} &
    % \textsc{exptime}\\ \hline
%
    % \textbf{atomic} & 2\textsc{exptime} & 2\textsc{exptime} & \textsc{ptime}
    % & \textsc{exptime}\\ \hline
  \end{tabular}\\[.3em]
  Query answering for fixed predicate arity.

\caption{Summary of results. All complexity bounds are tight.}
\label{fig:results}
\end{figure}

%%% Local Variables: 
%%% mode: latex
%%% TeX-master: "main"
%%% End: 

\section{Preliminaries}
\label{sec:preliminaries}

In this section we define the basic notions that we use throughout the paper.
% Additional definitions are given in Appendix~\ref{sec:app-preliminaries}.

\subsection{Relations, Instances and Queries}

A \emph{relational schema} $\R$ is a set of relational predicates, each having
an \emph{arity}---a non-negative integer that represents the number of
arguments the predicate takes. We write $r/n$ to say that a relational
predicate $r$ has arity $n$.  % Also, $\arity{r}$ will denote the arity of $r$.
Given an $n$-ary predicate $r\in\R$, a \emph{position} $r[k]$, where $1 \leq k
\leq n$, refers to the $k$-th argument of $r$.  We will assume an underlying
relational schema $\R$ and postulate that all queries and constraints use only
the predicates in $\R$.  The schema $\R$ will sometimes be omitted when it is
clear from the context or is immaterial.

We introduce the following pairwise disjoint sets of symbols: \textit{(i)} A
(possibly infinite) set $\dom$ of data \emph{constants}, which constitute the
``normal'' domain of the databases over the schema $\R$; \textit{(ii)} a set
$\freshdom$ of \textit{labeled nulls}, i.e., ``fresh'' Skolem constants; and
\textit{(iii)} an infinite set $\variables$ of variables, which are used in
queries and constraints.  Different constants represent different values
(\emph{unique name assumption}), while different nulls may represent the same
value.  We also assume a lexicographic order on $\dom \cup \freshdom$, with
every labeled null in $\freshdom$ following all constant symbols in $\dom$.
Sets of variables (or sequences, when the order is relevant) will be denoted by
$\vett{X}$, i.e., $\vett{X} = \dd{X}{k}$, for some $k$. 
% , where each $X_i\in\variables$.
The notation $\exists \vett{X}$ is a shorthand for $\exists X_1 \ldots \exists
X_k$, and similarly for $\forall \vett{X}$.

An \emph{instance} of a relational predicate $r/n$ is a (possibly infinite) set
of atomic formulas (atoms) of the form $r(\dd{c}{n})$, where $\set{\dd{c}{n}}
\subseteq \dom \cup \freshdom$.  Such atoms are also called \emph{facts}.  When
the fact $r(\dd{c}{n})$ is true, we say that the \emph{tuple} $\tup{\dd{c}{n}}$
belongs to the instance of $r$ (or just that it is in $r$, if confusion does
not arise).
An instance of the relational schema $\R = \set{\dd{r}{m}}$ is the set
comprised of the instances of $\dd{r}{m}$.  When instances are treated as
first-order formulas, each labeled null is viewed as an existential variable
with the same name, and relational instances with nulls correspond to a
conjunction of atoms preceded by the existential quantification of all the
nulls.
For instance, $\set{r(a, z_1, z_2, z_1), s(b,z_2,z_3)}$, where
$\set{z_1,z_2,z_3} \subseteq \freshdom$ and $\set{a,b} \subseteq \dom$, is
expressed as $\exists z_1 \exists z_2 \exists z_3\, r(a, z_1, z_2, z_1) \land
s(b,z_2,z_3)$.  In the following, we will omit these
quantifiers.

A fact $r(\dd{c}{n})$ is said to be \emph{ground} if $c_i \in \dom$ for all $i
\in \set{1, \ldots, n}$.  In such a case, also the tuple $\tup{\dd{c}{n}}$ is
said to be ground.  A relation or schema instance all of whose facts are ground
is said to be ground, and a ground instance of $\R$ is also called a
\emph{database}.

If $A$ is a \emph{sequence} of atoms $\tup{\dd{\atom{a}}{k}}$ or a
\emph{conjunction} of atoms $\atom{a}_1 \land \ldots \land \atom{a}_k$, we use
$\atoms{A}$ to denote the \emph{set} of the atoms in $A$: $\atoms{A} =
\set{\dd{\atom{a}}{k}}$.
Given a (ground or non-ground) atom $\atom{a}$, the \emph{domain} of
$\atom{a}$, denoted by $\adom{\atom{a}}$, is the set of all values (variables,
constants or labeled nulls) that appear as arguments in $\atom{a}$. If $A$ is a
set of atoms, we define $\adom{A} = \bigcup_{\atom{a} \in A} \adom{\atom{a}}$.
If $A$ is a sequence or a conjunction of atoms then we define $\adom{A} =
\adom{\atoms{A}}$.  If $A$ is an atom, a set, a sequence, or a conjunction of
atoms, we write $\vars{A}$ to denote the set of variables in $A$.

%% Added 09 Sept 2013
Given an instance $B$ of a relational schema $\R$, the \emph{Herbrand Base} of
$B$, denoted $\hb(B)$, is the set of all atoms that can be formed using the
predicate symbols of $\R$ and arguments in $\adom{B}$.  Notice that this is an
extension of the classical notion of Herbrand Base, which includes ground atoms
only.

An $n$-ary \emph{conjunctive query (CQ)} over $\R$ is a formula of the form
$q(\dd{X}{n}) \la \Phi(\vett{X})$, where $q$ is a predicate not appearing in
$\R$, % all symbols in $\vett{X}$ are in $\variables \cup \dom$,
all the variables $\dd{X}{n}$ appear in $\vett{X}$, and $\Phi(\vett{X})$,
called the \emph{body} of the query, is a conjunction of atoms constructed with
predicates from $\R$.  The arity of a query is the arity of its head predicate
$q$.  If $q$ has arity $0$, then the conjunctive query is called \emph{Boolean}
(BCQ).  For BCQs, it is convenient to drop the head predicate and simply view
the query as the \emph{set} of atoms in $\Phi(\vett{X})$.
If not stated otherwise, we assume that queries contain no constants, since
constants can be eliminated from queries by a simple polynomial time
transformation.  We will also sometimes refer to conjunctive queries by just
``queries''.  The \emph{size} of a conjunctive query $Q$ is denoted by $|Q|$;
it represents the number of atoms in $Q$.

% Henceforth, we shall deal with database instances, or simply databases, that
% may contain labeled nulls in $\freshdom$ as values, together with constants
% in $\dom$.

% Database instances, or simply databases, will be constructed with values from
% $\dom \cup \freshdom$.
%
% By an \emph{atom} we mean an atomic formula of the form $P(a_1,...,a_n)$,
% where $P$ is an $n$-ary relation name.  We will use the terms ``relation
% name'' and ``predicate'' interchangeably.  The constants appearing in an atom
% $\atom{a}$ will be denoted by $\adom{\atom{a}}$. This notation extends to
% sets and conjunctions of atoms.
%
% A \emph{position} $P[i]$ in a relational schema is identified by a relational
% predicate $P$ and its $i$-th attribute, identified by the integer $i$.

\subsection{Homomorphisms}

A \emph{mapping} from a set of symbols $S_1$ to another set of symbols $S_2$
can be seen as a function $\mu: S_1 \ra S_2$ defined as follows: \textit{(i)}
$\emptyset$ (the empty mapping) is a mapping; \textit{(ii)} if $\mu$ is a
mapping, then $\mu \cup \{X \ra Y\}$, where $X \in S_1$ and $Y \in S_2$ is a
mapping if $\mu$ does not already contain some $X \ra Y'$ with $Y \neq Y'$.  If
$X \ra Y$ is in a mapping $\mu$, we write $\mu(X) = Y$.
The notion of a mapping is naturally extended to atoms as follows.  If $\atom{a}
= r(\dd{c}{n})$ is an atom and $\mu$ a mapping, we define $\mu(\atom{a}) =
r(\mu(c_1), \ldots, \mu(c_n))$.  For a \emph{set} of atoms, $A =
\set{\dd{\atom{a}}{m}}$, $\mu(A) = \set{\mu(\atom{a}_1), \ldots,
  \mu(\atom{a}_m)}$.  The set of atoms $\mu(A)$ is also called \emph{image} of
$A$ \wrt~$\mu$.
For a \emph{conjunction} of atoms $C = \atom{a}_1 \land \ldots \land
\atom{a}_m$, $\mu(C)$ is a shorthand for $\mu(\atoms{C})$, that
is, $\mu(C) = \set{\mu(\atom{a}_1), \ldots, \mu(\atom{a}_m)}$.

A \emph{homomorphism} from a set of atoms $A_1$ to another set of atoms $A_2$,
with $\adom{A_1 \cup A_2} \subseteq \dom \cup \freshdom \cup \variables$
% both over the same relational schema $\R$,
is a mapping $\mu$ from $\adom{A_1}$ to $\adom{A_2}$
% $\dom \cup \freshdom \cup \variables$ to $\dom \cup \freshdom \cup
% \variables$
such that the following conditions hold:
\textit{(1)} if $c \in \dom$ then $\mu(c) = c$;
\textit{(2)}
$\mu(A_1) \subseteq A_2$, i.e.,
if an atom, $\atom{a}$, is in $A_1$, then the atom $\mu(\atom{a})$ is in
$A_2$.
In this case, we will say that $A_1$ \emph{maps} to $A_2$ via $\mu$.

The answer to a conjunctive query $Q$ of the form $q(\dd{X}{n}) \la
\Phi(\vett{X})$ over an instance $B$ of $\R$, denoted by $Q(B)$, is defined as
follows: a tuple $\vett{t} \in (\dom \cup \freshdom)^n$, is in $Q(B)$ iff there
is a homomorphism $\mu$ that maps $\Phi(\vett{X})$ to atoms of $B$, and
$\tup{\dd{X}{n}}$ to $\vett{t}$.  In this case, by abuse of notation, we also
write $q(\vett{t}) \in Q(B)$.  A Boolean conjunctive query $Q$ has a
\emph{positive} answer on $B$ iff $\tup{}$ (the tuple with no elements) is in
$Q(B)$; otherwise, it is said to have a \emph{negative} answer.

\subsection{Relational Dependencies}
\label{sec:dependencies}

We now define the main type of dependencies used in this paper, the
\emph{tuple-generating dependencies}, or \emph{TGDs}.

\begin{definition}\label{def:tgd}\rm
  Given a relational schema $\R$, a TGD $\sigma$ over $\R$ is a first-order
  formula of the form $\forall \vett{X} \forall \vett{Y}
  \Phi(\vett{X},\vett{Y}) \ra \exists{\vett{Z}} \Psi(\vett{X},\vett{Z})$, where
  $\Phi(\vett{X},\vett{Y})$ and $\Psi(\vett{X},\vett{Z})$ are conjunctions of
  atoms over $\R$, called \emph{body} and \emph{head} of the TGD, respectively;
  they are denoted by $\body{\sigma}$ and $\head{\sigma}$.  Such a dependency
  is satisfied in an instance $B$ of $\R$ if, whenever there is a homomorphism
  $h$ that maps the atoms of $\Phi(\vett{X},\vett{Y})$ to atoms of $B$, there
  exists an \emph{extension} $h_2$ of $h$ (i.e., $h_2 \supseteq h$) that maps
  the atoms of $\Psi(\vett{X},\vett{Z})$ to atoms of $B$.
\end{definition}

\begin{andrea}
  To simplify the notation, we will usually omit the universal quantifiers in
  TGDs.  We will also sometimes call TGDs \emph{rules} because of the
  implication symbol in them.  Notice that, in general, constants of~$\dom$ can
  appear not only in the body, but also in the heads of TGDs.  For simplicity
  and without loss of generality, we assume that all constants that appear in
  the the head of a TGDs also appear in the body of the same TGD.
  % This is
  % because one can always add a new unary relation that contains all the
  % missing
  % constants.
\end{andrea}

\begin{andrea}
  The symbol $\models$ will be used henceforth for the usual logical
  entailment, where sets of atoms and TGDs are viewed as first-order
  theories.  For such theories, we do not restrict ourselves to finite models:
  we consider arbitrary models that could be finite or infinite.  This aspect
  is further discussed in Section~\ref{sec:conclusions}.
\end{andrea}

\subsection{Query Answering and Containment under TGDs}
\label{sec:answering-containment}

We now define the notion of \emph{query answering} under TGDs.  A similar
notion is used in data exchange~\cite{FKMP05,GoNa06} and in query answering
over incomplete data~\cite{CaLR03}.
Given a % an incomplete database, i.e.,
database that does not satisfy all the constraints in $\dep$, we first define
the set of completions (or \emph{repairs}---see~\cite{ArBC99}) of that
database, which we call \emph{solutions}.

\begin{definition}\label{def:solutions}
  Consider a relational schema $\R$, a set of TGDs $\dep$, and a database $D$
  for $\R$.  The set of instances $\set{B \,\mid\, B \models D \cup \dep}$
  % (which in this case amounts to $B \supseteq D$ and $B \models \dep$),
  is called the set of \emph{solutions} of $D$ given $\dep$, and is denoted by
  $\sol{D}{\dep}$.
\end{definition}

% \begin{definition}\label{def:answers}
%   Consider a relational schema $\R$, a set of TGDs $\dep$, and a database
%   instance $D$ for $\R$.  The \emph{answers} to a conjunctive query $Q$ on
%   $D$ given $\dep$, denoted $\ans{Q}{D}{\dep}$, is the set of ground atoms
%   $\atom{a}$ such that for every $B \in \sol{D}{\dep}$, $\atom{a} \in Q(B)$
%   holds.
% \end{definition}

The following is the definition of the problem, which we denote by
\problem{CQAns}, of answering conjunctive queries under TGDs. The answers
defined here are also referred to as \emph{certain answers}
(see~\cite{FKMP05}).

\begin{definition}\label{def:answers}
  Consider a relational schema $\R$, a set of TGDs $\dep$, a database $D$ for
  $\R$, and a conjunctive query $Q$ on $\R$.  The \emph{answer} to a
  conjunctive query $Q$ on $D$ given $\dep$, denoted by $\ans{Q}{D}{\dep}$, is
  the set of tuples $\vett{t}$ such that for every $B \in \sol{D}{\dep}$,
  $\vett{t} \in Q(B)$ holds.
\end{definition}
Notice that the components of $\vett{t}$ in the above definition are
necessarily constants from $\dom$.  When $\vett{t} \in \ans{Q}{D}{\dep}$, we
also write $D \cup \dep \cup \set{Q} \models q(\vett{t})$, where % , as usual,
$Q$ is represented as a rule $\body{Q} \rightarrow q(\vett{X})$.
% and $\vett{X}$ is a sequence of variables of the same length as $\vett{t}$.

\medskip

Containment of queries over relational databases has long been considered a
fundamental problem in query optimization, especially query containment under
constraints such as TGDs. % and EGDs.
Below we formally define this problem, which we call \problem{CQCont}.

\begin{definition}\label{def:containment}
  Consider a relational schema $\R$, a set $\dep$ of TGDs on $\R$, and two
  conjunctive queries $Q_1, Q_2$ expressed over $\R$.  We say that \emph{$Q_1$
    is contained in $Q_2$ under $\dep$}, denoted by $Q_1 \subseteq_\dep Q_2$,
  if for every instance $B$ for $\R$ such that $B \models \dep$ we have
  $Q_1(B)$ is a subset of $Q_2(B)$.
\end{definition}

\subsection{The Chase}
\label{sec:chase}

The \emph{chase} was introduced as a procedure for testing implication of
dependencies~\cite{MaMS79}, but later also employed for checking query
containment~\cite{JoK84} and query answering on incomplete data under
relational dependencies~\cite{CaLR03}.
Informally, the chase procedure is a process of repairing a database with
respect to a set of dependencies, so that the result of the chase satisfies the
dependencies.  By ``chase'' we may refer either to the chase procedure or to
its output.  The chase works on a database through the so-called TGD
\emph{chase rule}, which defines the result of the applications of a TGD and
comes in two flavors: \emph{oblivious} and \emph{restricted}.

%\begin{definition}[Oblivious applicability] \label{def:o-applicability}
%  Consider a relational instance $B$ of a schema $\R$ (with domain $\dom
%  \cup \freshdom$) and a TGD $\sigma$ on $\R$ of the form
%  $\Phi(\vett{X},\vett{Y}) \ra \exists \vett{Z}\,\Psi(\vett{X},\vett{Z})$.  We
%  say that $\sigma$ is \emph{obliviously applicable} to $B$ if there exists a
%  homomorphism $h$ that maps the atoms of $\Phi(\vett{X},\vett{Y})$ to atoms of
%  $B$.
%\end{definition}

\begin{definition}[Oblivious Applicability]\label{def-o-applicability}
Consider an instance $B$ of a schema $\R$, and a TGD $\sigma =
\Phi(\vett{X},\vett{Y}) \rightarrow \exists
\vett{Z}\,\Psi(\vett{X},\vett{Z})$ over $\R$. We say that $\sigma$
is \emph{obliviously applicable} to $B$ if there exists a
homomorphism $h$ such that $h(\Phi(\vett{X},\vett{Y})) \subseteq B$.
%Let $B'$ be the instance $B \cup h'(\psi(\vett{X},\vett{Z}))$, where
%$h'$ is an extension of $h|_{\vett{X}}$ such that $h'(Z)$ is a
%``fresh'' labeled null of $\freshdom$ not occurring in $B$, and
%following lexicographically all those in $B$, for each $Z \in
%\vett{Z}$. We say that the result of applying $\sigma$ to $B$ with
%$h$ is $B'$, and write $B \stackrel{\sigma,h}{\longrightarrow}_{O}
%B'$; in fact, $B \stackrel{\sigma,h}{\longrightarrow}_{O} B'$
%defines a single oblivious TGD chase step.
\end{definition}

%\begin{definition}[Restricted applicability] \label{def:r-applicability}
%  Under the previous assumptions,
%  the TGD $\sigma$ is \emph{restrictively applicable} to $B$ if
%  there is a homomorphism $h$ that maps the atoms of
%  $\Phi(\vett{X},\vett{Y})$ to the atoms of $B$ and no extension
%  $h'$ of $h$ (a homomorphism such that $h' \supseteq h$) maps the
%  atoms of $\Psi(\vett{X},\vett{Z})$ to tuples of $B$.
%\end{definition}

\begin{definition}[Restricted Applicability]\label{def:r-applicability}
  Consider an instance $B$ of a schema $\R$, and a TGD $\sigma =
  \Phi(\vett{X},\vett{Y}) \rightarrow \exists
  \vett{Z}\,\Psi(\vett{X},\vett{Z})$ over $\R$. We say that $\sigma$ is
  \emph{restrictively applicable} to $B$ if there exists a homomorphism $h$
  such that $h(\varphi(\vett{X},\vett{Y})) \subseteq B$, but there is \emph{no}
  extension $h'$ of $h|_{\vett{X}}$ such that $h'(\psi(\vett{X},\vett{Z}))
  \subseteq B$.\footnote{$h|_{\vett{X}}$ denotes the restriction of $h$ to the set of variables of $\vett{X}$.} 
  % Let $B'$ be the instance $B \cup h'(\psi(\vett{X},\vett{Z}))$, where $h'$
  % is an extension of $h|_{\vett{X}}$ such that $h'(Z)$ is a ``fresh'' labeled
  % null of $\freshdom$ not occurring in $B$, and following lexicographically
  % all those in $B$, for each $Z \in \vett{Z}$. We say that the result of
  % applying $\sigma$ to $B$ with $h$ is $B'$, and write $B
  % \stackrel{\sigma,h}{\longrightarrow}_{O} B'$; in fact, $B
  % \stackrel{\sigma,h}{\longrightarrow}_{O} B'$ defines a single oblivious TGD
  % chase step.
\end{definition}

The oblivious form of applicability is called this way because it ``forgets''
to check whether the TGD is already satisfied.  In contrast, a TGD is
restrictively applicable only if it is not already satisfied.

%\begin{definition}[TGD Chase Rule] \label{tgd-chase-rule}
%%
%  Let $\sigma \in \dep$ be as above and suppose it is applicable to an
%  instance $B$ via a homomorphism $h$.  Let $h_1$ be a homomorphism that
%  extends $h$ as follows: for each $X_i\in \vett{X} \cup \vett{Y}$, $h_1(X_i)
%  = h(X_i)$; for each $Z_j\in \vett{Z}$, $h_1(Z_j)=z_j$, where $z_j$ is a
%  ``fresh'' null, i.e., $z_j\in \freshdom$, and $z_j$ lexicographically
%  follows all other labeled nulls already introduced.
%% and it is different from them.
%%
%  The application of $\sigma$ to $B$ adds
%%  \footnote{Being an instance a set of atoms, if some added atom is already
%%  in $B$ the addition has no effect.}
%  to $B$ all atoms in $h_1(\Psi(\vett{X},\vett{Z}))$ that are not already in
%  $B$.
%%  that are not already in $B$.
%\end{definition}

\begin{definition}[TGD Chase Rule]\label{tgd-chase-rule}
  Let $\sigma$ be a TGD of the form $\Phi(\vett{X},\vett{Y}) \rightarrow
  \exists \vett{Z}\,\Psi(\vett{X},\vett{Z})$ and suppose that it is obliviously
  (resp., restrictively) applicable to an instance $B$ via a homomorphism
  $h$. Let $h'$ be an extension of $h|_{\vett{X}}$ such that, for each $Z \in
  \vett{Z}$, $h'(Z)$ is a ``fresh'' labeled null of $\freshdom$ not occurring
  in $B$, and following lexicographically all those in $B$. The result of the
  \emph{oblivious} (resp., \emph{restricted}) \emph{application} of $\sigma$ on
  $B$ with $h$ is $B' = B \cup h'(\Psi(\vett{X},\vett{Z}))$. We write $B
  \stackrel{\sigma,h}{\longrightarrow}_{O} B'$ (resp., $B
  \stackrel{\sigma,h}{\longrightarrow}_{R} B'$) to denote that $B'$ is obtained
  from $B$ through a single oblivious (resp., restricted) chase step.
\end{definition}

% \textsc{Restricted TGD Chase Rule.}  \label{tgd-rchase-rule} %
% Let $\sigma$ be applicable and $h_1$ be an extension of $h$ with the above
% properties; the result of the application of $\sigma$ is the addition to $D$
% those atoms of $h_1(\Psi(\vett{X},\vett{Z}))$ that are not already in
% $D$. \markfull
%
% Thus, the only difference between the oblivious and the restricted cases is
% that the latter applies stricter criteria to the applicability of TGDs.
%
% It is important to keep in mind that, in the application of the TGD rule, an
% atom is added to the database only if it is not in the part of the chase
% constructed so far, and that an application of a chase rule may change the
% database.

\smallskip

The TGD chase rule, defined above, is the basic building block to construct the
chase of a database under a set of TGDs.  Depending on the notion of
applicability in use---oblivious or restricted---we get the oblivious or the
restricted chase. The formal definition of the chase is given below.

%How the chase procedure is performed is inductively defined below.

\begin{definition}[Oblivious and Restricted Chase]\label{def:chase}
  Let $D$ be a database and $\dep$ a set of TGDs.  An \emph{oblivious (resp.,
    restricted) chase sequence} of $D$ with respect to $\dep$ is a sequence of
%   oblivious (resp., restricted) chase steps
  instances $B_0, B_1, B_2, \ldots$ such that $B_0 = D$ and, for all $i
  \geqslant 0$, $B_i\stackrel{\sigma_i,h_i}{\longrightarrow}_{O} B_{i+1}$
  (resp., $B_i \stackrel{\sigma_i,h_i}{\longrightarrow}_{R} B_{i+1}$) and
  $\sigma_i \in \dep$.  We also assume that in any chase sequence the same pair
  $\tup{\sigma_i,h_i}$ is never applied more than once.  The oblivious (resp.,
  restricted) chase of $D$ with respect to $\dep$, denoted
  $\mathit{Ochase}(D,\dep)$ (resp., $\mathit{Rchase}(D,\dep)$), is defined as
  follows:
\begin{itemize}\itemsep-\parsep
\item A \emph{finite oblivious (resp., restricted) chase} of $D$ with respect
  to $\dep$ is a finite oblivious (resp., restricted) chase sequence $B_0,
  \ldots, B_m$ such that $B_i \stackrel{\sigma_i,h_i}{\longrightarrow}_{O}
  B_{i+1}$ (resp., $B_i \stackrel{\sigma_i,h_i}{\longrightarrow}_{R} B_{i+1}$)
  for all $0 \leqslant i < m$, and there is no $\sigma \in \dep$ such that its
  application yields an instance $B' \neq B_m$.  We define
  $\mathit{Ochase}(D,\dep) = B_m$ (resp., $\mathit{Rchase}(D,\dep) = B_m$).
\item An \emph{infinite oblivious (resp., restricted) chase sequence} $B_0,
  B_1, \ldots$, where $B_i \stackrel{\sigma_i,h_i}{\longrightarrow}_{O}
  B_{i+1}$ (resp., $B_i \stackrel{\sigma_i,h_i}{\longrightarrow}_{R} B_{i+1}$)
  for all $i \geqslant 0$, is \emph{fair} if whenever a TGD $\sigma =
  \Phi(\vett{X},\vett{Y}) \rightarrow \exists
  \vett{Z}\,\Psi(\vett{X},\vett{Z})$ of $\dep$ is obliviously (resp.,
  restrictedly) applicable to $B_i$ with homomorphism $h$, then there exists an
  extension $h'$ of $h|_{\vett{X}}$ and $k > 0$ such that $h'(\head{\sigma})
  \subseteq B_k$.
  An \emph{infinite oblivious chase} of $D$ with respect to $\dep$ is a fair
  infinite chase sequence $B_0, B_1, \ldots$ such that $B_i
  \stackrel{\sigma_i,h_i}{\longrightarrow}_{O} B_{i+1}$ (resp., $B_i
  \stackrel{\sigma_i,h_i}{\longrightarrow}_{R} B_{i+1}$) for all $i \geqslant
  0$.  In this case, we define
%  $\mathit{Ochase}(D,\dep) = \bigcup_{i=0}^{\infty} B_i$.
  $\mathit{Ochase}(D,\dep) = \lim_{i \ra \infty} B_i$ (resp.,
  $\mathit{Rchase}(D,\dep) = \lim_{i \ra \infty} B_i$).
\end{itemize}
\end{definition}

%\begin{definition}[Chase]\label{def:chase}
%  Let $D$ be a database and $\dep$ a set of TGDs.
%  \begin{itemize}\itemsep-\parsep
%  \item The \emph{oblivious (resp., restricted) chase up to derivation
%      level~$0$}, denoted $\pochase{D}{\dep}{0}$ (resp.,
%    $\prchase{D}{\dep}{0}$), is defined as $D$.
%  \item The \emph{oblivious (resp., restricted) chase up to derivation
%      level~$k$}, denoted $\pochase{D}{\dep}{k}$ (resp.,
%    $\prchase{D}{\dep}{k}$), where $k \geq 1$, is constructed as follows.
%    %%
%    Choose a TGD, $\sigma$, and a homomorphism, $h$, such that $\sigma$ is
%    applicable to $\pochase{D}{\dep}{k-1}$ via $h$.  The choice must be made
%    according to some deterministic strategy (e.g., according to some linear
%    order compatible with the lexicographic order).  Apply the chase rule for
%    $\sigma$ and $h$ to $\pochase{D}{\dep}{k-1}$ (obliviously or restrictively,
%    respectively) and assign the derivation level $k$ to every newly introduced
%    atom.
%  \item The \emph{oblivious} (resp., \emph{restricted}) chase
%    $\ochase{D}{\dep}$ (resp., $\rchase{D}{\dep}$) is defined as the limit of
%    $\pochase{D}{\dep}{k}$ (resp., $\prchase{D}{\dep}{k}$) for $k \ra \infty$.
%  \end{itemize}
%%
%\end{definition}
%
It is easy to see that the chase can be infinite, if the sequence of
applications of the chase rule is infinite.  We remark that the chase was
defined for databases of ground tuples.  However, the definition
straightforwardly applies also to arbitrary instances, possibly containing
labeled nulls.
% even those that contain labeled nulls.
We assume a fair deterministic strategy for constructing chase sequences.  We
use $\pochase{D}{\dep}{i}$ (resp., $\prchase{D}{\dep}{i}$) to denote the result
of the $i$-th step of the oblivious (resp., restricted) chase of $D$ with
respect to $\dep$.  Notice that $\pochase{D}{\dep}{i}$ (resp.,
$\prchase{D}{\dep}{i}$) is called the oblivious (resp., restricted) chase of
$D$ \wrt~$\dep$ \emph{up to the derivation level $i$}, as in~\cite{CaGL12}.

\begin{example}\label{exa:chase} 
  In this example, % , taken from~\cite{CaGL12},
  we show an oblivious chase
  procedure. %(which, in this case, coincides with the restricted chase).
  Consider the following set $\dep = \set{\sigma_1, \sigma_2, \sigma_3,
    \sigma_4}$ of TGDs.
\[
  \begin{array}{crcl}
    \sigma_1\colon& r_3(X,Y) &\ra &r_2(X)\\
    \sigma_2\colon& r_1(X,Y) &\ra& \exists Z\, r_3(Y,Z)\\
    \sigma_3\colon& r_1(X,Y), r_2(Y) &\ra& \exists Z\, r_1(Y,Z)\\
    % \sigma_3\colon& r_1(X,Y), r_2(Y) &\ra& r_1(Y,X)\\
    \sigma_4\colon& r_1(X,Y) &\ra &r_2(Y)
  \end{array}
\]  
and let $D = \set{r_1(a,b)}$.  The chase procedure adds to $D$ the following
sequence of atoms: $r_3(b,z_1)$ via $\sigma_2$, $r_2(b)$ via $\sigma_4$,
$r_1(b,z_2)$ via $\sigma_3$, $r_3(z_2,z_3)$ via $\sigma_2$, $r_2(z_2)$ via
$\sigma_4$, and so on.
% $r_3(b,z_1)$ via $\sigma_2$, $r_2(b)$ via $\sigma_1$, $r_1(b,a)$ via
% $\sigma_3$, $r_3(a,z_2)$ via $\sigma_2$ and $r_2(a)$ via $\sigma_1$.
%
% At this point the chase terminates as no new atoms can be added.
\end{example}

% In this case, it is important that applications of the TGD chase rule
% introduce atoms with the minumum possible level, so that the construction
% proceeds in a breadth-first fashion.
%
% For the moment, we put EGDs aside and deal only with TGDs.  We shall retake
% EGDs in Section~\ref{sec:egds}.

\subsection{Query Answering and the Chase}

The problems of query containment and answering under TGDs are closely related
to each other and to the notion of chase, as explained below.

% In the folllowing definition, explain that we need the existence of a
% homomorphism $\mu: \body{Q} \ra \rchase{D}{\dep}$.  Probably the definition
% with the query $Q$ treated as a rule should be replaced.

\begin{andrea}
  \begin{theorem}[see~\cite{NaDR06}]\label{the:ans-by-chase}
    % Consider a relational schema $\R$, a the:muset $\dep$ of TGDs on $\R$, a
    % conjunctive query $Q$ with head-predicate $q$, and a tuple $\atom{t}$; we
    % have that $\atom{t} \in \ans{Q}{D}{\dep}$ iff $\rchase{D}{\dep} \cup Q
    % \models q(\atom{t})$, where $Q$ here is treated as a rule (TGD).
    Consider a relational schema $\R$, a database $D$ for $\R$, a set $\dep$ of
    TGDs on $\R$, an $n$-ary conjunctive query $Q$ with head-predicate $q$, and
    an $n$-ary ground tuple $\vett{t}$ (with values in $\dom$). Then $\vett{t}
    \in \ans{Q}{D}{\dep}$
    % also written $D \cup \dep \cup \set{Q} \models q(\vett{t})$,
    iff there exists a homomorphism $h$ such that $h(\body{Q}) \subseteq
    \rchase{D}{\dep}$ and $h(\head{Q}) = q(\vett{t})$.
    % $\rchase{D}{\dep} \cup \set{Q} \models q(\atom{t})$, where $Q$ here is
    % treated as
    % a rule (TGD).
  \end{theorem}
  Notice that the fact that $h(\body{Q}) \subseteq \rchase{D}{\dep}$ and
  $h(\head{Q}) = q(\vett{t})$ is equivalent to saying that $q(\vett{t}) \in
  Q(\rchase{D}{\dep})$, or that
  % if we consider $Q$ as a rule (TGD) of the form $\body{Q} \ra \head{Q}$,
  % that
  $\rchase{D}{\dep} \cup \set{Q} \models q(\vett{t})$.
\end{andrea}
The result of Theorem~\ref{the:ans-by-chase} is important, and it holds because
the (possibly infinite) restricted chase is a \emph{universal
  solution}~\cite{FKMP05}, i.e., a representative of all instances in
$\sol{D}{\dep}$.
More formally, a universal solution for $D$ under $\dep$ is a (possibly
infinite) instance $U$ such that, for every instance $B \in \sol{D}{\dep}$,
there exists a homomorphism that maps $U$ to $B$.  In~\cite{NaDR06} it is shown
that the chase constructed with respect to TGDs is a universal solution.

A \emph{freezing homomorphism} for a query is a homomorphism that maps every
distinct variable in the query into a distinct labeled null in $\freshdom$.
The following well known result is a slight extension of a result
in~\cite{ChMe77}.
%
% The following is superfluous: Then we say that $\lambda(\body{Q_1})$ is a set
% of atoms obtained by freezing the atoms in the body of $Q_1$.
%
\begin{theorem} \label{the:cont-by-chase} Consider a relational schema $\R$, a
  set $\dep$ of TGDs on $\R$, and two conjunctive queries $Q_1, Q_2$ on $\R$.
  Then $Q_1 \subseteq_\dep Q_2$ iff $\lambda(\head{Q_1}) \in
  Q_2(\rchase{\lambda(\body{Q_1})}{\dep}$ for some freezing homomorphism
  $\lambda$ for $Q_1$.
\end{theorem}
From this and the results in~\cite{JoK84,NaDR06}, we easily obtain the
following result, which is considered folklore.
\begin{corollary}\label{cor:answering-containment}
  The problems \problem{CQAns} and \problem{CQCont} are mutually
  \textsc{logspace}-reducible.
% was \textsc{ptime}-reducible.
%
  % Under TGDs, the (decision) problem of query answering on incomplete data,
  % and
  % the problem of query containment are mutually \textsc{ptime}-reducible.
\end{corollary}

%%%%%%%%%%%%%%%%%%%%%%%%%%%%%%%%%%%%%%%%%%%%%%%%%%%%%%%%%%%%%%%%%%%%%%%%%%%%%

\subsection{Oblivious vs. Restricted Chase}

As observed in~\cite{JoK84} in the case of functional and inclusion
dependencies, things are more complicated if the restricted chase is used
instead of the oblivious one, since applicability of a TGD depends on the
presence of other atoms previously added to the database by the chase.  It is
technically easier to use the oblivious chase and it can be used in lieu of the
restricted chase because, as we shall prove now, a result similar to
Theorem~\ref{the:ans-by-chase} holds for the oblivious chase, i.e., it is also
universal.  This result, to the best of our knowledge, has never been
explicitly stated before.  For the sake of completeness, we present a full
proof here.

\begin{theorem}\label{the:oblivious-is-universal}
  Consider a set $\dep$ of TGDs on a relational schema $\R$, and let $D$ be a
  database on $\R$.  Then there exists a homomorphism $\mu$ such that
  $\mu(\ochase{D}{\dep}) \subseteq \rchase{D}{\dep}$.
\end{theorem}

\begin{proof}
  The proof is by induction on the number $m$ of applications of the TGD
  chase rule in the construction of the oblivious chase $\ochase{D}{\dep}$.
%
  %% This was defined earlier
%%
  % Let us consider the initial segment of the (oblivious) chase, obtained
  % starting from $D$ and by applying $m$ times the TGD chase rule, and let us
  % denote it by $\pochase{D}{\dep}{m}$ to $\rchase{D}{\dep}$.
%
  We want to prove that, for all $m$ with $m \geq 0$, there is a homomorphism
  from $\pochase{D}{\dep}{m}$ to $\rchase{D}{\dep}$.

  \noindent
  \textsl{Base case.}  In the base case, where $m=0$, no TGD rule has yet been
  applied, so $\pochase{D}{\dep}{0} = D \subseteq \rchase{D}{\dep}$ and the
  required homomorphism
  is simply the identity homomorphism $\mu_0$.

  \noindent
  \textsl{Inductive case.} Assume we have applied the TGD chase rule $m$
  times and
  obtained
  $\pochase{D}{\dep}{m}$.  By the induction hypothesis, there exists a
  homomorphism $\mu_m$ that maps $\pochase{D}{\dep}{m}$ into
  $\rchase{D}{\dep}$.  Consider the $(m+1)$-th application of the TGD chase
  rule, for a TGD of the form
  % \[
  $\Phi(\vett{X}, \vett{Y}) \ra \exists\vett{Z} \Psi(\vett{X}, \vett{Z})$.
  % \]
%
  By definition of applicability of TGDs, there is a
  homomorphism $\lambda_O$ that maps $\Phi(\vett{X}, \vett{Y})$ to atoms of
  $\ochase{D}{\dep}$ and that it can be suitably extended to another
  homomorphism,
  $\lambda'_O$, such that $\lambda'_O$ maps each
  of the variables in $\vett{Z}$ to a fresh null in $\freshdom$ not already
  present in $\pochase{D}{\dep}{m}$. As a result of the application of this
  TGD, all atoms in
  $\lambda'_O(\Psi(\vett{X}, \vett{Z}))$ are added to $\pochase{D}{\dep}{m}$,
  thus obtaining $\pochase{D}{\dep}{m+1}$.
  Consider the homomorphism $\lambda_R = \mu_m \circ \lambda_O$, which maps $\Phi(\vett{X},
  \vett{Y})$ to atoms of $\rchase{D}{\dep}$.
  Since $\rchase{D}{\dep}$ satisfies all the
  dependencies in $\dep$ (and so does $\ochase{D}{\dep}$), there is an
  extension $\lambda'_R$ of $\lambda_R$ that maps $\Psi(\vett{X}, \vett{Z})$ to
  tuples of $\rchase{D}{\dep}$.  Denoting $\vett{Z} = \dd{Z}{k}$, we now define
  \(
  \mu_{m+1} = \mu_m \cup \{ \lambda'_O(Z_i) \ra \lambda'_R(Z_i) \}_{1 \leq i
    \leq k}.
  \)
  To complete the proof, we now need to show that $\mu_{m+1}$ is indeed a
  homomorphism.
  % that maps $pochase{D}{\dep}{m+1}$ to atoms of $rchase{D}{\dep}$.
  The addition of $\lambda'_O(Z_i) \ra \lambda'_R(Z_i)$, with $1 \leq i \leq
  k$, is compatible with $\mu_m$ because none of the $\lambda'_O(Z_i)$ appears
  in $\mu_m$. Therefore $\mu_{m+1}$ is a well-defined mapping.  Now,
  consider an
  atom $r(\vett{X}, \vett{Z})$ in $\Psi(\vett{X}, \vett{Z})$. Then the atom
  $\lambda'_O(r(\vett{X}, \vett{Z}))$ is added to
  $\ochase{D}{\dep}$ in the $(m+1)$-th step and
  % the (ground) atom that is added in $\pochase{D}{\dep}{m+1}$ is
  % $\lambda'_O(R(\dd{X}{n}, \dd{Z}{k}))$;
  \(
  \mu_{m+1}(r(\vett{X},\vett{Y})) = \mu_{m+1}(r(\lambda'_O(\vett{X}),
  \lambda'_O(\vett{Z}))) = r(\mu_{m+1}(\lambda'_O(\vett{X}),
  \mu_{m+1}(\lambda'_O(\vett{Z}))
  \).
  % \begin{array}{rl}
  %   \mu_{m+1}(r(\vett{X},\vett{Y})) =
  %   \mu_{m+1}(r(\lambda'_O(\vett{X}), \lambda'_O(\vett{Z}))) &= \\
  %   r(\mu_{m+1}(\lambda'_O(\vett{X}), \mu_{m+1}(\lambda'_O(\vett{Z}))
  % \end{array}
  %
  Notice that $\mu_{m+1}(\lambda'_O(\vett{X})) =
  \mu_{m+1}(\lambda_O(\vett{X})) = \lambda_R(\vett{X}) = \lambda'_R(\vett{X})$,
  and $\mu_{m+1}(\lambda'_O(\vett{Z})) = \lambda'_R(\vett{Z})$.  Therefore,
  \(
  \mu_{m+1}(r(\vett{X},\vett{Z})) = r(\lambda'_R(\vett{X}), \lambda'_R(\vett{Z}))
  = \lambda'_R(r(\vett{X}, \vett{Z})),
  \)
  which is in $\rchase{D}{\dep}$, by construction.  The desired homomorphism
  from $\ochase{D}{\dep}$ to $\rchase{D}{\dep}$ is therefore $\mu =
  \bigcup_{i=0}^{\infty} \mu_i$.
  % $\mu = \lim_{M \ra \infty} \bigcup_{i=0}^{M} \mu_i$.
\end{proof}

\begin{corollary}\label{cor:ochase-universal}
  Given a set $\dep$ of TGDs over a relational schema $\R$ and a database $D$ for
  $\R$, $\ochase{D}{\dep}$ is a universal solution for $D$ under $\dep$.
\end{corollary}

\begin{corollary}\label{cor:chase-answering}
  Given a Boolean query $Q$ over a schema $\R$, a database $D$ for $\R$, and a
  set of TGDs $\dep$, $\ochase{D}{\dep} \models Q$ if and only if
  $\rchase{D}{\dep} \models Q$.
\end{corollary}

In the following, unless explicitly stated otherwise,
``chase'' will mean the oblivious chase, and $\chase{D}{\dep}$ will stand for
$\ochase{D}{\dep}$.

%%%%%%%%%%%%%%%%%%%%%%%%%%%%%%%%%%%%%%%%%%%%%%%%%%%%%%%%%%%%%%%%%%%%%%%%%%%%%

\subsection{Decision Problems}

Recall that,
by Theorem~\ref{the:ans-by-chase}, $D \cup \dep \models
Q$ iff $\chase{D}{\dep} \models Q$.  Based on this, we define two relevant
decision problems and prove their \textsc{logspace}-equivalence.

\begin{definition}\label{def:cq-eval}\rm
  The \emph{conjunctive query evaluation decision problem} \problem{CQeval} is
  defined as follows.  Given a conjunctive query $Q$ with $n$-ary head
  predicate $q$, a set of TGDs $\dep$, a database $D$ and a ground $n$-tuple
  $\vett{t}$, decide whether $\vett{t} \in \ans{Q}{D}{\dep}$
  % $q(\vett{t}) \in Q(\chase{D}{\dep})$
  or, equivalently, whether $\chase{D}{\dep} \cup \set{Q} \models q(\vett{t})$.
\end{definition}

\begin{definition}\label{def:bcq-eval}\rm
  The \emph{Boolean conjunctive query evaluation problem} \problem{BCQeval} is
  defined as follows.  Given a Boolean conjunctive query $Q$, a set of TGDs
  $\dep$, and a database $D$, decide whether $\chase{D}{\dep} \models Q$.
\end{definition}

The following result is implicit in~\cite{ChMe77}.
\begin{lemma}\label{lem:bcq-cq}
  The problems \problem{CQeval} and \problem{BCQeval} are
  \textsc{logspace}-equivalent.
\end{lemma}
\begin{proof}
  Notice that \problem{BCQeval} can be trivially made into a special instance
  of \problem{CQeval}, e.g., by adding a propositional atom as head atom.  It
  thus suffices to show that \problem{CQeval} polynomially reduces to
  \problem{BCQeval}.  Let $\tup{Q,D,\dep,q(\vett{t})}$ be an instance of
  \problem{CQeval}, where $q/n$ is the head predicate of $Q$ and $\vett{t}$ is
  a ground $n$-tuple.  Assume the head atom of $Q$ is $q(X_1, \ldots, X_n)$ and
  $\vett{t} = \tup{c_1, \ldots, c_n}$.  Then define $Q'$ to be the Boolean
  conjunctive query whose body is $\body{Q}\wedge q'(X_1,
  \ldots, X_n)$, where $q'$ is a fresh predicate symbol not occurring in $D$,
  $Q$, or $\dep$
  % (and therefore not in $\dep$, since the TGDs in $\dep$ are on the same
  % relational schema).
  % Moreover, let $\atom{a}' = r_h^*(c_1, \ldots, c_n)$.
  It is easy to see that $q(\vett{t}) \in Q(\chase{D}{\dep})$ iff
  $\chase{D \cup \set{q'(\dd{c}{n})}}{\dep} \models Q'$.
\end{proof}

By the above lemma and by the well-known equivalence of the problem of query
containment under TGDs with the \problem{CQeval} problem
(Corollary~\ref{cor:answering-containment}), the three following problems are
\textsc{logspace}-equivalent:
\textit{(1)} CQ-eval under TGDs, \textit{(2)} BCQeval under TGDs, \textit{(3)}
query containment under TGDs.
Henceforth, we will consider only one of these problems, the
BCQ-eval problem.  By the above,
all complexity results carry over to the other problems.

\paragraph{Dealing with multiple head-atoms.}  It turns out that dealing with
multiple atoms in TGD heads complicates the proof techniques, so we
assume that all TGDs have a single atom in their head.  After
proving our results for single-headed TGDs, we will extend
these results to the case of multiple-atom heads in
Section~\ref{sec:multiple-atoms}. % (page~\pageref{sec:multiple-atoms}).

%%%%%%%%%%%%%%%%%%%%%%%%%%%%%%%%%%%%%%%%%%%%%%%%%%%%%%%%%%%%%%%%%%%%%%%%%%%%%

\subsection{Tree Decomposition and Related Notions}
\label{sec:tree-decomp}

We now introduce the required notions about tree decompositions.  A
\emph{hypergraph} is a pair $\H = \tup{V,H}$, where $V$ is the set of nodes and
$H \subseteq 2^V$.  The elements of $H$ are thus subsets of $V$; they are
called \emph{hyperedges}.
The \emph{Gaifman graph} of a hypergraph $\H = \tup{V,H}$, denoted by $\G_\H$,
is an undirected graph where $V$ is the set of nodes and an edge $(v_1, v_2)$
is in the graph if $v_1$ and $v_2$ jointly occur in some hyperedge in $H$.

Given a graph $\G = \tup{V,E}$, a \emph{tree decomposition} of $\G$ is a pair
$\tup{T,\lambda}$, where $T = \tup{N,A}$ is a tree, and $\lambda$ a labeling
function $\lambda: N \ra 2^V$ such that:
\begin{rplist} \itemsep-\parsep
\item for all $v \in V$ there is $n \in N$ such that $v \in \lambda(n)$;
  that is,
  $\lambda(N) = \bigcup_{n \in N} \lambda(n) = V$;
\item for every edge $e = (v_1,v_2) \in E$ there is $n \in N$
  such that $\lambda(n) \supseteq \{v_1,v_2\}$;
\item for every $v \in V$, the set $\{n \in N ~\mid~ v\in \lambda(n)\}$ induces
  a connected subtree in $T$.
\end{rplist}
The \emph{width} of a tree decomposition $\tup{T,\lambda}$ is the integer value
$\max \{|\lambda(n)|-1 \mid n \in N\}$.  The \emph{treewidth} of a graph $\G =
\tup{V,E}$, denoted by $\tw{\G}$, is the minimum width of all tree
decompositions of $G$.  Given a hypergraph $\H$, its treewidth $\tw{\H}$ is
defined as the treewidth of its Gaifman graph: $\tw{\H} = \tw{\G_\H}$.
%
% \marginpar{\tiny MK: What is a structure? A first-order structure?}
Notice that the notion of treewidth immediately extends to relational
structures.
% Therefore, since we can view database instances and queries as structures,
% the treewidth of databases and queries is defined.

% \marginpar{\tiny MK: What is the point in having $\T$? It seems that
% eliminating it and replacing with $\C$ has the same effect.}

% \marginpar{\tiny AC: This, together with Theorem~\ref{the:answering-btwmp}
% could be eliminated.}  A class $\C$ of first-order formulas enjoys the
% \emph{bounded-treewidth model property} if there is some $k$ such that, for
% every theory $\T \subseteq \C$, whenever $\T$ is satisfiable, then $\T$ has a
% model of treewidth at most $k$.

% The following result straightforwardly follows from~\cite{GoGr00,Grae99}.
% %
% \begin{theorem}\label{the:answering-btwmp}
%   If a class of first-order formulas $\C$ has the bounded-treewidth model
%   property, then for every theory $\T \subseteq \C$ checking satisfiability
%   for formulas in~$\T$ is decidable.
%   %%   it was "model checking"
% \end{theorem}

%%% Local Variables:
%%% mode: latex
%%% TeX-master: "main"
%%% End:

\section{Guarded and Weakly-Guarded TGDs: Decidability Issues}
\label{sec:decidability}

This section introduces \emph{guarded TGDs (GTGDs)} and \emph{weakly guarded
  sets of TGDs (WGTGDs)}, which enjoy several useful properties.  In
particular, we show that query answering under these TGDs is decidable.

\begin{definition}\label{def:gtgd}\rm
  Given a TGD $\sigma$ of the form $\Phi(\vett{X},\vett{Y}) \ra
  \Psi(\vett{X},\vett{Z})$, we say that $\sigma$ is \emph{a (fully) guarded TGD
    (GTGD)} if there exists an atom in the body, called a \emph{guard}, that
  contains all the universally quantified variables of $\sigma$, i.e., all the
  variables $\vett{X}, \vett{Y}$ that occur in $\Phi(\vett{X},\vett{Y})$.
  % , the body of $\sigma$.
\end{definition}

%%%%%%%%%%%%%%%%%%%%%%%%%%%%%%%%%%%%%%%%%%%%%%%%%%%%%%%%%%%%%%%%%%%%%%%%%%%%%

% the following was moved here from the prelimiaries section

To define \emph{weakly guarded} sets of TGDs, we first give the notion of an
\emph{affected} position in a predicate of a relational schema, given a set of
TGDs $\dep$.  Intuitively, a position $\pi$ is affected in a set of TGDs $\dep$
if there exists a database $D$ such that a labeled null appears in some atom of
$\chase{D}{\dep}$ at position $\pi$.  The importance of affected positions for
our definitions is that no labeled null can appear in non-affected positions.
We define this notion below.

\begin{definition}\label{def:affected}\rm
  Given a relational schema $\R$ and a set of TGDs $\dep$ over $\R$, a position
  $\pi$ of a predicate $p$ of $\R$ is \emph{affected} with respect to $\dep$ if
  either:
  \begin{itemize}\itemsep-\parsep
  \item \textit{(base case)} for some $\sigma \in \dep$, an existentially
    quantified variable appears in $\pi$ in $\head{\sigma}$, or
  \item \textit{(inductive case)} for some $\sigma \in \dep$, the variable
    appearing at position $\pi$ in $\head{\sigma}$ also appears in
    $\body{\sigma}$, and \emph{only} at affected positions.
    % (in the body occurrences of
    % $\head{\sigma}$ or of the head predicates of other TGDs from $\dep$).
  \end{itemize}
\end{definition}

\begin{example} \label{exa:affected} 
Consider the following set of TGDs:
\[
\begin{array}{crcl}
  \sigma_1: & p_1(X,Y), p_2(X,Y) &\ra& \exists Z\, p_2(Y,Z)\\
  \sigma_2: & p_2(X,Y), p_2(W,X) &\ra& p_1(Y,X)
\end{array}
\]
Notice that $p_2[2]$ is affected since $Z$ is existentially quantified in
$\sigma_1$.  The variable $Y$ in $\sigma_1$ appears in $p_2[2]$ (which is an
affected position) and also in $p_1[2]$ (which is not an affected position).
Therefore $Y$ in $\sigma_1$ does \emph{not} make the position $p_2[1]$ an
affected one.  Similarly, in $\sigma_2$, $X$ appears in the affected position
$p_2[2]$ and also in the non-affected position $p_2[1]$.  Therefore, $p_1[2]$
is not affected.  On the other hand, $Y$ in $\sigma_2$ appears in $p_2[2]$ and
nowhere else. Since we have already established that $p_2[2]$ is an affected
position, this makes $p_1[1]$ also an affected position.
\end{example}

%%%%%%%%%%%%%%%%%%%%%%%%%%%%%%%%%%%%%%%%%%%%%%%%%%%%%%%%%%%%%%%%%%%%%%%%%%%%%

\begin{definition}\label{def:wgtg}\rm
  Consider a set of TGDs $\dep$ on a schema $\R$.  
  A TGD $\sigma \in \dep$ of the form $\Phi(\vett{X},\vett{Y}) \ra \exists\
  \Psi(\vett{X},\vett{Z})$ is said to be \emph{weakly guarded with respect
    to~$\dep$} ($WGTGD$) if there is an atom in $\body{\sigma}$, called a \emph{weak
    guard}, that contains all the universally quantified variables of $\sigma$
  that appear in affected positions \wrt~$\dep$ and do not also appear in
  non-affected positions \wrt~$\dep$.  The set $\dep$ is said to be a
  \emph{weakly guarded set of TGDs} if each TGD $\sigma \in \dep$ is
  weakly guarded \wrt~$\dep$.
\end{definition}

A GTGD or WGTGD may have more than one guard.  In such a case, we will pick a
lexicographically first guard or use some other criterion for fixing \emph{the}
guard of a rule.  The actual choice will not affect our proofs and results.
% We shall sometimes use the acronym WGTGDs for a weakly guarded set of TGDs.

% It is important to notice that the transformation described in
% Lemma~\ref{lem:single-head-tgds} preserves the guardedness and weak
% guardedness properties.  Therefore, we can still assume that TGDs have
% singleton heads.

% \begin{theorem}
%   Given a relational schema $\R$, a set of TGDs $\dep$, a database instance
%   $D$ for $\R$, a query $Q$ over $\R$ and a ground atom $t$ of the same arity
%   as $Q$, the problem of determining whether $t \in \ans{Q}{D}{\dep}$ is
%   undecidable.
% \end{theorem}

The following theorem shows the undecidability of conjunctive query answering
under TGDs.  This result, in its general form, follows from undecidability
results for TGD implication (see~\cite{BeVa81,ChLM81}).  We show here that the
CQ answering problem remains undecidable even in case of a \emph{fixed} set
$\dep$ of single-headed TGDs with a \emph{single} non-guarded rule, and a
ground atom as query.  Our proof is ``from first principles'' as it reduces the
well-known halting problem for Turing machines to query-answering under TGDs.
More recently, it was shown in~\cite{BLMS11} that CQ answering is undecidable
also in case $\dep$ contains a single TGD, which, however, contains multiple
atoms in its head.
%
% The following theorem shows that it is essentially non-guarded rules that are
% responsible for the undecidability of the main problems treated in this paper.
% Even a \emph{single} unguarded rule can destroy the decidability of simplest
% reasoning tasks under TGDs.

\begin{theorem}\label{theo:grid} 
  There exists a fixed atomic BCQ $Q$ and a fixed set of TGDs $\dep_u$,
  where
  all TGDs in $\dep_u$ are guarded except one, such that it is undecidable to
  determine whether for a database $D$, $D \cup \dep_u \models Q$ or,
  equivalently, whether $\chase{D}{\dep_u} \models Q$.
\end{theorem}

\begin{proof}
  The proof hinges on the observation that, with appropriate input facts $D$,
  using a fixed set of TGDs that consists of guarded TGDs and a single
  unguarded TGD, it is possible to force an infinite grid to appear in
  $\chase{D}{\dep_u}$.  By a further set of guarded rules, one can then easily
  simulate a deterministic universal Turing machine (TM) $\M$, which executes
  every deterministic TM with an empty input tape, whose transition table is
  specified in the database $D$.  This is done by using the infinite grid,
  where the $i$-th horizontal line of the grid represents the tape content at
  instant $i$.
  We assume that transitions of the Turing machine $\M$ are encoded into a
  relation $\mathit{trans}$ of $D$, where for example, the ground atom
  $\mathit{trans}(s_1,a_1,s_2,a_2,\mathit{right})$ means \emph{``if the current
    state is $s_1$ and symbol $a_1$ is read, then switch to state $s_2$, write
    $a_2$, and move to the right''}.

  We show how the infinite grid is defined.  Let $D$ contain (among other
  initialization atoms that specify the initial configuration of $\M$) the atom
  $\mathit{index}(0)$, which defines the initial point of the grid.  Also, we
  make use of three constants $\mathit{right}, \mathit{left}, \mathit{stay}$
  for encoding the three types of moves.  Consider the following TGDs:
\[
\begin{array}{c}
  \mathit{index}(X) \ra \exists Y\ \mathit{next}(X,Y)\\[1.7mm]
  next(X,Y) \ra \mathit{index}(Y)\\[1.7mm]
  \mathit{trans}(\vett{T}),
  next(X_1,X_2), next(Y_1,Y_2) \ra
  \mathit{grid}(\vett{T},X_1,Y_1,X_2,Y_2)
\end{array}
\]
where $\vett{T}$ stands for the sequence of argument variables
$S_1,A_1,S_2,A_2,M$, as appropriate for the predicate $\mathit{trans}$.  Note
that only the last of these three TGDs is non-guarded.  The above TGDs define
an infinite grid whose points have co-ordinates $X$ and $Y$ (horizontal and
vertical, respectively) and where for each point its horizontal and vertical
successors are also encoded.  In addition, each point appears together with
each possible transition rule.
It is not hard to see that we can simulate the progress of our Turing machine
$\M$ using suitable initialization atoms in $D$ and guarded TGDs.  To this end,
we need additional predicates $\mathit{cursor}(Y,X)$, meaning that the cursor
is in position $X$ at time $Y$, $\mathit{state}(Y,S)$, expressing that $\M$ is
in state $S$ at time $Y$, and $\mathit{content}(X,Y,A)$, expressing that at
time $Y$, the content of position $X$ in the tape is $A$.  The following rule
encodes the behavior of $\M$ on all transition rules that move the cursor to
the right:
%
% \[
% \begin{array}{l}
\begin{multline*}
  \mathit{grid}(S_1,A_1,S_2,A_2,\mathit{right},X_1,Y_1,X_2,Y_2),\\
  \mathit{cursor}(Y_1,X_1),
  \mathit{state}(Y_1,S_1), \mathit{content}(X_1,Y_1,A_1) \ra\\
  \mathit{cursor}(Y_2,X_2), \mathit{content}(X_1,Y_2,A_2),
  \mathit{state}(Y_2,S_2), \mathit{mark}(Y_1,X_1)
\end{multline*}
% \end{array}
% \]
%
Such a rule has also obvious sibling rules for ``left'' and ``stay'' moves.
For the sake of brevity only, the above rule contains multiple atoms in the
head. This is not a problem, as such rules have no existentially quantified
variables in the head.  Therefore, each TGD with multiple head-atoms can be
replaced by an equivalent set of TGDs with single-atom heads and identical
bodies.
%
% We also need state transition rules, as well as inertia rules that just copy
% well contents in case the curser is somewhere else.
%

Notice that the $\mathit{mark}$ predicate in the head marks the tape cell that
is modified at instant $Y_1$.  We now need additional ``inertia'' rules, which
ensure that all other positions in the tape are not modified between $Y_1$ and
the following time instant $Y_2$.  To this end, we use two different markings:
$\mathit{keep}_f$ for the tape positions that follow the one marked with
$\mathit{mark}$, and $\mathit{keep}_p$ for the preceding tape positions.  In
this way, we are able, by making use of guarded rules only, to ensure that, at
every instant $Y_1$, every tape cell $X$, such that $\mathit{keep}_p(Y_1,X)$ or
$\mathit{keep}_f(Y_1,X)$ is true, keeps the same symbol at the instant $Y_2$
following $Y_1$.  The rules below then propagate the aforementioned markings
forward and backwards, respectively, starting from the marked tape positions.
\[
\begin{array}{l}
  \mathit{mark}(Y_1,X_1), \mathit{grid}(\vett{T},X_1,Y_1,X_2,Y_2) \ra
  \mathit{keep}_f(Y_1,X_2)\\
  \mathit{keep}_f(Y_1,X_1), \mathit{grid}(\vett{T},X_1,Y_1,X_2,Y_2) \ra
  \mathit{keep}_f(Y_1,X_2)\\
  \mathit{mark}(Y_1,X_2), \mathit{grid}(\vett{T},X_1,Y_1,X_2,Y_2) \ra
  \mathit{keep}_p(Y_1,X_1)\\
  \mathit{keep}_p(Y_1,X_2), \mathit{grid}(\vett{T},X_1,Y_1,X_2,Y_2) \ra
  \mathit{keep}_p(Y_1,X_1)
\end{array}
\]
We also have inertia rules for all $a\in \set{a_1, \ldots, a_\ell, \blank}$,
where $\set{a_1, \ldots, a_\ell, \blank}$ is the tape alphabet:
\[
\begin{array}{l}
\mathit{keep}_f(Y_1,X_1), \mathit{grid}(\vett{T},X_1,Y_1,X_2,Y_2),
\mathit{content}(X_1,Y_1,a) \ra \mathit{content}(X_1,Y_2,a)\\
\mathit{keep}_p(Y_1,X_1), \mathit{grid}(\vett{T},X_1,Y_1,X_2,Y_2),
\mathit{content}(X_1,Y_1,a) \ra \mathit{content}(X_1,Y_2,a)
\end{array}
\]
%
% Notice that $a$ is a generic symbol in the tape alphabet $\set{a_1, \ldots,
%   a_\ell, \blank}$; 
Notice that we use the constant $a$ instead of a variable in the above rules in
order to have the guardedness property.  We therefore need two rules as above
for every tape symbol, that is, $2\ell + 2$ inertia rules altogether.

Finally, we assume, without loss of generality, that our Turing machine $\M$
has a single halting state $s_0$ which is encoded by the atom
$\mathit{halt}(s_0)$ in $D$.  We then add a guarded rule \(
\mathit{state}(Y,S), \mathit{halt}(S) \ra \mathit{stop} \).  It is now clear
that the machine halts iff $\chase{D}{\dep_u} \models \mathit{stop}$, i.e., iff
$D \cup \dep_u \models \mathit{stop}$.  We have thus reduced the halting
problem to the problem of answering atomic queries over a database under
$\dep_u$.  The latter problem is therefore undecidable.
\end{proof}

%%%%%%%%%%%%%%%%%%%%%%%%%%%%%%%%%%%%%%%%%%%%%%%%%%%%%%%%%%%%%%%%%%%%%%%%%%%%%

\begin{definition}[Guarded chase forest, restricted GCF] \label{def:gcgraph} Given a set of
  WGTGDs $\dep$ and a database $D$, the \emph{guarded chase forest (GCF)} for $D$ and
  $\dep$, denoted $\gcf{D}{\dep}$, is constructed as follows.
  \begin{aplist} \itemsep-\parsep
  \item For each atom (fact) $\atom{d}$ in $D$, add a node labeled with
    $\atom{d}$.
    % \item Each node of $\gcf{D}{\dep}$ is labeled with an atom of
    %   $\chase{D}{\dep}$.
  \item For every node $v$ labeled with $\atom{a} \in \chase{D}{\dep}$ and for
    every atom $\atom{b}$ obtained from $\atom{a}$ (and possibly other
    atoms) by
    a one-step application of a TGD $\sigma \in \dep$, if $\atom{a}$ is the
    image
    of the guard of $\sigma$ then add one node $v'$ labeled with $\atom{b}$ and
    an arc going from $v$ to $v'$.
%%
    % The subtree of a node $v$ labeled with an atom $\atom{a}$ (also simply
    % called subtree of $\atom{a}$) in this forest, denoted
    % $\subtree{\atom{a}}$,
    % is the restriction of the forest to all descendants of $v$.
%%
%% Sketchy version (13/08/13)
%%
    % Let $B_i \stackrel{\sigma_i,h_i}{\longrightarrow} B_{i+1}$, with $i \geq
    % 0$
    % and $B_0 = D$, be the sequence corresponding to the chase of $D$ with
    % respect to $\dep$ (notice that it is possible to have $B_j = B_{j+1}$ for
    % some $j > 0$).  For each pair $\tup{\sigma_i,h_i}$, let $\atom{g}_i$ be
    % the
    % guard of $\sigma_i$, $\atom{a_i} = h_i(\atom{g}_i)$ and $\atom{b_i} =
    % h'_i(\head{\sigma_i})$, where $h'_i$ is the corresponding extension of
    % $h_i$ as in definition~\ref{tgd-chase-rule} (TGD~Chase Rule).  Starting
    % at
    % $i =0$ and continuing for all $i \geq 0$, at each application of
    % $\tup{\sigma_i,h_i}$, add to all nodes labeled with $\atom{a}_i$ a child
    % labeled with $b_i$.
\end{aplist}
Assuming the chase forest $\gcf{D}{\dep}$ is built inductively, following
precisely the strategy of a fixed deterministic chase procedure, the set of all
non-root nodes of the chase forest is totally ordered by a relation $\prec$
that reflects their order of generation. The \emph{restricted GCF} for $D$ and
$\dep$, denoted $\rgcf{D}{\dep}$, is obtained from $\gcf{D}{\dep}$ by
eliminating % from $\gcf{D}{\dep}$
each subtree rooted in a node $w$ whose label is a duplicate of an earlier
generated node.  Thus, if $v$ and $w$ are nodes labeled by the same atom, and
$v\prec w$, $w$ and all nodes of the subtree rooted in $w$ are eliminated from
$\gcf{D}{\dep}$ so as to obtain $\rgcf{D}{\dep}$.  Note that in
$\rgcf{D}{\dep}$ each label occurs only once, therefore we can identify the
nodes with their labels and say, for instance, ``the node $\atom{a}$'' instead
of ``the node $v$ labeled by $\atom{a}$''.
%%
%% Old version follows; it was not consistent with the JWS paper
%%
  % \begin{aplist}
  % \item Every atom $\atom{d} \in D$ is a node of the forest.
  % \item Let $\atom{a}, \atom{b} \in \chase{D}{\dep}$ be 
  %   such that:
  %   %% 
  %   \begin{itemize}
  %   \item $\atom{a}$ is a node in the forest;
  %   \item $\atom{b}$ was added to $\chase{D}{\dep}$ by an application of a
  %     one-step chase rule using some $\sigma \in \dep$ and a homomorphism $h$;
  %     and
  %   \item $\atom{a} = h(\atom{g})$, where $\atom{g}$ is a guard in $\sigma$.
  %   \end{itemize}
  %   %% 
  %   Then $\atom{b}$ is also a node in the forest and $\gcf{D}{\dep}$ has
  %   an edge from $\atom{a}$ to $\atom{b}$.
  % \end{aplist}
\end{definition}

\begin{example}\label{exa:chase-forest}
  Consider again Example~\ref{exa:chase} on page~\pageref{exa:chase}.  The
  corresponding (infinite) guarded chase forest is shown in
  Figure~\ref{fig:chase-forest}.
  % The numbers in small circles denote the derivation depth of each atom,
  % i.e., the step of the chase procedure in which each atom is introduced.
  Every edge from an $\atom{a}$-node to a $\atom{b}$-node is labeled with the
  TGD whose application causes the introduction of $\atom{b}$.  Notice that
  some atoms (e.g., $r_2(b)$ or $r_2(z_2)$) label more than one node in the
  forest.  The nodes belonging also to the restricted GCF are shaded in the
  figure.
\end{example}

\begin{figure}[tbh]
  \centering
  \resizebox{9cm}{!}{\input{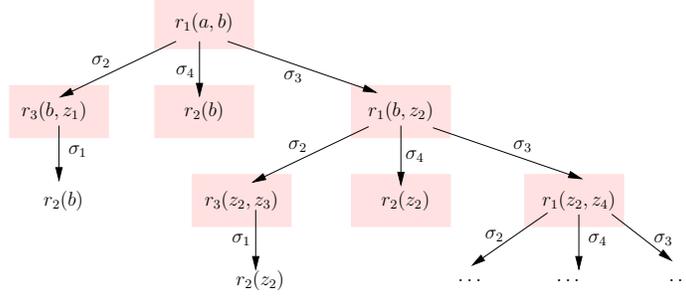}}
  \caption{Chase forest for Example~\ref{exa:chase-forest}.}
  \label{fig:chase-forest}
\end{figure}

%%%%%%%%%%%%%%%%%%%%%%%%%%%%%%%%%%%%%%%%%%%%%%%%%%%%%%%%%%%%%%%%%%%%%%%%%%%%%
%Insert Georg:

The goal of the following material is to show that, for weakly guarded sets
$\dep$ of TGDs, the possibly infinite set of atoms $\chase{D}{\dep}$ has finite
treewidth (Lemma~ \ref{lem:finite-treewidth}).  This will then be used to show
the decidability of query-answering under WGTGDs
(Theorem~\ref{the:wgtgds-decidable}).  As a first step towards proving that
$\chase{D}{\dep}$ has finite treewidth, we generalize the notion of acyclicity
of an instance, and then point out the relationship between this notion and
treewidth.  We will then show that $\chase{D}{\dep}$ enjoys (a specific version
of) our generalized form of acyclicity (Lemma~\ref{lem:d-acyclic}), from which
the finite treewidth result immediately follows.

\begin{definition}\label{def:s-acyclic}\rm
  Let $B$ be a (possibly infinite) % relational structure
  instance for a schema $\R$ and let $S \subseteq \adom{B}$.
  % We are mainly interested in sets $S$ such that $S \subseteq \adom{D}$, but
  % for the sake of generality, we do not exclude other sets here.
%
  \begin{itemize} \itemsep-\parsep
  \item An $[S]$-join forest $\tup{F,\mu}$ of $B$ is an undirected labeled
    forest $F = \tup{V,E}$ (finite or infinite), whose labeling function $\mu:
    V \ra B$ is such that:
    \begin{plist}\itemsep-\parsep
    \item $\mu$ is an epimorphism, i.e., $\mu(V) = B$;
    \item $F$ is $[S]$-connected, i.e., for each $c \in \adom{B} - S$, the set
      $\{v \in V ~\mid~ c \in \adom{\mu(v)}\}$ induces a connected subtree in
      $F$.
    \end{plist}
  \item We say that $B$ is \emph{$[S]$-acyclic} if $B$ has an $[S]$-join
    forest.
  \end{itemize}
\end{definition}

Notice that we are dealing with a relational instance, but the above definition
works for any relational structure, including queries.
Definition~\ref{def:s-acyclic} generalizes the classical notion of hypergraph
acyclicity~\cite{BFMM*81} of an instance or of a query: an instance or a query,
seen as a hypergraph, is hypergraph-acyclic (which is the same as
$\alpha$-acyclic according to~\cite{Fagi83}) if and only if it is
$[\emptyset]$-acyclic .

% \noteac{The above statement depends on the definition of hypergraph-acyclic.
% There is an $alpha$-acyclicity which is defined as the property of having a
% join-tree (or equivalently join forest), which for us is a $[\emptyset]$-join
% tree.}

The following Lemma follows from the definitions of $[S]$-acyclicity.

\begin{lemma}\label{lem:acyctreewidth}
  Given an instance $B$ for a schema $\R$, and a set $S \subseteq \adom{B}$, if
  $B$ is $[S]$-acyclic, then $\tw{B} \leq |S| + w$, where $w$ is the maximum
  predicate arity in $\R$ and $\tw{B}$ is the treewidth of $B$.
\end{lemma} 

\begin{proof}
  By hypothesis, $B$ is $[S]$-acyclic and therefore has an $[S]$-join forest
  $\tup{F,\mu}$, where $F = \tup{V,E}$.
  \begin{andrea}
    A tree decomposition $\tup{T,\lambda}$ with $T = \tup{N,A}$, is
    constructed as follows.  First, we take $N = V \cup \set{n_0}$, where $n_0$ is an
    auxiliary node. % that we use to merge the trees of the forest $T$.
    Let $V_r \subseteq V$ be the set of nodes that
    are roots in the $[S]$-join forest $F$ and let
    $A_r$ be the set of edges from $n_0$ to each node in $V_r$. We define  $A = E \cup A_r$.  The
    labeling function is defined as follows: $\lambda(n_0) = S$, and for all nodes $v\neq n_0$,
$\lambda(v)=  \adom{\mu(v)} \cup S$.
%   \[
%  \lambda(v) = \left\{
%  \begin{array}{ll}
%   S & \ \textrm{for} \ v = n_0\\
% \adom{\mu(v)} \cup S & \ \textrm{for} \ v \neq n_0
%\end{array}
%\right.
% \]
We now show that $\tup{T,\lambda}$ is a tree decomposition.  Recalling the
    definition of tree decompositions in
    Section~\ref{sec:tree-decomp}, \textit{(i)} holds
    trivially because $F$ is a join forest and $\mu(V) = B$.
    As for \textit{(ii)}, we notice that edges in the Gaifman graph of $B$
    % can be grouped in cliques, in particular
    are such that for each atom $\atom{d} = r(\dd{c}{m})$ in $B$ there is a
    clique among nodes $\dd{c}{m}$.  Since for the same atom there exists $v
    \in V$ such that $\mu(v) = \atom{d}$ and $\lambda(v) \supseteq
    \adom{\mu(v)}$, \textit{(ii)} holds immediately.  Finally we consider
    connectedness. Let us take a value $c$ appearing in $B$ as argument.  If $c
    \in S$, the set $\set{v \in N \mid c \in \lambda(v)}$ is the entire $N$, by
    construction, so connectedness holds. If $c \not\in S$, the set $\set{v \in
      N \mid c \in \lambda(v)}$ induces a connected subtree in $F$ and
    therefore in $T$, since $\lambda(v) = \mu(v) \cup S$.  Therefore,
    \textit{(iii)} holds.  Notice also that the width of such a tree
    decomposition is at most $|S|+w$ by construction.
  \end{andrea}
\end{proof}

%%%%%%%%%%%%%%%%%%%%%%%%%%%%%%%%%%%%%%%%%%%%%%%%%%%%%%%%%%%%%%%%%%%%%%%%%%%%%

\begin{definition}\label{def:dom}\rm
  Let $D$ be a database for a schema $\R$, and $\hb(D)$ be the Herbrand Base of
  $D$ as defined in Section~\ref{sec:preliminaries}.
  % The \emph{Herbrand Base} $\hb(D)$ of $D$ is the set of all atoms that can
  % be formed using the predicate symbols of $\R$ and arguments in $\adom{D}$.
  We define:
  \begin{itemize}\itemsep-\parsep
  \item $\chasehb{D}{\dep} = \chase{D}{\dep} \cap \hb(D)$, and
  \item $\chasefresh{D}{\dep}= \chase{D}{\dep} - \chasehb{D}{\dep}$
  \end{itemize}
\end{definition}
In plain words,
$\chasehb{D}{\dep}$ is the finite set of all null-free atoms in
$\chase{D}{\dep}$. In contrast,
$\chasefresh{D}{\dep}$ may be infinite; it is the set of
all atoms in $\chase{D}{\dep}$ that have at least one null as an argument.
Note that $\chasehb{D}{\dep} \cup \chasefresh{D}{\dep} = \chase{D}{\dep}$ and
$\chasehb{D}{\dep} \cap \chasefresh{D}{\dep} =\emptyset$.

\begin{lemma}\label{lem:d-acyclic}
  If $\dep$ is a weakly guarded set of TGDs and $D$ a database, then
  $\chasefresh{D}{\dep}$ is $[\adom{D}]$-acyclic, and so is therefore
$\chase{D}{\dep}$.
\end{lemma}

In order to prove this result, we resort to an auxiliary lemma.
\begin{lemma} \label{lem:wgtgds-chase-connectedness} Let $D$ be a database and
  $\dep$ a weakly guarded set of TGDs.  Let $\atom{a}_s$ be a node of
  $\rgcf{D}{\dep}$ where the null value $\zeta \in \freshdom$ is \emph{first}
  introduced, and let $\atom{a}_f$ be a descendant node of $\atom{a}_s$ in
  $\rgcf{D}{\dep}$ that has $\zeta$ as an argument.  Then, $\zeta$ appears in
  every node (=atom) on the (unique) path from $\atom{a}_s$ to $\atom{a}_f$.
  \end{lemma}
\begin{proof}
  Let $\atom{a}_1 = \atom{a}_s, \atom{a}_2, \ldots, \atom{a}_n =
  \atom{a}_f$ be the path from $\atom{a}_s$ to $\atom{a}_f$.  By the
  definition of affected positions, $\zeta$ appears \emph{only}
  in affected positions in the atoms in the chase.
  Suppose, to the contrary,
  that $\zeta$ does not
  appear in some intermediate atom in the above path.  Then, there is $i$, $2
  \leq i \leq n-1$, such that $\zeta$ does not appear in $\atom{a}_i$, but
  appears in $\atom{a}_{i+1}$.  Since $\zeta$ appears only in affected
  positions, in order to be in $\atom{a}_{i+1}$ it must either appear in
  $\atom{a}_i$ or to be invented when $\atom{a}_{i+1}$ was added.  The
  first case is ruled out by the assumption, and the second is impossible because
  $\zeta$ was first introduced in $\atom{a}_1$, not in $\atom{a}_{i+1}$---a
  contradiction. %%This concludes the proof.
\end{proof}
We now come back to the proof of Lemma~\ref{lem:d-acyclic}.
\begin{andrea}
  \begin{proof}
    The proof is constructive, by exhibiting a $[\adom{D}]$-join forest $F =
    \tup{V,E}$ for $\chase{D}{\dep}$.  We take $F$ as $\rgcf{D}{\dep}$ and
    define, for each atom $\atom{d} \in \rgcf{D}{\dep}$, the labeling function
    $\mu$ for $F$ as $\mu(\atom{d}) = \atom{d}$.  Since every atom of
    $\chase{D}{\dep}$ is ``covered'' by its corresponding node of $F$, it only
    remains to show that $\chase{D}{\dep}$ is $[\adom{D}]$-connected.  Take a
    pair of distinct atoms $\atom{a}_1, \atom{a}_2$ in $\rgcf{D}{\dep}$ that
    both have the same value $c \in \freshdom$ in an argument. The atoms
    $\atom{a}_1$ and $\atom{a}_2$ must have a common ancestor $\atom{a}$ in
    $\rgcf{D}{\dep}$ where $c$ was first invented: if they do not, the value $c$
    would have to be introduced twice in $\chase{D}{\dep}$.  By
    Lemma~\ref{lem:wgtgds-chase-connectedness}, $c$ appears in all atoms on the
    paths from $\atom{a}$ to $\atom{a}_1$ and from $\atom{a}$ to $\atom{a}_2$.
    It thus follows that the set $\set{v \in V ~\mid~ c \in \mu(v)}$ induces a
    connected subtree in $F$.
  \end{proof}
\end{andrea}

\begin{lemma} \label{lem:finite-treewidth} If $\dep$ is a weakly guarded set of
  TGDs and $D$ a database for a schema $\R$, then $\tw{\chase{D}{\dep}} \leq
  |\adom{D}| + w$,
  % was $\tw{\gcg{D}{\dep}} \leq |D| + w$, then
  % was $\tw{\chase{D}{\dep}} \leq |D| + w$
  where $w$ is the maximum predicate arity in $\R$.
\end{lemma}

\begin{proof}
  The claim follows from
  Lemmas~\ref{lem:acyctreewidth} and~\ref{lem:d-acyclic}.
\end{proof}  

%%%%%%%%%%%%%%%%%%%%%%%%%%%%%%%%%%%%%%%%%%%%%%%%%%%%%%%%%%%%%%%%%%%%%%%%%%%%%

\begin{theorem} \label{the:wgtgds-decidable} %% theorem 0
  Given a relational schema $\R$, a weakly guarded set of TGDs $\dep$, a
  Boolean conjunctive query $Q$, and a database $D$ for $\R$, the problem of
  checking whether $D \cup \dep \models Q$, or equivalently $\chase{D}{\dep}
  \models Q$, is decidable.
\end{theorem}

\begin{proofsk} We first remind a key result in~\cite{courcelle1990mso}, that
  generalizes an earlier result in~\cite{rabin1969dso}.  Courcelle's result
  states that the satisfiability problem is decidable for classes of
  first-order theories (more generally, theories of monadic second-order logic)
  that enjoy the finite treewidth model property.  A class $\C$ of theories has
  the finite-treewidth model property if for each satisfiable theory $\T \in
  \C$ it is possible to compute an integer $f(\T)$ such that $\T$ has a model
  of treewidth at most $f(\T)$---see also the works~\cite{GoGr00,Grae99}, where
  a more general property, called the generalized tree-model property, is
  discussed.  We apply this to prove our theorem.

  Let $\neg Q$ be the universal sentence obtained by negating the existentially
  quantified conjunctive query $Q$.  For all classes of TGDs, $D \cup
  \dep\models Q$ iff $\chase{D}{\dep} \models Q$ iff $D\cup \dep \cup \neg Q$
  is unsatisfiable.  Trivially, deciding whether $D \cup \dep \models Q$ is
  equivalent under Turing reductions to deciding whether $D \cup \dep \not \models Q$.
  The latter holds iff $D \cup \dep\cup \{\neg Q\}$ is satisfiable or,
  equivalently, iff $\chase{D}{\dep}$ is a model of $\neg Q$ which, in turn,
  holds iff $\chase{D}{\dep}$ is a model of $D \cup \dep \cup \{\neg Q\}$.  By
  Lemma~\ref{lem:finite-treewidth}, for WGTGDs, $\chase{D}{\dep}$ has finite
  treewidth.  Our decision problem thus amounts to checking whether a theory
  belonging to a class $\C^*$ of first-order theories (of the form $D\cup
  \dep\cup \{\neg Q\}$) is satisfiable, where it is guaranteed that whenever a
  theory in this class is satisfiable, then it has a model of finite treewidth
  (namely, $\chase{D}{\dep}$), and where $\C^*$ therefore enjoys the finite
  treewidth model property.  Decidability thus follows from Courcelle's result.
\end{proofsk}

\smallskip Determining the precise complexity of query answering under sets of
guarded and weakly guarded sets of TGDs will require new techniques, which are
the subject of the next sections.

%%% Local Variables: 
%%% mode: latex
%%% TeX-master: "main"
%%% End: 

\section{Complexity: Lower Bounds}
\label{sec:lower}

In this section we prove several lower bounds for the complexity of the
decision problem of answering Boolean conjunctive queries under
guarded and weakly guarded sets of TGDs.

% \subsection{EXPTIME Hardness}

\begin{theorem} \label{theo:hardness} %% theorem C
  \begin{andrea}
    The problem \problem{BCQeval} under WGTGDs is \textsc{exptime}-hard in case
    the TGDs are fixed.
    % maximum predicate arity $w$ is bounded.
    The same problem is \textsc{2exptime}-hard when the predicate arity is not
    bounded.  Both hardness results also hold for fixed atomic ground queries.
  \end{andrea}
  % The problem \problem{BCQeval} under weakly guarded sets of TGDs is
  % \textsc{2exptime}-hard.  The same problem is \textsc{exptime}-hard in case
  % the TGDs are fixed.  Both hardness results also hold for fixed atomic
  % ground queries.
%%
  % Given a relational schema $\R$, a set of WGTGDs $\dep$, a conjunctive query
  % $Q$, and a database instance for $\R$, the problem of determining whether
  % $\chase{D}{\dep} \models Q$ is \textsc{exptime}-hard.  In the case where
  % the
  % arity of predicates in $\R$ is not fixed, the same problem is
  % \textsc{2exptime} hard.
\end{theorem}

\begin{proof}
  We start with the \textsc{exptime}-hardness result for fixed WGTGD sets
  $\dep$.  It is well-known that \textsc{apspace} (alternating \textsc{pspace},
  see~\cite{ChKS81}) equals \textsc{exptime}.  Notice that there are
  \textsc{apspace}-hard languages that are accepted by alternating
  polynomial-space machines that use at most $n$ worktape cells, where $n$ is
  the input size, and where the input is initially placed on the
  worktape. (This is well-known and can be shown by trivial padding arguments).
  To prove our claim, it thus suffices to simulate the behavior of such a
  restricted linear space (\textsc{linspace}) Alternating Turing Machine (ATM)
  $\M$ on an input bit string $I$ by means of a weakly guarded set of TGDs
  $\dep$ and a database $D$.
  % Namely, we will exhibit a BCQ $Q$ and we will show that $\M$ accepts the
  % input %$I$ iff
  % $D \cup \dep \models Q$, or equivalently $\chase{D}{\dep} %\models Q$.
  Actually, we will show a stronger result: that a {\em fixed} set $\dep$ of
  WGTGDs can simulate a \emph{universal} ATM that in turn simulates every
  \textsc{linspace} ATM that uses at most $n$ tape cells on every input.  Here
  both the ATM transition table and the ATM input string will be stored in the
  database $D$. Then $D \cup \dep \models Q$ for some atomic ground query $Q$
  iff the ATM accepts the given input.\footnote{This technique was proposed
    in~\cite{CaGK08}.  It is similar to a technique later described
    in~\cite{HeLS11} in the proof of undecidability of the existence of
    so-called CWA (closed-world assumption) universal solutions in data
    exchange.} %~\cite{HeLS11}.}.

  Without loss of generality, we can assume that the ATM $\M$ has exactly one
  accepting state $s_a$.  We also assume that $\M$ never tries to read beyond
  its tape boundaries.  Let $\M$ be defined as
  \[
  \M = (S, \talph, \blank, \delta, s_0, \set{s_a})
  \]
  where $S$ is the set of states, $\talph = \{ 0, 1, \blank \}$ is the tape
  alphabet,
  $\blank$ is the blank tape symbol, $\delta$ is the
  transition function, defined as $\delta: S \times \talph \ra (S \times \talph
  \times \{\ell, r, \bot\})^2$ ($\bot$ denotes the ``stay'' head move, while
  $\ell$ and $r$ denote ``left'' and ``right'', respectively), $s_0 \in S$ is
  the initial state, and $\set{s_a}$ is the singleton-set of accepting states.  Since $\M$ is an
  alternating Turing machine (ATM), its set of states $S$ is \emph{partitioned}
  into two sets: $S_\forall$ and $S_\exists$ (universal and existential states,
  respectively).
  The general idea of the encoding is that configurations of $\M$ (except for the initial configuration $\init$)  will be
  represented by fresh nulls $v_i$, $i\geq 1$, that are generated by the chase.

  \ourpar{The relational schema}.  We now describe the predicates of the schema
  which we use in the reduction.  Notice that the schema is \emph{fixed} and
  does not depend on the particular ATM that we encode.  The schema predicates
  are as follows.
  \begin{plist} \itemsep-\parsep
  \item \textsl{Tape.}  The ternary predicate $\mathit{symbol}(a,c,v)$ denotes
    that in configuration $v$ the cell $c$ contains the symbol $a$, with $a
    \in \talph$.  Also, a binary predicate $\mathit{succ}(c_1,c_2)$ denotes
    the fact that cell $c_1$ \emph{follows} cell $c_2$ on the tape.  Finally,
    $\mathit{neq}(c_1,c_2)$ says that two cells are distinct.
  \item \textsl{States.}  A binary predicate $\mathit{state}(s,v)$ says that in
    configuration $v$ the ATM $\M$ is in state $s$.  We use three additional
    unary predicates: $\mathit{existential}$, $\mathit{universal}$, and
    $\mathit{accept}$.  The atom $\mathit{existential}(s)$ (resp., $
    \mathit{universal}(s)$) denotes that the state $s$ is existential (resp.,
    universal), while $\mathit{accept}(c)$ says that $c$ is an accepting
    configuration, that is, one whose state is the accepting state.
  \item \textsl{Configurations.}  A unary predicate $\mathit{config}(v)$
    expresses the fact that the value $v$ identifies a configuration.  A
    ternary predicate $\mathit{next}(v,v_1,v_2)$ is used to say that both
    configurations $v_1$ and $v_2$ are derived from $v$.  Similarly, we use
    $\mathit{follows}(v,v')$ to say that configuration $v'$ is derived from
    $v$.  Finally, a unary predicate $\mathit{init}(v)$ states that the
    configuration $v$ is initial.
  \item \textsl{Head (cursor).}  We use the fact $\mathit{cursor}(c,v)$ to say
    that the head (cursor) of the ATM is on cell $c$ in configuration $v$.
  \item \textsl{Marking.}  Similarly to what is done in the proof of
    Theorem~\ref{theo:grid}, we use $\mathit{mark}(c,v)$ to say that a cell $c$
    is marked in a configuration $v$.  Our TGDs will ensure that all non-marked
    cells keep their symbols in a transition from one configuration to another.
  \item \emph{Transition function.}  To represent the transition
    function $\delta$ of $\M$, we use a single 8-ary predicate
    $\textit{transition}$: for every transition rule $\delta(s,a) =
    ((s_1,a_1,m_1), (s_2,a_2,m_2))$ we will have $\mathit{transition}(s,a,
    s_1,a_1,m_1, s_2,a_2,m_2)$.
  \end{plist}

  \ourpar{The database $D$}.  The data constants of the database $D$ are
  used to identify cells, configurations, states and so on.  In particular, we
  will use an accepting state $s_a$ and an initial state $s_0$ plus a special
  initial configuration $\init$.  The database describes the initial
  configuration of the ATM with some technicalities.
  \begin{aplist} \itemsep-\parsep
  \item We assume, without loss of generality, the $n$ symbols of the input $I$
    to occupy the cells numbered from $1$ to $n$, i.e., $c_1, \ldots, c_n$.
    For technical reasons, in order to obtain a simpler TGD set below, we also
    use the dummy cell constants $c_0$ and $c_{n+1}$, that intuitively
    represent border cells without symbols.  For the $i$-th symbol $a_i$ of
    $I$, the database has the fact $\mathit{symbol}(a,c_i,\init)$, for all $i
    \in \set{1,\ldots,n}$.
    %With similar facts, we encode the fact that the other tape cells (beyond
    %which, recall, the ATM does not go by assumption) contain the $\blank$
    %(blank) symbol in the initial configuration $\init$.
    % (resp.~$1$), $\mathit{zero}(c_i, \init)$ (resp.~$\mathit{one}(c_i,
    % \init)$).
  \item An atom $\mathit{state}(s_0,\init)$ specifies that $\M$ is in state
    $s_0$ in its initial configuration $\init$.
  \item For every existential state $s_E$ and universal state $s_U$, we have
    the facts $\mathit{existential}(s_E)$ and $\mathit{universal}(s_U)$.  For
    the accepting state, the database has the fact $\mathit{accept}(s_a)$.
  \item An atom $\mathit{cursor}(c_1, \init)$ indicates that, in the initial
    configuration, the cursor points at the first cell.
  \item The atoms $\mathit{succ}(c_1, c_2), \ldots, \mathit{succ}(c_{n-1},
    c_n)$ encode the fact that the cells $c_1, \ldots, c_n$ of the tape (beyond
    which the ATM does not operate) are adjacent.  For technical reasons, we
    also use the analogous facts $\mathit{succ}(c_0, c_1)$ and
    $\mathit{succ}(c_n, c_{n+1})$.  
Also, atoms of the form
  $\mathit{neq}(c_i,c_j)$, for all $i,j$ such that $1 \leq i \leq n$, $1 \leq
   j \leq n$ and $i \neq j$, denote the fact that the cells $c_1, \ldots, c_{n}$
   are pairwise distinct.
  \item The atom $\mathit{config(\init)}$ says that $\init$ is a valid
    configuration.
  \item The database has atoms of the form
\(
\mathit{transition}(s,a,
    s_1,a_1,m_1, s_2,a_2,m_2),
\)
which encode the transition function $\delta$, as described above.
\end{aplist}

\ourpar{The TGDs}.  We now describe the TGDs that define the transitions and
the accepting configurations of the ATM.
\begin{aplist} \itemsep-\parsep
\item \textsl{Configuration generation.}  The following TGDs say that, for
  every configuration (halting or non halting---we do not mind having
  configurations that are derived from a halting one), there are two
  configurations that follow it, and that a configuration that follows another
  configurations is also a valid configuration:
    \begin{eqnarray*}
      \mathit{config}(V), % \mathit{state}(S,V) 
      & \ra & 
      \exists V_1\exists V_2\,\mathit{next}(V,V_1,V_2)\\
      \mathit{next}(V,V_1,V_2) & \ra & \mathit{config}(V_1), 
      \mathit{config}(V_2)\\
      \mathit{next}(V,V_1,V_2) & \ra & \mathit{follows}(V,V_1)\\
      \mathit{next}(V,V_1,V_2) & \ra & \mathit{follows}(V,V_2)
    \end{eqnarray*}
  \item \textsl{Configuration transition.}  The following TGD encodes the
    transition where the ATM starts at an existential state, moves right in
    its first configuration and left in the second.
    Here $C$ denotes the current cell, $C_1$ and $C_2$ are the new cells in the
    first and the second configuration (on the right and on the left of $C$,
    respectively), and the constants $r$, $\ell$, and $\bot$
    represent the ``right,'' the ``left,'' and the ``stay'' moves,
    respectively.
%
    % \[
    \begin{multline*}
      \mathit{transition}(S,A,S_1,A_1,r, S_2, A_2, \ell),
      % M_1=r, M_2=\ell,
      \mathit{next}(V,V_1,V_2),\\
      \mathit{state}(S,V), % \mathit{existential}(S)\\
      \mathit{cursor}(C,V), \mathit{symbol}(A,C,V),
      \mathit{succ}(C_1,C), \mathit{succ}(C,C_2) \ra \\
      \mathit{state}(S_1,V_1), \mathit{state}(S_2,V_2),
      \mathit{symbol}(A_1,C_1,V_1), \mathit{symbol}(A_2,C_2,V_2),\\
      \mathit{cursor}(C_1,V_1), \mathit{cursor}(C_2,V_2),
      % \mathit{next}(V,V_1,V_2)\\
      \mathit{mark}(C,V),% \mathit{mark}(C_2,V_2)
    \end{multline*}
    % \]
%
    There are nine rules like the above one, corresponding to all the
    possible moves of the head in the child configurations $C_1$ and $C_2$.
    These
    other moves are encoded via similar TGDs.  These rules suitably mark
    the cells that are written by the transition by means of the predicate
    $\mathit{mark}$.  The cells that are not involved in the transition must
    retain their symbols, which is specified by the following TGD:
 \begin{multline*}\hspace*{-2em}
   \mathit{config}(V), \mathit{follows}(V,V_1), \mathit{mark}(C,V),
   \mathit{symbol}(C_1,A,V), \mathit{neq}(C_1,C) \ra \mathit{symbol}(C_1,A,V_1)
  \end{multline*}
\item \textsl{Termination.} The rule \( \mathit{state}(s_a,V) \ra
  \mathit{accept}(V) \) defines a configuration $V$ to be accepting if its
  state is the accepting state.  The following TGDs state that, for an
  existential state, at least one configuration derived from it must be
  accepting. For universal states, both configurations must be accepting.
    \[
    \begin{array}{l}
      \mathit{next}(V,V_1,V_2), \mathit{state}(S,V), \mathit{existential}(S),
      \mathit{accept}(V_1) \ra  \mathit{accept}(V)\\[2mm]
      \mathit{next}(V,V_1,V_2), \mathit{state}(S,V), \mathit{existential}(S),
      \mathit{accept}(V_2) \ra  \mathit{accept}(V)\\[2mm]
      \mathit{next}(V,V_1,V_2), \mathit{state}(S,V), \mathit{universal}(S),
      \mathit{accept}(V_1),\mathit{accept}(V_2) \ra
      \mathit{accept}(V)
    \end{array}
    \]
  \end{aplist}

  Note that, for brevity, some of the TGDs we used have multiple atoms in the
  head.  However, these heads have no existentially quantified variables, so
  such multi-headed TGDs can be replaced with sets of TGDs that have only one
  head-atom.  Note also that the database constants ($r$, $\ell$, and $\bot$,
  and $s_a$) appearing in some rules can be eliminated by introducing
  additional predicate symbols and database atoms.  For example, if we add the
  predicate $\mathit{acceptstate}$ to the signature and the fact
  $\mathit{acceptstate}(s_a)$ to the database $D$, the rule \(
  \mathit{state}(s_a,V) \ra \mathit{accept}(V) \) can be replaced by the
  equivalent constant-free rule \( \mathit{acceptstate}(X), \mathit{state}(X,V)
  \ra \mathit{accept}(V) \).

  It is not hard to show that the encoding described above is sound and
  complete.  That is, $\M$ accepts the input $I$ if and only if
  $\chase{D}{\dep} \models \mathit{accept}(\init)$.  It is also easy to verify
  that the set of TGDs we have used is weakly guarded---this can be done by
  checking that each variable appearing only in affected positions also appears
  in a guard atom.  For instance, take the above rule
\(
%  \begin{multline}
\mathit{next}(V,V_1,V_2), \mathit{state}(S,V), \mathit{existential}(S),
\mathit{accept}(V_1) \ra \mathit{accept}(V).
%  \end{multline}
\)
It is immediate to see that $\state[1]$ and $\existential[1]$ are non-affected
(the TGDs never invent new states), and that all variables appearing in
affected positions only, namely $V,V_1,V_2$, appear in the guard atom
$\mathit{next}(V,V_1,V_2)$.  This proves the claim.

\smallskip

We now turn to the case where $\dep$ is not fixed and has unbounded predicate
arities.  For obtaining the \textsc{2exptime} lower bound, it is sufficient to
adapt the above proof so as to simulate an ATM having $2^n$ worktape cells,
i.e., an \textsc{aexpspace} machine whose space is restricted to $2^n$ tape
cells. Actually, to accommodate two dummy cells to the left and right of the
$2^n$ effective tape cells, that are used for technical reasons, we will
feature $2^{n+1}$ tape cells instead of just $2^n$.

We will make sure that the input string is put on cells $1,\ldots, n$ of the
worktape. Given that there are now many more than $n$ worktape cells, we will
fill all cells to the right of the input string with the blank symbol $\blank$.

This time, the WGTGD set $\dep$ will not be fixed, but will depend on $n$.
Since a much stronger result will be shown in Section~\ref{sec:guarded}
(Theorem~\ref{theo:icdt}), we do not belabor all the details in what follows,
but just explain how the above proof for fixed sets of TGDs needs to be
changed.

Rather than representing each tape cell by a data constant, each tape cell is
now represented by a vector $(b_0,b_1,b_2,\ldots,b_n)$ of Boolean values from
$\{0,1\}$.
The database $D$ is the same as before, except for the following changes:
\begin{itemize} \itemsep-\parsep
\item $D$ contains the additional facts $\bool(0)$, $\bool(1)$, $\zero(0)$,
  $\one(1)$.
\item Each fact $\mathit{symbol}(a,c_i,\kappa)$ is replaced by the fact
  $\mathit{symbol}(a,b_0,b_1,b_2,\ldots,b_n,\kappa)$, where
  $(b_0,b_1,b_2,\ldots,b_n)$ is the Boolean vector of length $n$ representing
  the integer $i$, with $0\leq i\leq n\leq 2^{n+1}$.
\item The fact $\mathit{cursor}(c_1, \init)$ is replaced by the $(n+2)$-ary
  fact $\mathit{cursor}(0,0,\cdots,0,1, \init)$.
\item All $\mathit succ$ and $\mathit neq$ facts described under item
  \textit{(e)} are eliminated.  (Vectorized versions of these predicates will
  be defined via Datalog rules---see below).
\end{itemize}

The TGD set from before is changed as follows.  In all rules, each
cell-variable $C$ is replaced by a vector $\vett{C}$ of $n$ variables.  For
example, the atom $\mathit{succ}(C_1,C)$ now becomes
$\mathit{succ}(\vett{C}_1,\vett{C}) = \mathit{succ}(C_1^0,C_1^1,\ldots C_1^n\,
,\,C^0,C^1,\ldots,C^n)$.

We add Datalog rules for $n$-ary $\mathit{succ}$ and $\mathit{neq}$ predicates.
For example, the $n$-ary predicate $\succp$ can be implemented by the following
rules:

    \[
    \begin{array}{rcl}
      \bool(X_0),\ldots, \bool(X_{n-1}) &\ra& 
      \succp(X_0,\ldots,X_{n-1},0\, {\bf ,}\, X_0, \ldots, X_{n-1},1),\\ 
      \bool(X_0),\ldots, \bool(X_{n-2}) &\ra& 
      \succp(X_0,\ldots,X_{n-2},0,1\, {\bf ,}\, X_0, \ldots, X_{n-2},1,0),\\
      &\vdots\\
      \bool(X_0),\ldots, \bool(X_{n-i}) &\ra& 
      \succp(X_0,\ldots,X_{n-i},0,1\ldots 1\, {\bf ,} \,X_0, \ldots, X_{n-i},1,0,\ldots,0),\\
      &\vdots\\
     &\ra& 
      \succp(0,1,\ldots, 1\, {\bf ,} \, 1,0,\ldots,0)
    \end{array}
    \]

    These rules contain constants which can be easily eliminated by use of the
    $\mathit{zero}$ and $\mathit{one}$ predicates, which are extensional
    database (EDB) predicates.  We further add simple Datalog rules that use
    the vectorized $\succp$ predicate to define vectorized versions for the
    $\mathit{less\_than}$ and the $\mathit{neq}$ predicates.  Using
    $\mathit{less\_than}$, we add a single rule that, for the initial
    configuration $\init$, puts blanks into all tape cells beyond the last cell
    $n$ of the input: $\mathit{less\_than}(\vett{n}, \vett{C}) \ra
    \mathit{symbol}(\blank,\vett{C},\init)$,
%%
%     \[
%     \begin{array}{rcl}
%      \mathit{less\_than}(\vett{n}, \vett{C})
%  &\ra& 
%    \mathit{symbol}(\blank,\vett{C},\init),
% \end{array}
%     \]
%
    where $\vett{n}$ is an $n$-ary binary vector representing the number $n$
    (i.e., the input size).

    The resulting set of rules is weakly guarded and correctly simulates the
    \textsc{aexpspace} (alternating exponential space) Turing machine whose
    transition table is stored in the database $D$.  Our reduction is
    polynomial in time.
    % \textsc{aspace}$(2^n)$-hard.\footnote{The notation
    % \textsc{aspace}$(f(n))$ denotes the class of decision problems solved by
    % an alternating Turing machine in space $f(n)$, where $n$ is of course the
    % input size.  An alternative notation for \textsc{aspace}$(2^n)$ is
    % therefore \textsc{aexpspace}} %
    Since \textsc{aexpspace}$ = $\textsc{2exptime}, it immediately follows that
    when the arity is not bounded the problem is \textsc{2exptime}-hard.
\end{proof}

%%%%%%%%%%%%%%%%%%%%%%%%%%%%%%%%%%%%%%%%%%%%%%%%%%%%%%%%%%%%%%%%%%%%%%%%%%%%%

%%% Local Variables: 
%%% mode: latex
%%% TeX-master: "main"
%%% End: 

\section{Complexity: Upper Bounds}
\label{sec:upper}

In this section we present upper bounds for query answering under weakly
guarded TGDs.

\subsection{Squid Decompositions}

We now define the notion of a \emph{squid decomposition}, and prove a lemma
called ``Squid Lemma'' which will be a useful tool for proving our complexity
results in the following sub-sections.

\begin{definition}\label{def:rcover}
  Let $Q$ be a Boolean conjunctive query with $n$ body atoms
  over a schema $\R$.
  An \emph{\rcover} of $Q$ is a Boolean conjunctive query
  $Q^+$ over $\R$ that contains in its body all the body
  atoms of $Q$. In
  addition, $Q^+$ may contain at most $n$ other $\R$-atoms.
\end{definition}

\begin{example}\label{exa:rcover}
  Let ${\R} = \{r/2, s/3, t/3\}$ and $Q$ be the Boolean conjunctive query with
  body atoms $\set{r(X,Y), r(Y,Z), t(Z,X,X)}$.  The following query $Q^+$ is an
  \rcover{} of $Q$: $Q^+ = \{ r(X,Y), r(Y,Z), t(Z,X,X), t(Y,Z,Z), s(Z,U,U) \}$.
\end{example}

\begin{lemma} \label{lem:cover-vs-query} Let $B$ be an % (finite or infinite)
  instance over a schema $\R$ and $Q$ a Boolean conjunctive query over $B$.
  Then $B \models Q$ iff there exists an \rcover{} $Q^+$ of $Q$ such that
  $B \models Q^+$.
\end{lemma}

\begin{proof}
  The \emph{only-if} direction follows trivially from the fact that $Q$ is an
  \rcover{} of itself.  The \emph{if} direction follows straightforwardly from
  the fact that whenever there is a homomorphism $h:\vars{Q^+} \ra \adom{B}$,
  such that $h(Q^+)\subseteq B$, then, given that $Q$ is a subset of $Q^+$, the
  restriction $h'$ of $h$ to $\vars{Q}$ 
  is a homomorphism $\vars{Q} \rightarrow \adom{B}$ such that
  $h'(Q) = h(Q) \subseteq B$.
    Therefore $B \models Q^+$ implies $B \models Q$.
\end{proof}

\begin{definition}\label{def:squid}
  Let $Q$ be a Boolean conjunctive query over a schema $\R$.  A \emph{squid
    decomposition} % (or, briefly, \emph{squidd})
  $\delta=(Q^+,h,H,T)$ of $Q$ consists of an \rcover{} $Q^+$ of $Q$, a mapping
  $h: \vars{Q^+} \ra \vars{Q^+}$, and a partition of $h(Q^+)$ into two sets $H$
  and $T$, with $T = h(Q^+) - H$, for which there exists a set of variables
  $V_\delta \subseteq h(\vars{Q^+})$ such that: \textsl{(i)} $H = \{ \atom{a}
  \in h(Q^+) \mid \vars{\atom{a}} \subseteq V_\delta \}$, and \textsl{(ii)} $T$
  is $[V_\delta]$-acyclic.
  If an appropriate set $V_\delta$ is given together with a squid decomposition
  $\delta=(Q^+,h,H,T)$, then, by a slight terminology overloading, we may just
  speak about \emph{the squid decomposition} $(Q^+,h,H,T,V_\delta)$.
%
%   The set of all squid decompositions of $Q$ is referred to as $\squidd{Q}$.
\end{definition}

% \noteac{Do we really need the symbol for $\squidd{Q}$?}

%
Note that a squid decomposition $\delta=(Q^+,h,H,T)$ of $Q$ does not
necessarily define a query folding~\cite{ChMe77,Qian96} of $Q^+$, because $h$
does not need to be an endomorphism of $Q^+$; in other terms, we do not require
that $h(Q^+) \subseteq Q^+$.  However, $h$ is trivially a homomorphism from
$Q^+$ to $h(Q^+)$.

Intuitively, a squid decomposition $\delta=(Q^+,h,H,T,V_\delta)$ describes a
way how a query $Q$ may be mapped homomorphically to $\chase{D}{\dep}$.  First,
instead of mapping $Q$ to $\chase{D}{\dep}$, we equivalently map $h(Q^+)=H\cup
T$ to $\chase{D}{\dep}$.  The set $V_\delta$ specifies those variables of
$h(Q^+)$ that ought to be mapped to constants, i.e., to elements of
$\adom{D}$. The atoms set $H$ is thus mapped to ground atoms, that is, elements
of the finite set $\chasehb{D}{\dep}$, which may be highly cyclic. The
[$V_\delta$]-acyclic atom set $T$ shall be mapped to the possibly infinite set
$\chasefresh{D}{\dep}$ which, however, is [$\adom{D}$]-acyclic. The
``acyclicities'' of $\chasefresh{D}{\dep}$ and of $T$ will be exploited for
designing appropriate decision procedures for determining whether
$\chase{D}{\dep}\models Q$. All this will be made formal in the sequel.

One can think of the set $H$ in a squid decomposition $\delta=(Q^+,h,H,T,V_\delta)$ as the \emph{head} of a
squid, and the set $T$ as a join-forest of \emph{tentacles} attached to that head. This will become clear in the following example and the associated 
Figure~\ref{fig:squid}.

%When representing a squid decomposition 
%$\delta=(Q^+,h,H,T)$ graphically,
%we may depict the {\em join graph}\footnote{
    %%
%    The join graph has the query atoms as nodes. An edge
 %   between two atoms exists iff the atoms share at least one variable.
  %%
 % } of its head $H$ and encircle it, and 

\begin{example}\label{exa:squid}
  Consider the following Boolean conjunctive query:
  \[
  \begin{array}{l}
    Q = \{ r(X,Y), r(X,Z), r(Y,Z),\\
    r(Z,V_1), r(V_1,V_2), r(V_2,V_3), r(V_3,V_4), r(V_4,V_5),\\
    r(V_1,V_6), r(V_6,V_5), r(V_5,V_7),
    r(Z,U_1), s(U_1,U_2,U_3),\\ 
    r(U_3,U_4), r(U_3,U_5), r(U_4,U_5) \}.
  \end{array}
  \]
  Let $Q^+$ be the following Boolean query:
  $Q^+ = Q \cup \set{s(U_3,U_4,U_5)}$.
  A possible squid decomposition $(Q^+,h,H,T,V_\delta)$ can be based on the homomorphism
  $h$, defined as follows: $h(V_6)=V_2$, $h(V_4)=h(V_5)=h(V_7)=V_3$, and
  $h(X)=X$ for any other variable $X$ in $Q^+$.  The result of the squid
  decomposition with $V_\delta=\{X,Y,Z\}$ is the query shown in  Figure~\ref{fig:squid}. 
%ggnew:
  Here the cyclic head $H$ (encircled in the oval) is represented by its {\em
    join graph},\footnote{The join graph of $H$ has the atoms as nodes. An edge
    between two atoms exists iff the atoms share at least one variable.} and
  the [$V_\delta$]-acyclic ``tentacle" set $T$ is depicted as a
  [$V_\delta$]-join forest. Moreover, the forest representing $T$ is rooted in
  the ``bag'' of $H$-atoms, so that the entire decomposition takes on the shape
  of a squid.  Note that if we eliminated the additional atom $s(U_3,U_4,U_5)$,
  the original set of atoms $\set{r(U_3,U_4), r(U_3,U_5), r(U_4,U_5)}$ would
  form a non-$[V_\delta]$-acyclic cycle, and therefore they would not all be
  part of the tentacles.
\end{example}

\begin{figure}[tbh]
  \centering 
  \input{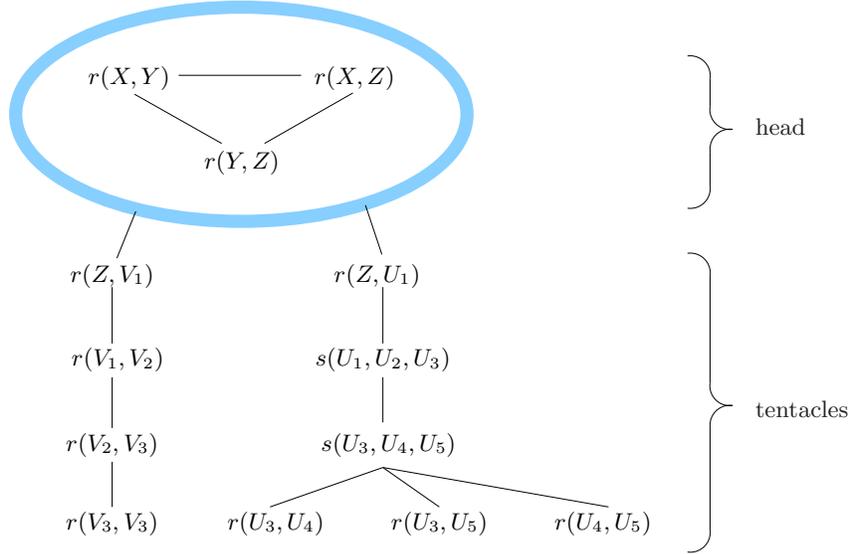}
  \caption{Squid decomposition from Example~\ref{exa:squid}.  Atoms in $h(Q^+)$
    are shown.}
  \label{fig:squid}
\end{figure}

The following two lemmas are auxiliary technical results.

\begin{lemma}\label{lem:nice}
  Let $Q$ be a Boolean conjunctive query
  and let $U$ be a (possibly infinite) $[A]$-acyclic instance, where $A
  \subseteq \adom{U}$.  Assume $U \models Q$, i.e., there is a homomorphism $f:
  \adom{Q} \ra \adom{U}$ with $f(Q) \subseteq U$.  Then:
  \begin{plist}\itemsep-\parsep
  \item There is an $[A]$-acyclic subset $W \subseteq U$ such that:
    \textsl{(i)} $f(Q)\subseteq W$ and \textsl{(ii)} $|W| < 2|Q|$.
  \item There is a cover
    $Q^+$ of $Q$ such that $|Q^+| < 2 |Q|$, and there is a
    homomorphism $g$ that extends $f$ and $g(Q^+) = W$.
  \end{plist}
\end{lemma}

\begin{proof}\mbox{}

  \textit{Part~(1).}  By the assumption,\footnote{
    One may be tempted to conjecture that $W
    = f(Q)$, but this does not work because acyclicity (and thus also
    $[A]$-acyclicity) is not a hereditary property:  it may well be the case
    that $U$ is acyclic, while the subset $f(Q) \subseteq U$ is not.  
    However, taking $W = f(Q)$ works in case of arities at most~$2$.
  }
  $U$ is $[A]$-acyclic and $f: \adom{Q} \ra \adom{U}$ is a homomorphism such
  that $f(Q) \subseteq U$.  Since $U$ is $[A]$-acyclic, it has a (possibly
  infinite) $[A]$-join forest $T = \tup{\tup{V,E},\lambda}$.  We assume,
  without loss of generality, that distinct vertices $u,v$ of $T$ have
  different labels, i.e., $\lambda(u) \neq \lambda(v)$.  This assumption can be
  made by removing all subforests rooted at nodes labeled by duplicate atoms.
  Let $T_Q$
  % was; $U_Q$
  be the finite subforest of
  % was: $U$
  $T$ that contains all ancestors in $T$ of nodes $s$ such that $\lambda(s)\in
  f(Q)$.
  % roots of those trees of $T$ that contain at least one vertex labeled by
  % some element of $f(Q)$, and such that $U_Q$ contains, in addition, all
  % simple paths from each such root $r$ to each descendant $s$ of $r$ such
  % that $\lambda(s)\in f(Q)$.
  Let $F = \tup{\tup{V',E'},\lambda'}$ be the forest obtained from $T$ as
  follows.
  \begin{itemize} \itemsep-\parsep
  \item $V' = \{ v \in V \mid \lambda(v)\in f(Q)\} \cup K$, where $K$ is the
    set of all vertices of
    % was: $U_Q$
    $T_Q$ that have at least two children.
  \item If $v,w \in V'$ then there is an edge from $v$ to $w$ in $E'$ iff $w$
    is a descendant of $v$ in $T$, and if the unique shortest path from $v$ to
    $w$ in $T$ does not contain any other node from $V'$.
  \item Finally, for each $v\in V'$, $\lambda'(v)=\lambda(v)$.
  \end{itemize}
  Let us define $W = \lambda(V')$.  We claim that the forest $F$ is an
  $[A]$-join forest of $W$.  Since Condition~\textit{(1)} of
  Definition~\ref{def:s-acyclic} ($[S]$-join forest) is immediately satisfied,
  it suffices to show Condition~\textit{(2)}, that is, that $F$ satisfies the
  $[A]$-connectedness condition.  Assume, for any pair of distinct vertices $v_1$
  and $v_2$ in $F$, that for some value $b \in \adom{U}-A$ it holds $b\in
  \adom{\lambda'(v_1)} \cap \adom{\lambda'(v_2)}$.
  \begin{andrea}
    In order to prove the aforementioned $[A]$-connectedness condition, we need
    to show that there is at least one path in $F$ between $v_1$ and $v_2$
    (here we view $F$ as a undirected graph), and that for every node $v \in
    V'$ lying on each such path we have $b \in \adom{\lambda'(v)}$.  By
    construction of $F$, $v_1$ and $v_2$ are connected in $T$ and $v$ lies on
    the (unique) path between $v_1$ and $v_2$ in $T$. Since $T$ is an
    $[A]$-join forest, we have $b\in
    \adom{\lambda(v)}=\adom{\lambda'(v)}$. Thus $F$ is an $[A]$-join forest of
    $W$.
  \end{andrea}

  % Then, by construction of $F$, $u$ and $w$ are connected by a path in $F$.
  % Assume now that some vertex $v$ lies on the unique shortest path between
  % $u$ and $w$ in $F$.  Then, by construction of $F$, $v$ lies on the unique
  % shortest path between $u$ and $w$ in $T$, and since $T$ is an $[A]$-join
  % forest, $b\in \adom{\lambda(v)}=\adom{\lambda'(v)}$.  Thus $F$ is an
  % $[A]$-join forest of $W$.
%
  Moreover, by construction of $F$, the number of children
  % was: out-degree
  of each inner vertex of $F$ is at least~$2$, and $F$ has at most $|Q|$
  leaves.  It follows that $F$ has at most $2|Q| - 1$ vertices.  Therefore $W$
  is an $[A]$-acyclic set of atoms such that $|W| \leq 2|Q|$ and $W \supseteq
  f(Q)$.
  % of cardinality $< 2|Q|$ containing $f(Q)$ as a subset.

  \textit{Part~(2).} $Q$ can be extended to $Q^+$ as follows.  For each atom
  $r(t_1,\ldots,t_k)$ in $W - f(Q)$, add to $Q$ a new query atom
  $r(\xi_1,\ldots,\xi_k)$ such that for each $1\leq i\leq k$, $\xi_i$ is a
  newly invented variable.
  Obviously, $W \models Q^+$ and thus there is a homomorphism $g$ extending
  $f$ such that $g(Q^+) = W$.  Moreover, by construction $|Q^+| < 2|Q|$.
\end{proof}

\begin{lemma}~\label{lem:nizza} Let $G$ be an $[A]$-acyclic instance and let
  $G'$ be an instance obtained from $G$ by eliminating a set $S$ of atoms where
  $\adom{S}\subseteq A$.  Then $G'$ is $[A]$-acyclic.
\end{lemma}

\begin{proof}
  If $T = \tup{\tup{V,E},\lambda}$ is an $[A]$-join forest for $G$ then an
  $[A]$-join forest $T'$ for $G'$ can be obtained from $G$ by repeatedly
  eliminating each vertex $v$ from $T$ where $\lambda(v)\in S$.
  \begin{andrea}
    By construction, each atom $\atom{e}$ eliminated from $G$ in this way has
    the property that $\adom{\atom{e}} \subseteq A$.  Hence, for every value $b
    \in \adom{G} - A$, the node $u \in V$ such that $\lambda(u) = \atom{e}$
    cannot belong to the induced (connected) subtree $\set{v \in V ~\mid~ b \in
      \adom{\lambda(v)}}$.  We thus get that $G'$ enjoys the
    $[A]$-connectedness property.
  \end{andrea}
\end{proof}

The following Lemma will be used as a main tool in the subsequent complexity
analysis.

  \begin{lemma}[Squid Lemma]\label{lem:squid}
    Let $\dep$ be a weakly guarded set of TGDs on a schema $\R$, $D$ a database
    for $\R$, and $Q$ a Boolean conjunctive query. Then $\chase{D}{\dep}
    \models Q$ iff there is a squid decomposition $\delta=(Q^+,h,H,T)$
    % \in \squidd{Q}$, such that $|Q^+| < 2|Q|$,
    and a homomorphism $\theta: \adom{h(Q^+)} \ra \adom{\chase{D}{\dep}}$ such
    that: \textit{(i)} $\theta(H)\subseteq \chasehb{D}{\dep}$, and
    \textit{(ii)} $\theta(T)\subseteq \chasefresh{D}{\dep}$.
  \end{lemma}

  \begin{proof}

    \begin{andrea}
      \ifdirection If there is a squid decomposition $\delta=(Q^+,h,H,T)$ of
      $Q$ and a homomorphism $\theta$ as described, then the
      composition $\theta \circ h$ is a homomorphism such that $(\theta \circ
      h)(Q^+) = \theta(h(Q^+)) \subseteq \chase{D}{\dep}$.  Hence,
      $\chase{D}{\dep} \models Q^+$ and, by Lemma~\ref{lem:cover-vs-query},
      $\chase{D}{\dep} \models Q$.
    \end{andrea}

    \onlyifdirection Assume $U = \chase{D}{\dep} \models Q$.  Then, there
    exists a homomorphism $f: \vars{Q} \ra \adom{U}$ with $f(Q) \subseteq
    \chase{D}{\dep}$.  By Lemma~\ref{lem:d-acyclic}, $\chasefresh{D}{\dep}$ is
    $[\adom{D}]$-acyclic. By Lemma~\ref{lem:nice}, it then follows that there
    is a Boolean query $Q^+$ with $< 2|Q|$ atoms, such that all atoms of
    $Q$ are contained in $Q^+$, and there is a homomorphism $g: \adom{Q^+}
    \rightarrow \adom{U}$ with $g(Q^+)\subseteq U$, such that $g(Q^+)$ is
    $[\adom{D}]$-acyclic.
 
    Partition $\vars{Q^+}$ into two sets $\varshb{Q^+}$ and $\varsfresh{Q^+}$
    as follows:
    \begin{itemize} \itemsep-\parsep
    \item $\varshb{Q^+} = \{X \in \vars{Q^+} ~\mid~ g(X) \in
      \adom{D}\}$
    \item $\varsfresh{Q^+} = \vars{Q^+} - \varshb{Q^+}$.
    \end{itemize}
    Define a mapping $h: \vars{Q^+} \ra \vars{Q^+}$ as follows.  For each $X
    \in \vars{Q^+}$, let $h(X)$ be the lexicographically first variable in the
    set $\{ Y \in \vars{Q^+} \mid g(Y) = g(X) \}$.
    Let us define $V_\delta$ as $V_\delta = h(\varshb{Q^+})$.
    Moreover, let $H$ be the set of all those atoms $\atom{a}$ of $h(Q^+)$, such
    that $\vars{\atom{a}} \subseteq V_\delta = h(\varshb{Q^+})$, and let
    $T = h(Q^+)-H$.  Note that, by definition of $H$, $g(H)\subseteq
    \chasehb{D}{\dep}$ and, by definition of $T$, $g(T) \subseteq
    \chasefresh{D}{\dep}$.  Let $\theta$ be the restriction of $g$ to
    $\adom{h(Q^+)}$.  Clearly, $\theta$, $h$, $H$, and $T$ fulfill the
    conditions \textit{(i)} and \textit{(ii)} of the statement of this lemma.
    % was Lemma~\ref{lem:squid}.
    It thus remains to prove that $\delta=(Q^+,h,H,T)$ is actually a squid
    decomposition of $Q$.  For this, we only need to show that $T$ is
    $[V_\delta]$-acyclic.
    \begin{andrea}
      \begin{andrea}
        To prove this, first observe that for each pair of variables $X,Y$ in
        $\vars{Q^+}$ such that $g(X) = g(Y)$ we have $h(X) = h(Y)$.
        Therefore
        $\theta$ is, by construction, a bijection between $h(\adom{Q^+})$ and
        $\adom{\theta(Q^+)}$.
      \end{andrea}
      In particular, $T \subseteq h(Q^+)$ is isomorphic to $\theta(T)$ via the
      restriction $\theta_T$ of $\theta$ to $\adom{T}$.  Since
      $\theta_T(T)=\theta(T)$ is obtained from the $[\adom{D}]$-acyclic
      instance $\theta(Q^+)$ by eliminating only atoms all of whose arguments
      are in $\adom{D}$ (namely the atoms in $\theta(H)$), by
      Lemma~\ref{lem:nizza}, $\theta_T(T)$ is itself $[\adom{D}]$-acyclic
      and,
      therefore, trivially also $[\adom{D}\cap \adom{\theta(T)}]$-acyclic.
      Now, since for every $X \in \adom{T}$ it holds that $X \in V_\delta$ iff
      $\theta_T(X) \in \adom{D}$, it immediately follows that, since
      $\theta_T(T)$ is $[\adom{D}]$-acyclic, $T$ is $[V_\delta]$-acyclic.
    \end{andrea}
 %
    % To see this, observe that $\theta$ is, by construction, a bijection
    % between
    % $h(\adom{Q^+})$ and $\adom{\theta(Q^+)}$, such that $h(Q^+)$ and
    % $\theta(Q^+)$ are isomorphic. Therefore, in particular, $T\subseteq
    % h(Q^+)$
    % is isomorphic to $\theta(T)$ via the restriction $\theta_T$ of $\theta$
    % to
    % $\adom{T}$.  Note that since $\theta_T(T)=\theta(T)$ is obtained from the
    % $[\adom{D}]$-acyclic instance $\theta(Q^+)$ by eliminating only atoms all
    % of whose arguments are in $\adom{D}$ (namely the atoms in $\theta(H)$),
    % by
    % Lemma~\ref{lem:nizza}, $\theta_T(T)$ is itself $[\adom{D}]$-acyclic, and
    % therefore trivially also $[\adom{D}\cap \adom{\theta_T(T)}]$-acyclic.
    % $\theta_T$ is thus by construction a bijection $\adom{T} \leftrightarrow
    % \theta_T(\adom{T})=\adom{\theta_T(T)}$ such that $\theta_T(V_{\delta}\cap
    % \adom{T})= \adom{D}\cap \adom{\theta_T(T)}$. Hence the two pairs
    % $(T,V_{\delta}\cap \adom{T})$ and $(\theta_T(T), \adom{D}\cap
    % \adom{\theta_T(T)}$ are isomorphic, and therefore, given that
    % $\theta_T(T)$
    % is $[\adom{D}\cap \adom{\theta_T(T)}]$-acyclic, $T$ is $[V_{\delta}\cap
    % \adom{T}]$-acyclic, and thus also $[V_{\delta}]$-acyclic (in fact, if an
    % instance is $[A]$-acyclic, and if $A\subseteq B$, then it is also
    % $[B]$-acyclic).
  \end{proof}

%%%%%%%%%%%%%%%%%%%%%%%%%%%%%%%%%%%%%%%%%%%%%%%%%%%%%%%%%%%%%%%%%%%%%%%%%%%%%%

  \subsection{Clouds and the Complexity of Query Answering under WGTGDs}
  \label{sec:clouds}

The goal of this subsection is to prove the following theorem:

\begin{theorem}\label{theo:upperbound}
  Let $\dep$ be a weakly guarded set of TGDs, $D$ a database for a schema $\R$,
  and $Q$ a Boolean conjunctive query.  The problem of determining whether $D
  \cup \dep \models Q$ or, equivalently, whether $\chase{D}{\dep} \models Q$ is
  in \textsc{exptime} in case of bounded arity, and in \textsc{2exptime} in
  general.
\end{theorem}

For the general case (of unbounded arities), we first outline a short
high-level proof aimed at specialists in Computational Logic.  This proof makes
sophisticated use of previous results.  We will then give a much longer,
self-contained proof, that works for both the general case and the case of
bounded arities.  The self-contained proof also introduces some concepts that
will be used in the following sections.

\bigskip

\noindent \textit{High Level Proof Sketch of Theorem~\ref{theo:upperbound}
  (General Case).}  We transform the original problem instance $(D,\dep,Q)$
into a guarded first-order theory $\Gamma=\tau(D,\dep,Q)$ such that
$\chase{D}{\dep} \models Q$ iff $\Gamma$ is unsatisfiable.
The signature $\sigma$ of $\Gamma$ uses $\Sigma$ as the set of constants plus a
constant for each element of $\adom{D}$. Moreover, $\sigma$ includes all
predicate symbols occurring in $D$, $\dep$, or $Q$, plus a special nullary
(i.e., propositional) predicate symbol $q$.
% plus at most $|Q|$ auxiliary predicate symbols, whose role will be explained
% later on.

$\Gamma$ contains all ground facts of $D$, plus all instances of each rule
$r\in\dep$ obtained from $r$ by replacing all variables of $r$ that occur in
non-affected positions with constants. Note that the resulting rules are
universally quantified guarded sentences.  Moreover, for each squid
decomposition $\delta=(Q^+,h,H,T,V_\delta)$, and each possible replacement
$\theta$ of the set of variables $V_\delta$ by constants of the signature
$\sigma$, $\Gamma$ contains a guarded sentence $\phi_\delta^\theta$ obtained as
follows. Notice that $Q_\delta^\theta:=\theta(H)\cup T$ is a Boolean acyclic
conjunctive query.  By the results in~\cite{GoLS03},\footnote{See Theorem~3 in
  that paper, its proof, the remark after that proof, and Corollary~3.}
$Q_\delta^\theta$ can thus be rewritten (in polynomial time) into an equivalent
guarded sentence $\psi_\delta^\theta$.  We define $\phi_\delta^\theta$ to be
$(\psi_\delta^\theta\rightarrow q)$, which is obviously guarded, too.  Let
$\Gamma^-$ denote the sentences of $\Gamma$ mentioned so far.  From this
construction and the Squid Lemma (Lemma~\ref{lem:squid}), it follows that
$\chase{D}{\dep} \models Q$ iff $\Gamma^-\models q$. Now let
$\Gamma=\Gamma^-\cup \{\neg q\}$. Obviously, $\Gamma$ is unsatisfiable iff
$\chase{D}{\dep} \models Q$.

Note that the reduction $\tau$ is an arity-preserving
\textsc{exptime}-reduction. Let $t$ be an exponential upper bound on the
runtime required by reduction $\tau$ (and thus also on the size of
$\tau(D,\dep,Q)$).  A deterministic version of the algorithm in~\cite{Grae99}
for deciding whether a guarded theory of unbounded arity is satisfiable or
unsatisfiable runs in double-exponential time $O(2^{O(s\cdot w^w)})$, where $s$
is the size of the theory and $w$ is its maximum predicate arity.  Therefore,
the overall runtime of first computing $\Gamma=\tau(D,\dep,Q)$ for an input
$(D,\dep,Q)$ of size $n$ and maximum arity $w$, and then checking whether
$\Gamma$ is unsatisfiable is $O(2^{O(t(n)\cdot w^w)})$, which is still only
double-exponential.  % The problem of
Deciding $D \cup \dep \models Q$ is thus in \textsc{2exptime}. \hfill$\Box$

\smallskip

Note that in case of bounded $w$, a similar proof does not provide an
\textsc{exptime} bound, as $2^{t(n)\cdot w^w}$ would still be doubly
exponential due to the exponential term $t(n)$, even if $w$ is constant.
Actually, as noted in~\cite{BaGO10}, evaluating non-atomic conjunctive queries
against guarded first-order theories of bounded predicate arity is in fact
\textsc{2exptime}-complete. Surprisingly, this remains true even for guarded
theories $D\cup \dep$ where $D$ is a (variable) database and $\dep$ a {\em
  fixed} guarded theory of a very simple form involving
disjunctions~\cite{bourhis-etal-13}.  We therefore needed to develop different
proof ideas.

\smallskip

In the rest of this section we present an independent and self-contained proof
of Theorem~\ref{theo:upperbound} by developing tools for analyzing the
complexity of query answering under WGTGDs. To this end we introduce the notion
of a \emph{cloud} of an atom $\atom{a}$ in the chase of a database $D$ under a
set $\dep$ of WGTGDs.  Intuitively, the cloud of $\atom{a}$ is the set of atoms
of $\chase{D}{\dep}$ whose arguments belong to $\adom{\atom{a}} \cup \adom{D}$.
In other words, the atoms in the cloud cannot have nulls that do not appear in
$\atom{a}$.  The cloud is important because we will show that the subtree of
$\gcf{D}{\dep}$ rooted in $\atom{a}$ depends only on $\atom{a}$ and its cloud.

\begin{definition}\label{def:cloud}
  Let $\dep$ be a weakly guarded set of TGDs on a schema $\R$ and $D$ be a
  database for $\R$.  For every atom $\atom{a}\in \chase{D}{\dep}$ the
  \emph{cloud} of $\atom{a}$ with respect to $\dep$ and $D$ is the following
  set: $\cloud{D}{\dep}{\atom{a}} = \{ \atom{b} \in \chase{D}{\dep} \mid
  \adom{\atom{b}} \subseteq \adom{\atom{a}} \cup \adom{D} \}$.
  Notice that for every atom $\atom{a} \in \chase{D}{\dep}$ we have $D
  \subseteq \cloud{D}{\dep}{\atom{a}}$.
  Moreover, we define
  \[
  \begin{array}{rcl}
    \clouds{D}{\dep} &=& 
    \{ \cloud{D}{\dep}{\atom{a}} \mid \atom{a} \in \chase{D}{\dep}\}\\
    % The old version follows
    % \clouds{D}{\dep} &=& \bigcup_{\atom{a} \in \chase{D}{\dep}}
    % \{ \cloud{D}{\dep}{\atom{a}} \}\\
%
    \cloudsplus{D}{\dep} &=& \{ (\atom{a}, \cloud{D}{\dep}{\atom{a}})
    \mid \atom{a} \in \chase{D}{\dep} \}
  \end{array}
  \]
  Any subset $S \subseteq \cloud{D}{\dep}{\atom{a}}$ is called a \emph{subcloud} of
  $\atom{a}$ (with respect to $\dep$ and $D$).  The set of all subclouds of an
  atom $\atom{a}$ is denoted by $\subclouds{D}{\dep}{\atom{a}}$.  Finally, we
  define $\subcloudsplus{D}{\dep}= \{(\atom{a},C) ~\mid~ \atom{a} \in
  \chase{D}{\dep} ~\textrm{and}~ C \subseteq \cloud{D}{\dep}{\atom{a}} \}$.
\end{definition}

\begin{definition}\label{def:d-isomorphic}
  Let $B$ be an instance (possibly with nulls) over a schema, $\R$, and $D$ be
  a database over $\R$.  Let $\alpha$ and $\beta$ be atoms from the Herbrand
  Base $\hb(B)$.  We say that $\alpha$ and $\beta$ are \emph{$D$-isomorphic},
  denoted $\alpha \simeq_D \beta$, or simply $\alpha \simeq \beta$ in case $D$
  is understood, if there is a bijective homomorphism\footnote{
    Recall that, by definition, the
    restriction of a homomorphism to $\adom{D}$ is the identity mapping.
  }
  $f:~\adom{\alpha} \ra \adom{\beta}$ such that $f(\alpha) = \beta$ (and thus
  also $f^{-1}(\beta) = \alpha$).  This definition extends directly to the
  cases when $\alpha$ and $\beta$ are sets of atoms or atom-set pairs (in a
  similar fashion as in $\cloudsplus{D}{\dep}$).
\end{definition}

\begin{example}
  If $\set{a,b} \subseteq \adom{D}$ and $\set{\zeta_1, \zeta_2, \zeta_3,
    \zeta_4} \subseteq \freshdom$, we have: $p(a,\zeta_1,\zeta_2) \simeq p(a,
  \zeta_3, \zeta_4)$ and $(p(a,\zeta_3), \{q(a,\zeta_3), q(\zeta_3, \zeta_3),
  r(\zeta_3)\}) \simeq (p(a,\zeta_1), \{q(a,\zeta_1), q(\zeta_1,\zeta_1),
  r(\zeta_1)\})$.  On the other hand, $p(a, \zeta_1, \zeta_2) \not\simeq p(a,
  \zeta_1, \zeta_1)$ and $p(a, \zeta_1, \zeta_2) \not\simeq p(\zeta_3, \zeta_1,
  \zeta_2)$.
\end{example}

\begin{lemma} \label{lem:equivalence-rel} Given a database $D$ for a schema
  $\R$ and an instance $B$ for $\R$, the $D$-isomorphism relation $\simeq_D$ on
  $\hb(B)$ (resp., $2^{\hb(B)}$ or $\hb(B) \times 2^{\hb(B)}$, as in
  Definition~\ref{def:d-isomorphic}) is an equivalence relation.
\end{lemma}

The above lemma follows directly from the definitions; it lets us define, for
every set $A$ of atoms of $\hb(B)$, the quotiont set $A/_{\simeq_D}$ with
respect to the above defined equivalence relation $\simeq_D$.  Such notion of
quotient set naturally extends to sets of sets of atoms such as
$\clouds{D}{\dep}$, or sequences (pairs, in particular) thereof.

% The above lemma follows directly from the definitions; it lets us define
% quotient sets of atoms with values in $\adom{D} \cup \freshdom$, \wrt~the
% above defined equivalence relation $\simeq_D$.  Such notion of quotient set
% naturally extends to sets of sets of atoms such as $\clouds{D}{\dep}$, or
% sequences (pairs, in particular) thereof.
%
\begin{lemma}
  \label{lem:cloudsize}
  Let $\dep$ be a weakly guarded set of TGDs and let $D$ be a database for a schema
  $\R$.  Let $|\R|$ denote the number of predicate symbols in $\R$, and $w$ be
  the maximum arity of the symbols in $\R$.  Then:
  \begin{plist} \itemsep-\parsep
  \item For every atom $\atom{a} \in \chase{D}{\dep}$, we have \(
    |\cloud{D}{\dep}{\atom{a}}| \leq |\R| \cdot (|\adom{D}| + w)^w \).  Thus,
    $\cloud{D}{\dep}{\atom{a}}$ is polynomial in the size of the database $D$
    if the arity $w$ is bounded and exponential otherwise.
  \item For each atom $\atom{a} \in \chase{D}{\dep}$,
    $|\subclouds{D}{\dep}{\atom{a}}| \leq 2^{|\R| \cdot (|\adom{D}| + w)^w}$.
  \item $|\clouds{D}{\dep}/_\simeq| \leq
    % |\subclouds{D}{\dep}{\atom{a}}/_\simeq| \leq
    2^{|\R| \cdot (|dom(D)| + w)^w}$, i.e., there are---up to isomorphism---at
    most exponentially many possible clouds or subclouds in a chase, if the
    arity $w$ is bounded, otherwise doubly exponentially many.
\item $|\cloudsplus{D}{\dep}/_\simeq| \leq |\subcloudsplus{D}{\dep}/_{\simeq}|
  \leq |\R| \cdot (|\adom{D}| + w)^w \cdot 2^{|\R| \cdot (|\adom{D}| + w)^w}$.
\end{plist}
\end{lemma}

\begin{proof}
  The claims are proved by combinatorial arguments as follows.
  \begin{plist} \itemsep-\parsep
  \item All distinct atoms in a cloud are obtained by placing the symbols of
    $\atom{a}$, plus possibly symbols from $\adom{D}$, in at most $w$ arguments
    of some predicate symbol in $\R$.  For each predicate in $\R$, the number
    of symbols to be thus placed is $|\adom{D}| + w$.
  \item The different ways we can choose $\subclouds{D}{\dep}{\atom{a}}$
    clearly determines the set of all subsets of $\cloud{D}{\dep}{\atom{a}}$.
  \item It is easy to see that the size of the set of all
    non-pairwise-isomorphic clouds in the chase is bounded by the number of
    possible subclouds of a fixed atom.
  \item Here, we are counting the number of all possible subclouds, each
    associated with its ``generating'' atom.  The inequality holds because,
    once we choose all non-pairwise-isomorphic clouds, each of their possible
    generating atoms can have as arguments only $|\adom{D}| + w$ symbols with
    which to construct the subclouds.
  \end{plist}
\vspace{-0.7cm}
\end{proof}

\vspace{-0.4cm}
%%%%%%%%%%%%%%%%%%%%%%%%%%%%%%%%%%%%%%%%%%%%%%%%%%%%%%%%%%%%%%%%%%%%%%%%%%%%%

\begin{definition}
  \label{def:newgcf}
  Given a database $D$ and a set of WGTGDs, let $\atom{a}$ be an atom in
  $\chase{D}{\dep}$.  We define the following notions:
  \begin{itemize} \itemsep-\parsep
  \item $\troot{\atom{a}}$ is the set of all atoms that label nodes of the
    subtrees of $\gcf{D}{\dep}$ rooted in $\atom{a}$;
    % \item $^\downarrow\atom{a} = ^\downarrow \atom{a} -
    %   \cloud{D}{\dep}{\atom{a}}$;
    %%   Never used in the following; superfluous.
  \item $\nabla \atom{a} = \troot{\atom{a}} \cup \cloud{D}{\dep}{\atom{a}}$;
  % \item $\ggcf{\atom{a}}{S} = \gcf{D}{\dep} \cap \chase{\set{\atom{a}}\cup
  %     S}{\dep}$.\footnote{%% 
  %     We omit $D$ and $\dep$ in $\ggcf{\atom{a}}{S}$,  to avoid clutter.
      %% 
  %  }
    %% 
  \item if $S$ is a subset of atoms in $\gcf{D}{\dep}$, then
    $\ggcf{\atom{a}}{S}$\footnote{
      $D$ and $\dep$ are implicit here, to avoid clutter.
    }
    is inductively defined as follows: \textit{(i)} $S \cup\{ \atom{a} \}
    \subseteq
    \ggcf{\atom{a}}{S}$;\\
    \textit{(ii)} $\atom{b} \in \ggcf{\atom{a}}{S}$ if $\atom{b} \in
    \troot{\atom{a}}$, and $\atom{b}$ is obtained via the chase rule applied
    using a TGD with body $\Phi$ and head-atom $\atom{\psi}$, and a homomorphism
    $\theta$, such that $\theta(\atom{\psi}) = \atom{b}$
   %%(recall that we are dealing with single-headed TGDs)
   and $\theta(\Phi) \subseteq \ggcf{\atom{a}}{S}$.
  \end{itemize}
\end{definition}

%%%%%%%%%%%%%%%%%%%%%%%%%%%%%%%%%%%%%%%%%%%%%%%%%%%%%%%%%%%%%%%%%%%%%%%%%%%%%%

%%Observe that, intuititively, $\ggcf{\atom{a}}{S}$ is the set of atoms in
%%the subtree of $\gcf{D}{\dep}$ having $\atom{a}$ as root and generated by
%%using the set $S$ of atoms as the (sub)cloud of $\atom{a}$ to start.
%
% The central importance of clouds in the context of weakly guarded TGDs is
% that if $\atom{a}$ is an atom of a generalized chase tree $\gcf{D}{\dep}$,
% then $\trootc{\atom{a}}$ is \emph{determined} by $\cloud{D}{\dep}{\atom{a}}$
% (modulo, of course, renaming of labeled nulls).

\begin{theorem}\label{theo:dependcloud}
  If $D$ is a database for a schema $\R$, $\dep$ is a weakly guarded set of
  TGDs, and $\atom{a} \in \chase{D}{\dep}$, then $\trootc{\atom{a}} =
  \ggcf{\atom{a}}{\cloud{D}{\dep}{\atom{a}}}$.
\end{theorem}

\begin{proof}
  \begin{andrea}
    By the definitions of $\trootc{\atom{a}}$ and
    $\ggcf{\atom{a}}{\cloud{D}{\dep}{\atom{a}}}$, we have
    $\ggcf{\atom{a}}{\cloud{D}{\dep}{\atom{a}}} \subseteq \trootc{\atom{a}}$.
    It remains to show the converse inclusion: $\trootc{\atom{a}} \subseteq
    \ggcf{\atom{a}}{\cloud{D}{\dep}{\atom{a}}}$.  Define
    $\level_{\atom{a}}(\atom{a}) = 0$ and for each fact $\atom{b} \in
    \cloud{D}{\dep}{\atom{a}} - \trootc{\atom{a}}$ we also define
    $\level_{\atom{a}}(\atom{b}) = 0$. For every other atom $\atom{c} \in
    \troot{\atom{a}}$, $\level_{\atom{a}}(\atom{c})$ is defined as 
    the distance (i.e., the
    length of the path) from $\atom{a}$ to $\atom{c}$ in $\gcf{D}{\dep}$.

    % Let $U = \chase{D}{\dep}$.  For each set $A \subseteq \adom{U}$, let
    % $U|_A = \{\atom{c} \in U ~\mid~ \adom{\atom{c}} \subseteq A\}$.
    We first show the following facts in parallel by induction on
    $\level_{\atom{a}}(\atom{b})$:
    \begin{plist} \itemsep-\parsep
    \item If $\atom{b} \in \trootc{\atom{a}}$
      % was previously \ggcf{\atom{a}}{\cloud{D}{\dep}{\atom{a}}},
      then $\cloud{D}{\dep}{\atom{b}} \subseteq
      \ggcf{\atom{a}}{\cloud{D}{\dep}{\atom{a}}}$.
    \item If $\atom{b} \in \trootc{\atom{a}}$ then $\atom{b} \in
      \ggcf{\atom{a}}{\cloud{D}{\dep}{\atom{a}}}$.
    \end{plist}
    Statement \textit{(2)} above is the converse inclusion we are after.

    \ourpar{Induction basis.}  If $\level_{\atom{a}}(\atom{b}) = 0$, we
    have either
    \textit{(a)} $\atom{b} \in \cloud{D}{\dep}{\atom{a}} - \set{\atom{a}}$, or
    \textit{(b)} $\atom{b} = \atom{a}$.
%
    % \begin{plist} \itemsep-\parsep
    % \item[\textit{(a)}] $\atom{b} \in \cloud{D}{\dep}{\atom{a}} -
    % \set{\atom{a}}$, or
    % \item[\textit{(b)}] $\atom{b} = \atom{a}$.
    % \end{plist}
    In case
    \textit{(a)}, $\cloud{D}{\dep}{\atom{a}} \subseteq
    \ggcf{\atom{a}}{\cloud{D}{\dep}{\atom{a}}}$ and therefore $\atom{b} \in
    \ggcf{\atom{a}}{\cloud{D}{\dep}{\atom{a}}}$, which proves \textit{(1)}.
    Moreover, since $\atom{b} \in \cloud{D}{\dep}{\atom{a}}$,
    $\atom{b}$ cannot contain more labeled nulls than $\atom{a}$,
    so $\adom{\atom{b}} - \adom{D} \subseteq \adom{\atom{a}} - \adom{D}$.
    Therefore $\cloud{D}{\dep}{\atom{b}} \subseteq \cloud{D}{\dep}{\atom{a}}
    \subseteq \ggcf{\atom{a}}{\cloud{D}{\dep}{\atom{a}}}$, which proves
    \textit{(2)}.  In case \textit{(b)}, $\atom{b} =
    \atom{a}$ and thus $\cloud{D}{\dep}{\atom{a}} = \cloud{D}{\dep}{\atom{b}}
    \subseteq \ggcf{\atom{a}}{\cloud{D}{\dep}{\atom{a}}}$, which proves
    \textit{(1)}.  Since $\atom{b} = \atom{a} \in
    \ggcf{\atom{a}}{\cloud{D}{\dep}{\atom{a}}}$, \textit{(2)} follows as
    well.

    \ourpar{Induction step}.  Assume that \textit{(1)} and \textit{(2)} are
    satisfied for all $\atom{c} \in \trootc{\atom{a}}$ such that
    $\level_{\atom{a}}(\atom{c}) \leq i$ and  assume $\level_{\atom{a}}(\atom{b})
    = i+1$, where $i \geq 0$.  The atom $\atom{b}$ is produced by a TGD whose guard
    $\atom{g}$ matches some atom $\atom{b}^-$ at level $i$, which is, by
    the induction hypothesis, in $\ggcf{\atom{a}}{\cloud{D}{\dep}{\atom{a}}}$.
    The body atoms of such a TGD then match atoms whose arguments must be in
    $\cloud{D}{\dep}{\atom{b}}$ and thus also in
    $\ggcf{\atom{a}}{\cloud{D}{\dep}{\atom{a}}}$, again by the induction
    hypothesis.  Therefore, \textit{(2)} holds for $\atom{b}$.  To show
    \textit{(1)}, consider an atom $\atom{b}' \in \cloud{D}{\dep}{\atom{b}}$.
    In case $\adom{\atom{b}'} \subseteq \adom{\atom{b}^-}$, we have
    % $U|_{\adom{\atom{b}'}}\subseteq U|_{\adom{\atom{b}^-}} \subseteq
    $\cloud{D}{\dep}{\atom{b}'} \subseteq \cloud{D}{\dep}{\atom{b}^-} \subseteq
    \ggcf{\atom{a}}{\cloud{D}{\dep}{\atom{a}}}$.  Otherwise, $\atom{b}'$
    contains
    at least one new labeled null that was introduced during the
    generation of $\atom{b}$.  Given that $\dep$ is a weakly guarded set and
    each labeled null in $\freshdom$ is introduced only once in the chase,
    there must be a path from $\atom{b}$ to $\atom{b}'$ in $\gcf{D}{\dep}$
    (and therefore also in $\trootc{\atom{b}}$).  A simple additional
    induction
    on $\level_{\atom{b}}(\atom{b}')$ shows that all the applications of
    TGDs on that path must have been fired on elements of
    $\ggcf{\atom{a}}{\cloud{D}{\dep}{\atom{a}}}$ only.  Therefore,
    $\atom{b}'\in \ggcf{\atom{a}}{\cloud{D}{\dep}{\atom{a}}}$, which
    proves \textit{(1)}.
  \end{andrea}
\end{proof}

The corollary below follows directly from the above theorem.

\begin{corollary}\label{col:sim}
  If $D$ is a database for a schema $\R$, $\dep$ is a weakly guarded set of TGDs,
  $\atom{a}, \atom{b} \in \chase{D}{\dep}$, and
  $(\atom{a},\cloud{D}{\dep}{\atom{a}}) \simeq
  (\atom{b},\cloud{D}{\dep}{\atom{b}})$, then $\nabla \atom{a} \simeq \nabla
  \atom{b}$.
\end{corollary}

%%%%%%%%%%%%%%%%%%%%%%%%%%%%%%%%%%%%%%%%%%%%%%%%%%%%%%%%%%%%%%%%%%%%%%%%%%%%%

\begin{definition}\label{def:cancloud}
  Let $D$ be a database and $\atom{a}$ an atom.  The \emph{canonical renaming}
  $\can_{\atom{a}}: \adom{\atom{a}} \cup \adom{D} \ra \candom{\atom{a}} \cup
  \adom{D}$, where $\candom{\atom{a}} = \{\xi_1, \ldots,
  \xi_h\}\subset \freshdom$ is a set of labeled nulls \emph{not} appearing in
  $\atom{a}$, is a 1-1
  substitution that maps each element of $\adom{D}$ into itself and each
  null-argument of $\atom{a}$ to the first unused element
  $\xi_i\in \candom{\atom{a}}$.
  If $S \subseteq \cloud{D}{\dep}{\atom{a}}$ 
  then $\can_{\atom{a}}(S)$ is well-defined
  and the pair
  $(\can_{\atom{a}}(\atom{a}),\can_{\atom{a}}(S))$ will be denoted by
  $\can(\atom{a},S)$.
\end{definition}

\begin{example}
  Let $\atom{a} = g(d, \zeta_1, \zeta_2, \zeta_1)$ and $S = \{p(\zeta_1),
  r(\zeta_2, \zeta_2), s(\zeta_1, \zeta_2, b)\}$, where $\set{d,b} \subseteq
  \adom{D}$ and $\set{\zeta_1, \zeta_2} \subseteq \freshdom$.  Then
  $\can_{\atom{a}}(\atom{a}) = g(d,\xi_1,\xi_2,\xi_1)$, and $\can_{\atom{a}}(S)
  = \set{p(\xi_1), r(\xi_2,\xi_2)$, $s(\xi_1,\xi_2,b)}$.
\end{example}

\begin{andrea}
  % \noteac{The following definition has been polished as it the original
  % version had some slightly imprecise notation.}
  \begin{definition}\label{def:app-answersubtree}
    If $D$ is a database for a schema $\R$, $\dep$ is a weakly guarded set of
    TGDs on $\R$, $S$ is a set of atoms and $\atom{a} \in S$, then we write
    $(D, \dep, \atom{a}, S) \models Q$ iff there exists a homomorphism $\theta$
    such that $\theta(Q) \subseteq S \cup \troot{\atom{a}}$.
  \end{definition}
\end{andrea}

The following result straightforwardly follows from
Theorem~\ref{theo:dependcloud} and the previous definitions.

\begin{corollary}\label{col:app-sssim}
  If $D$ is a database for a schema $\R$, $\dep$ is a weakly guarded set of
  TGDs, $\atom{a} \in \chase{D}{\dep}$, and $Q$ is a Boolean conjunctive query,
  then the following statements are equivalent:
  \begin{plist}\itemsep-\parsep
  \item $\nabla \atom{a} \models Q$
  \item $(D, \dep, \atom{a}, \cloud{D}{\dep}{\atom{a}}) \models Q$
  \item $(D, \dep, \can_{\atom{a}}(a),
    can_{\atom{a}}(\cloud{D}{\dep}{\atom{a}})) \models Q$
  \item there is a subset $S' \subseteq \cloud{D}{\dep}{\atom{a}}$ such
    that $(D, \dep, \can_{\atom{a}}(\atom{a}), \can_{\atom{a}}(S')) \models Q$.
  \end{plist}
\end{corollary}

%%%% MK: The sentence below is unparsable and is hard to understand. Does
%%%% not convey much info. So, I commented it out.
%%To design an alternating algorithm that computes the part of $\chase{D}{\dep}$
%%that is necessary for answering a query $Q$, we can represent any
%%sub-computation of $\ggcf{\atom{a}}{S}$, with $S \subseteq
%%\cloud{D}{\dep}{\atom{a}}$, % which we indicate with the pair $(\atom{a}, S)$,
%%to its canonical form $\ggcf{\can_{\atom{a}}(\atom{a})}{\can_{\atom{a}}(S)}$.

We will use the pair $\can(\atom{a}, \cloud{D}{\dep}{\atom{a}})$ as a unique
canonical representative of the equivalence class \( \set{
  (\atom{b},\cloud{D}{\dep}{\atom{b}}) \mid (\atom{b},
  \cloud{D}{\dep}{\atom{b}}) \simeq (\atom{a},\cloud{D}{\dep}{\atom{a}})} \) in
\linebreak $\cloudsplus{D}{\dep}$.  Therefore, the set \( \set{
  \can(\atom{a},\cloud{D}{\dep}{\atom{a}}) ~\mid~ \atom{a} \in \chase{D}{\dep}}
\) and the quotient set $\cloudsplus{D}{\dep} /_\simeq$ are isomorphic.  Note
that, by Lemma~\ref{lem:cloudsize}, these sets are finite and have size
exponential in $|D| + |\dep|$ if the schema is fixed (and double exponential
otherwise).

Now, given a database $D$ for a schema $\R$, a weakly guarded set of TGDs
$\dep$ on $\R$, and an \emph{atomic} Boolean conjunctive query $Q$, we describe
an alternating algorithm $\acheck(D,\dep,Q)$ that decides whether $D \cup \dep
\models Q$.  We assume that $Q$ has the form $\exists Y_1,\ldots, Y_\ell$,
$q(t_1,t_2, \ldots, t_r)$, where the $t_1, \ldots, t_r$, with $r \geq \ell$,
are terms (constants or variables) in $\adom{D} \cup \set{Y_1, Y_2, \ldots,
  Y_\ell}$.

The algorithm $\acheck${} returns ``true'' if it accepts some configuration,
according to the criteria explained below; otherwise, it returns
``\emph{false''}.
$\acheck$ uses tuples of the form $(\atom{a},S,S',\prec,\atom{b})$ as its basic
data structures (\emph{configurations}).  Intuitively, each such configuration
corresponds to an atom $\atom{a}$ derived at some step of the chase
computation together with a set $S'$ of already derived atoms belonging  the cloud of $\atom{a}$. 
The informal meaning of the parameters of a configuration is as follows.
\begin{plist}\itemsep-\parsep
\item $\atom{a}$ is the root atom of the chase subtree under consideration.
\item $S \subseteq \cloud{D}{\dep}{\atom{a}}$; $S$ is intuitively a subcloud containing a set of  atoms of $\cloud{D}{\dep}{\atom{a}}$ that, while computing $\chase{D}{\dep}$, are originally derived \emph{outside} the subtree of the guarded chase forest rooted in $\atom{a}$ (and are thus outside 
the subtree rooted in $\atom{a}$ of $\rgcf{D}{\dep}$). We expect these atoms to serve as  ``side atoms" (i.e., atoms matching  non-guard atoms of a TGD) when deriving the desired atom $\atom{b}$ starting at  $\atom{a}$. 
  \item $S'$ contains, at every step in the computation, the subset of
    $\cloud{D}{\dep}{\atom{a}}$ that has been computed so far, or can be assumed to be valid, as it will be verified in another branch of the computation.
  \item $\prec$ is a total ordering of the atoms in $S$ consistent with the
    order in which the atoms of $S$ are proved by the algorithm (by simulating
    the chase procedure).
\item $\atom{b}$ is an atom that needs to be derived.  In some cases (namely,
  on the ``main'' path in the proof tree developed by $\acheck$), the algorithm
  will not try to derive a specific atom, but will just match the query atom
  $q(t_1, \ldots, t_r)$ against the atoms of that path. In that case, we use
  the symbol $\star$ in place of $\atom{b}$.
\end{plist}
We are now ready to describe the algorithm $\acheck$ at a sufficiently detailed level. However, we  omit many low-level details.  

\bigskip

$\acheck$  first checks if $D
\models Q$.  If so, $\acheck$ returns ``true'' and halts.
Otherwise, the algorithm attempts to guess a path, the so called {\em main branch},  that contains an atom
$\atom{q}$ that is an instance of $Q$. This is done as follows.

\ourpar{Initialization.}  The algorithm $\acheck$ starts at $D$ and guesses
some atom $\atom{a} \in D$, which it will expand into a main branch that will
eventually lead to an atom $\atom{q}$ matching the query $Q$.  To this end, the
algorithm guesses a set $S \subseteq \cloud{D}{\dep}{\atom{a}}$ and a total
order $\prec$ on $S$, and then generates the configuration $c_0 =
(\atom{a},S,S',\prec,\star)$.  The set $S'$ is initialized as $S'=S$.

% \medskip

\ourpar{Form of a configuration---additional specifications.} 
  In each configuration, the set 
set $S$ is implicitly partitioned into two sets $S^\bot$ and $S^+$, where $S^\bot
\subseteq D$ and $S^+ = \{\atom{a}_1, \atom{a}_2, \ldots, \atom{a}_k\}$ is
disjoint from $D$.  The total order $\prec$ is such that all elements of
$S^\bot$ precede those of $S^+$.  On $S^+$, $\prec$ is defined as $\atom{a}_1
\prec \atom{a}_2 \prec \cdots \prec \atom{a} \prec \cdots \prec \atom{a}_k$.

% \medskip

\ourpar{Summary of tasks $\acheck$ performs for each configuration.}  Assume
the $\acheck$ algorithm generates a configuration
$c=(\atom{a},S,S',\prec,\atom{b})$, where $\atom{b}$ might be
$\star$. $\acheck$ then performs the following tasks on $c$:
\begin{itemize} \itemsep-\parsep
\item $\acheck$ verifies that the guessed set $S$ of $c$ is actually a subset
  of $\cloud{D}{\dep}{\atom{a}}$. This is achieved by a massive universal
  branching that will be described below under the heading ``Universal
  Branching''.  Let us, however, anticipate here how it works, as this may
  contribute to the understanding of the other steps.  $\acheck$ will verify
  that each of the atoms $\atom{a}_1,\ldots, \atom{a}_k$ is in
  $\chase{D}{\dep}$, where, for each $i \in \set{1, \ldots, k}$, the proof of
  $\atom{a}_i \in \chase{D}{\dep}$ can use as premises only the atoms of $S$
  that precede $\atom{a}_i$, according to $\prec$.  The algorithm thus finds
  suitable atoms $\atom{d}_1, \ldots, \atom{d}_k \in D$ and builds proof trees
  for $\atom{a}_1, \ldots, \atom{a}_k$.  For each $1 \leq i \leq k$, it
  generates configurations of the form $(\atom{d}_i, S, S^{\bot}\cup
  \{\atom{a}_1, \atom{a}_2, \ldots, \atom{a}_{i-1}\}, \prec, \atom{a}_i)$.
  Each such configuration will be used as a starting point in a proof of
  $\atom{a}_i \in \chase{D}{\dep}$ assuming that $\atom{a}_1, \ldots,
  \atom{a}_{i-1}\in \chase{D}{\dep}$ has already been established.  $\acheck$
  thus simulates a sequential proof of all atoms of $\cloud{D}{\dep}{\atom{a}}$
  that are in $S$ via a parallel universal branching from $c$.
\item $\acheck$ tests whether $c$ is a final configuration (i.e., an accepting
  or rejecting one).  This is described under the heading ``Test for final
  Configuration" below.
\item If $c$ is not a final configuration of $\acheck$, this means that its
  first component $\atom{a}$ is not yet the one that will be matched to
  $\atom{b}$ (or the query, if $\atom{b}=\star$).  $\acheck$ then ``moves down"
  the chase tree by one step by replacing $\atom{a}$ with a child of
  $\atom{a}$. This step is described under the heading ``Existential Branching".
\end{itemize} 
In the following, let $c = (\atom{a},S,S',\prec,\atom{b})$ be a configuration,
where $\atom{b}$ may be $\star$.
%$S=\{\atom{a}_1,\atom{a}_2, \ldots, \atom{a}_k\}$, $S'=\{\atom{a}_1,
%\atom{a}_2, \ldots, \atom{a}_i\}$, and $\atom{a}_1 \prec \atom{a}_2 \prec
%\cdots \prec \atom{a} \prec \cdots \prec \atom{a}_k$.  

% \bigskip

\ourpar{Test for final configuration.}  If $\atom{b} \in D$, then $\acheck$
accepts this configuration, and does not expand it further.  If $\atom{b} =
\star$, then $\acheck$ checks (via a simple subroutine) whether $Q$ matches
$\atom{a}$, i.e., if $\atom{a}$ is a homomorphic image of the query atom
$q(t_1, \ldots, t_r)$.
% \marginpar{\tiny MK: does "accept" mean returns "true?" Say so!}
If so, $\acheck$ accepts $c$ (and thus returns ``true'') and does not expand it
further.  If $\atom{b} \neq \star$, $\acheck$ checks whether
$\atom{a}=\atom{b}$.  %GGGG
If this is true, then $\acheck$ % also %GGGG
accepts the configuration $c$ and does not expand it further.  Otherwise, the
configuration tree is expanded as described next.

% \medskip

\ourpar{Existential Branching.} $\acheck$ guesses a TGD $\rho \in \dep$ having
body $\Phi$ and head-atom $\atom{\psi}$, and whose guard $\atom{g}$ matches
$\atom{a}$ via some substitution $\theta$ (that is, $\theta(\atom{g}) =
\atom{a}$) such that $\theta(\Phi) \subseteq S'$. $\theta(\atom{\psi})$ then
corresponds to a newly generated atom (possibly containing some fresh labeled
nulls in $\freshdom$).  Note that, if no such guess can be made, this
existential branching automatically fails and $\acheck$ returns \emph{false}.
To define the configuration $c_1$ that $\acheck$ creates out of $c$, we first
introduce an intermediate auxiliary configuration $\hat{c} = (\hat{\atom{a}},
\hat{S}, \hat{S'}, \hat{\prec}, \hat{\atom{b}})$, where:
\begin{aplist}\itemsep-\parsep
\item $\hat{\atom{a}} = \theta(\atom{\psi})$ is the new atom generated by the
  application of $\rho$ with the substitution $\theta$.
\item $\hat{S}$ contains $\hat{\atom{a}}$ and each atom $\atom{d}$ of $S$ such
  that $\adom{\atom{d}} \subseteq \adom{\hat{\atom{a}}} \cup \adom{D}$.  Thus,
  in addition to the new atom $\hat{\atom{a}}$, $\hat{S}$ \emph{inherits} all
  atoms that were in the subcloud $S$ of the parent configuration $c$ that are
  ``compatible'' with $\hat{\atom{a}}$.  In addition, $\hat{S}$ includes a set
  $\newatoms(\hat{c})$ of new atoms that are guessed by the $\acheck$
  algorithm. All arguments of each atom of $\newatoms(\hat{c})$ must be
  elements of the set $\adom{\hat{\atom{a}}} \cup \adom{D}$.
  % Each of these atoms must contain at least one labeled null from
  % $\hat{\atom{a}}$ or a domain element from $D$ that does not occur in $S$.
  % \marginpar{\tiny MK: Don't think it's correct remark.}  (Otherwise they
  % could
  % not be new w.r.t.~S; of course they cannot have as arguments other nulls
  % than
  % those in $\hat{\atom{a}}$).

  % *****
  % Do not think it follows that $\newatoms(\hat{c})$ must have a domain
  % element not in $S$. It is possible that all their domain elements are in
  % $S$, but are in a different order so that these atoms are not in $D$ or $S$.
  % *****
\item $\hat{S'} = \hat{S}$.
\item $\hat{\prec}$ is a total order on $\hat{S'}$ obtained from $\prec$ by
  eliminating all atoms
  in $S - \hat{S}$ and by
  ordering the atoms from $\newatoms(\hat{c})$
  after all the atoms from the set $\oldproved(\hat{c}) = \hat{S'} \cap S'$
  (these are 
  assumed to have already been proven at the parent configuration $c$).
\item $\hat{\atom{b}}$ is defined as $\hat{\atom{b}} = \atom{b}$.
\end{aplist}

Next, $\acheck$ constructs the configuration $c_1$ out of $\hat{c}$ by
canonicalization: $c_1 = \can_{\hat{\atom{a}}}(\hat{c})$, that is
\(
c_1 = (\can_{\hat{\atom{a}}}(\hat{\atom{a}}),\, \can_{\hat{\atom{a}}}(\hat{S}),\,
\can_{\hat{\atom{a}}}(\hat{S'}),\, \can_{\hat{\atom{a}}}(\hat{\prec}),\,
\can_{\hat{\atom{a}}}(\hat{\atom{b}}))
\),
where $\can_{\hat{\atom{a}}}(\hat{\prec})$ is the total order on the atoms in
$\can_{\hat{\atom{a}}}(\hat{S'})$ derived from $\hat{\prec}$.

Intuitively, $c_1$ is the ``main'' child of $c$ on the way to deriving the
query atom $q(t_1, \ldots, t_r)$ assuming that all atoms of the guessed
subcloud $S$ are derivable.

% \medskip

\ourpar{Universal Branching.}  In the above generated configuration $\hat{c}$,
the set $\hat{S'}$ is equal to $\hat{S}$. As already said, this means that it
is assumed for that configuration that the set of atoms $\hat{S}$ is derivable.
To verify that this is indeed the case, $\acheck$ generates in parallel, using
universal computation branching, a set of auxiliary configurations for proving
that all the guessed atoms in $\can_{\hat{\atom{a}}}(\newatoms(\hat{c}))$ are
indeed derivable through the chase of $D$ \wrt~$\dep$.

Let $\can_{\hat{\atom{a}}}(\newatoms(\hat{c})) = \set{\atom{n}_1, \ldots,
  \atom{n}_m}$ and let the linear order $\hat{\prec}$ on $\hat{S}$ be a
concatenation of the order $\prec$, restricted to $\oldproved(\hat{c})$, and
the order $\atom{n}_1 \hat{\prec} \atom{n}_2 \hat{\prec} \cdots \hat{\prec}
\atom{n}_m$.
For each $1 \leq i \leq m$, $\acheck$ generates a configuration $c^{(i)}_2$
defined as
\[
c^{(i)}_2 = (\can_{\hat{a}}(\atom{\hat{a}}),\, \can_{\hat{a}}(\hat{S}),\,
\can_{\hat{a}}(\oldproved(\hat{c})) \cup \set{\atom{n}_1, \ldots,
  \atom{n}_{i-1}},\, \can_{\hat{\atom{a}}}(\hat{\prec}),\, \atom{n}_i).
\]
% which is the canonical form (w.r.t.~$\hat{\atom{a}}$) of an intermediate
% configuration $c^{(i)}_2$, that is $c^{(i)}_3 =
% \can_{\hat{\atom{a}}}(c^{(i)}_2)$, where
% \[
% c^{(i)}_2 = (\hat{\atom{a}}, \hat{S}, \oldproved(\hat{c}) \cup
% \set{\atom{n}_1, \ldots, \atom{n}_{i-1}}, \hat{\prec}, \atom{n}_i)
% \]
%
This completes the description of the $\acheck$ algorithm.

\begin{theorem}\label{theo:app-acheck}
  The $\acheck$ algorithm is correct and runs in exponential time in case of
  bounded arities, and in double exponential time otherwise.
\end{theorem}

\begin{proof}\mbox{}

  % \marginpar{\tiny\textbf{You talk about true and accept in describing
  % Acheck. Which is which?}}
  \ourpar{Soundness.}  It is easy to see that the algorithm is sound with
  respect to the standard chase, i.e., if $\acheck(D,\dep,Q)$ returns
  ``\emph{true}'', then $\chase{D}{\dep} \models Q$.  In fact, modulo variable
  renaming, which preserves soundness according to
  Corollary~\ref{col:app-sssim}, the algorithm does nothing but chasing $D$
  with respect to $\dep$, even if the chase steps are not necessarily in the
  same order as in the standard chase.  Thus, each atom derived by $\acheck$
  occurs in \emph{some} chase.  Since every chase computes a universal solution
  that is complete with respect to conjunctive query answering, whenever
  $\acheck$ returns ``true'', $Q$ is entailed by some chase, and thus also by
  the standard chase, $\chase{D}{\dep}$.

  \ourpar{Completeness.}  The completeness of $\acheck$ with respect to
  $\chase{D}{\dep}$ can be shown as follows.  Whenever $\chase{D}{\dep} \models
  Q$, there is a finite proof of $Q$, i.e., a finite sequence
  $\mathit{proof}_Q$ of generated atoms that ends with some atom $\atom{q}$,
  which is an instance of $Q$.  This proof can be simulated by the alternating
  computation $\acheck$ as follows: \textit{(i)} steer the main branch of
  $\acheck$ towards (a variant of) $\atom{q}$ by choosing successively the same
  TGDs and substitutions $\theta$ (modulo the appropriate variable renamings)
  as those used in the standard chase for the branch of $\atom{q}$;
  \textit{(ii)} whenever a subcloud $S$ has to be chosen for some atom
  $\atom{a}$ by $\acheck$, choose the set of atoms $\cloud{D}{\dep}{\atom{a}}
  \cap (D \cup \atoms{\mathit{proof}_Q})$, modulo appropriate variable
  renaming; \textit{(iii)} for the ordering $\prec$, always choose the one
  given by $\mathit{proof}_Q$.  The fact that no $Q$-instance is lost when
  replacing configurations % $\hat{c}$
  by their canonical versions % $c_1 =\can_{\atom{a}}(\hat{c})$
  is guaranteed by Corollary~\ref{col:app-sssim}.

  % ******* \textbf{$c_1$ is not used! Only $c_2$ is really
  % used!! See comment above!} ******

  \ourpar{Computational cost.}  In case of bounded arity, the size of each
  configuration $c$ is polynomial in $D \cup \dep$.  Thus, $\acheck$ describes
  an alternating \textsc{pspace} (i.e., \textsc{apspace}) computation.  It is
  well-known that \textsc{apspace} = \textsc{exptime}.  In case the arity is
  not bounded, each configuration requires at most exponential space. The
  algorithm then describes a computation in Alternating \textsc{expspace},
  which is equal to 2\textsc{exptime}.
\end{proof}

\begin{corollary} \label{col:app-chasebottom} Let $\dep$ be a weakly guarded
  set of TGDs, and let $D$ be a database over a schema $\R$.  Then, computing
  $\chasehb{D}{\dep}$ can be done in exponential time in case of bounded arity,
  and in double exponential time otherwise.
\end{corollary}

\begin{proof}
  It is sufficient to start with an empty set $A$ and then cycle over 
  ground atoms $\atom{b}$ in the Herbrand base $\hb(D)$ while checking
  whether $\chase{D}{\dep} \models \atom{b}$.  If this holds, we add $\atom{b}$
  to $A$.  The result is $\chasehb{D}{\dep}$.  The claimed time bounds follow
  straightforwardly.
\end{proof}

\medskip

We can now finally state our independent proof of
Theorem~\ref{theo:upperbound}.

\medskip

\noindent{\em Proof of Theorem~\ref{theo:upperbound}}.\   
  We construct an algorithm $\qcheck$ such that 
  $\qcheck(D,\dep,Q)$ outputs ``true'' iff $D \cup \dep \models Q$
  (i.e., iff $\chase{D}{\dep} \models Q$).  The algorithm relies
  on the notion of squid decompositions, and on Lemma~\ref{lem:squid};
  it works as follows.

  \begin{plist} \itemsep-\parsep
  \item $\qcheck$ starts by computing $\chasehb{D}{\dep}$.
  \item $\qcheck$ nondeterministically guesses a squid decomposition $\delta =
    (Q^+,h,H,T)$ of $Q$ based on a set $V_\delta \subseteq \vars{h(Q^+)}$,
    where $H = \{\atom{a} \in h(Q^+) \mid \vars{\atom{a}} \subseteq
    V_\delta\}$ and $T$ is $[V_\delta]$-acyclic. Additionally,
    $\qcheck$ guesses
    a substitution $\theta_0: V_\delta \ra \adom{D}$
    % \adom{\chase{D}{\dep}}$,
    such that $\theta_0(H) \subseteq \chasehb{D}{\dep}$.  Note that this is an
    \textsc{np}~guess, because the number of atoms in $Q^+$ is at most twice
    the the number of atoms in $Q$.
%
  % and because we can assume without loss of generality that the extra
  % variables occurring in the \rcover{} $Q^+$ of $Q$ are all from some fixed
  % set $\{X_1,\ldots,X_r\}$, where $r=a\times |Q|$ and $a$ is the maximum
  % arity of a relation symbol in $\R$.
%
  \item $\qcheck$ checks whether $\theta_0$ can be extended to a homomorphism
    $\theta$ such that $\theta(T) \subseteq \chasefresh{D}{\dep}$.  By
    Lemma~\ref{lem:squid}, this is equivalent to check if $\chase{D}{\dep}
    \models Q$.  Such a $\theta$ exists iff for each connected subgraph $t$ of
    $\theta_0(T)$, there is a homomorphism $\theta_t$
    % that leaves all elements of $dom(D)$ unchanged,
    such that $\theta_t(t) \subseteq \chasefresh{D}{\dep}$.  The $\qcheck$
    algorithm thus identifies the connected components of $\theta_0(T)$.  Each
    such component is a $[dom(D)]$-acyclic conjunctive query, some of whose
    arguments may contain constants from $\adom{D}$.  Each such component can
    thus be represented as a $[\adom{D}]$-join tree $t$.  For each such join
    tree $t$, $\qcheck$ now tests whether there exists a homomorphism
    $\theta_t$ such that $\theta_t(t) \subseteq \chasefresh{D}{\dep}$.  This is
    done by the subroutine $\tcheck$, that takes the TGD set $\dep$, the
    database $D$, and a connected subgraph (i.e., a subtree) $t$ of
    $\theta_0(T)$ as input.  The inner workings of $\tcheck(D, \dep, t)$ are
    described below.
\item $\qcheck$ outputs ``true'' iff % \emph{all}
  the above check % \textit{(1)}, \textit{(2)} and
  \textit{(3)}
  gives a positive result.
\end{plist}

The correctness of $\qcheck$ follows from Lemma~\ref{lem:squid}.  Given that
step \textit{(2)} is nondeterministic, the complexity of $\qcheck$ is in
\textsc{np}$^X$, i.e., \textsc{np} with an oracle in $X$, where $X$ is a
complexity class that is sufficiently powerful for: \textit{(i)} computing
$\chasehb{D}{\dep}$, and \textit{(ii)} performing the tests $\tcheck(D, \dep,
t$).

\medskip

\noindent
We now describe the $\tcheck$ subroutine.

\ourpar{General notions.}  $\tcheck(D, \dep, t$) is obtained from $\acheck$ via
the following modifications.  Each configuration of $\tcheck$ maintains a
pointer $\tpoint$ to a vertex of $t$ (an atom $\atom{a}_q$).  Intuitively, this
provides a link to the root of the subtree of $t$ that still needs to be
matched by descendant configurations of $c$.
In addition to the data structures carried by each configuration of $\acheck$,
each configuration of $\tcheck$ also maintains an array $\subst$ of length $w$,
where $w$ is the maximum predicate arity in $\R$.  Informally, $\subst$ encodes
a substitution that maps the current atom of $t$ to (the canonicalized version
of) the current atom of $\chase{D}{\dep}$.
%
% Each array element of $\subst$ describes a substitution that replaces some
% element $X \in \adom{t} - \adom{D}$ of $t$ by some element from $\{X_1, X_2,
% \ldots\}$, where the $X_i $ are the new ``canonical'' elements dynamically
% generated by $\tcheck$ (see the description of $\acheck$, where the
% generation of the canonical elements is done in the same way).

% ****** \textbf{In the above, what are the ``canonical elements'' exactly? The
% $n_1$, $n_2$, ...? Then explain or, at least, use the same notation instead
% of $X_1$, ...} *******

$\tcheck$ works like $\acheck$, but instead of nondeterministically
constructing a main configuration path of the configuration tree such that
eventually some atom matches the query, it nondeterministically constructs a
main configuration (sub)tree $\tau$ of the configuration tree, such that
eventually all atoms of the join tree $t$ get consistently translated into some
vertices of $\tau$.  An important component of each main configuration $c$ of
$\tcheck$ is its \emph{current atom} $\atom{a}$.  Initially, $\atom{a}$ is some
nondeterministically chosen atom of $D$.  For subsequent configurations of the
alternating computation tree, $\atom{a}$ will take on nodes of $\gcf{D}{\dep}$.

% ****** \textbf{What labels??? You never mentioned any labels in the
% definition of gcf.} **********

\ourpar{Initialization.}  Similarly to $\acheck$, the computation starts by
generating an initial configuration
$(\atom{a},S,S,\prec,\star,\tpoint,\subst)$, where $\atom{a}$ is
nondeterministically chosen from the database $D$, $\tpoint$ points to the root
$\atom{r}$ of $t$, and $\subst$ is a homomorphic substitution that
$\subst(\atom{r}) = \atom{a}$, if $\atom{r}$ is homomorphically mappable on
$\atom{a}$; otherwise $\subst$ is empty.  This configuration will now be the
root of the main configuration tree.  

In general, the pointer $\tpoint$ of each main configuration
$c=(\atom{a},S,S',\prec,\star,\tpoint,\subst)$ points to some atom $\atom{a}_q$
of $t$, which has not yet been matched.  The algorithm attempts to expand this
configuration by successively guessing a subtree of configurations, mimicking a
suitable subtree of $\gcf{D}{\dep}$ that satisfies the subquery of $t$ rooted
at $\atom{a}_q$.

Whenever $\tcheck$ generates a further configuration, just as for $\acheck$,
$\tcheck$ generates via universal branching a number of configurations whose
joint task is to verify that all elements of $S$ are indeed provable. (We do not provide further details on how this branching is done.)

\ourpar{Expansion.}  The expansion of a main configuration $c =
(\atom{a},S,S',\prec,\star,\tpoint,\subst)$ works as follows.
For a configuration $c$, $\tcheck$ first checks whether there exists a
homomorphism $\mu$
% $\mu: \adom{\subst(\atom{a}_q)} \ra \adom{\atom{a}}$
such that $\mu(\subst(\atom{a}_q)) = \atom{a}$.  

\begin{enumerate}[label*=\arabic*.] \itemsep-\parsep
\item \textsl{($\mu$ exists.)}  If $\mu$ exists, we have two cases:
  \begin{enumerate}[label*=\arabic*.] \itemsep-\parsep
  \item If $\atom{a}_q$ is a leaf of $t$, then the current configuration turns
    into an accepting one.
  \item If % a suitable homomorphism $\mu$ exists and if
    $\atom{a}_q$ is not a leaf of $t$, then $\tcheck$ nondeterministically
    guesses whether $\mu$ is a \emph{good match}, i.e.,
    one that contributes to a global query answer and can be expanded to map
    the entire tree $t$ into $\gcf{D}{\dep}$.
    \begin{enumerate}[label*=\arabic*.] \itemsep-\parsep
    \item \textsl{(Good match).}  In case of a good match, $\tcheck$ branches
      universally and does the following {\em for each} child $\atom{a}_{qs}$
      of $\atom{a}_q$ in $t$.  It nondeterministically (i.e., via existential
      branching) creates a new configuration
      \[
      c_s = \can_{\atom{a}_s}(\atom{a}_s, S_s, S'_s, \prec_s, \star, \tpoint_s,
      \subst_s)
      \]
      where $\tpoint_s$ points to $\atom{a}_{qs}$, and where $\subst_s$ encodes
      the composition $\mu\circ \subst_s$. 
      % $\can_{\atom{a}_1}(\mu)$.
      The atom $\atom{a}_s$ is guessed, analogously to what is done in
      $\acheck$, by guessing some TGD $\rho \in \dep$ having body $\Phi$ and
      head atom $\atom{\psi}$, such that the guard atom $\atom{g}$ matches
      $\atom{a}$ via some homomorphism $\theta$ (that is, $\theta(\atom{g}) =
      \atom{a}$) and where $\theta(\Phi) \subseteq S'$.  The cloud subsets
      $S_s$ and $S'_s$ are chosen again as in $\acheck$.  Intuitively, here
      $\tcheck$, having found a good match of $\atom{a}_q$ on $\atom{a}$, tries
      to match the children of $\atom{a}_q$ in $t$ to children (and,
      eventually, descendants) of $\atom{a}$ in $\gcf{D}{\dep}$. Finally, the function 
$\can_{\atom{a}_s}$ indicates that appropriate canonizations are made to obtain $c_s$ from $c$ (we omit the tedious details).
    \item \textsl{(No good match).}  In case no good match exists, a child
      configuration 
      \[
      c_{new} = can_{\atom{a}_{new}}(\atom{a}_{new},S,S',\prec,\star,\tpoint,\subst)
      \]
      of $c$ is nondeterministically created, whose first component represents
      a child $\atom{a}_{new}$ of $\atom{a}$, and where $c_{new}$ inherits all
      of its remaining components from $c$.  Intuitively, after having failed
      at matching $\atom{a}_q$ (to which, we remind, $\tpoint$ points) to
      $\atom{a}$, $\tcheck$ attempts at matching the same $\atom{a}_q$ to some
      child of $\atom{a}$ in $\gcf{D}{\dep}$.  By analogy with the previous
      case, $a_{\mathit{new}}$ is obtained by guessing some TGD $\rho \in \dep$
      having body $\Phi$ and head atom $\atom{\psi}$, such that the guard atom
      $\atom{g}$ matches $\atom{a}$ via some homomorphism $\theta$ (that is,
      $\theta(\atom{g}) = \atom{a}$), $\theta(\Phi) \subseteq S'$, and where
      $\atom{a}_{new}:=\theta(\atom{\psi})$. Again, the function term
      $can_{\atom{a}_{new}}$ indicates that appropriate canonizations are
      applied (which we do not describe in detail).
    \end{enumerate}
  \end{enumerate}
  \item \textsl{($\mu$ does not exist.)}  In this case, $\tcheck$ proceeds
    exactly as in case 1.2.2, namely, it attempts at matching the same
    $\atom{a}_q$ to some child (or eventually some descendant) of $\atom{a}$ in
    $\gcf{D}{\dep}$.
%We again have two cases.
 %   \begin{enumerate}[label*=\arabic*.]
 % \item If $\atom{a}_q$ is a leaf of $t$, the configuration is rejecting.
  %  \item If $\atom{a}_q$ is not a leaf of $t$, then $\tcheck$ performs an
   %   existential branching analogous to that of~1.2.2, creating a child
   % configuration whose first component is a child $\atom{a}'$ of $\atom{a}$,
   %and the remaining components are inherited from $c$.
   %\end{enumerate}
  \end{enumerate}
  % Otherwise (in case of no match or no good match), the configuration is
  % either \textit{(a)} rejecting in case of a leaf of $t$,
  % or \textit{(b)} in case $\tpoint$ points to an inner node of $t$,

  \ourpar{Correctness.}  The correctness of $\tcheck$ can be shown along
  similar lines as for $\acheck$.  An important additional point to
  consider for $\tcheck$ is that, given that the query $t$ is acyclic, it is
  actually sufficient to remember at each configuration $c$ only the latest
  ``atom'' substitution $\subst$.  The correctness of $\qcheck$ then follows
  from the correctness of $\tcheck$ and from Lemma~\ref{lem:squid}.

  \ourpar{Computational cost.}  As for the complexity of $\qcheck$, note that
  in case the arity is bounded, $\tcheck$ runs in \textsc{apspace} =
  \textsc{exptime}, and computing $\chasehb{D}{\dep}$ is in \textsc{exptime} by
  Corollary~\ref{col:app-chasebottom}.  Thus, $\qcheck$ runs in time
  $\textsc{np}^\textsc{exptime} = \textsc{exptime}$.  In case of unbounded
  arities, both computing $\chasehb{D}{\dep}$ and running $\tcheck$ are in
  \textsc{2exptime}, therefore $\qcheck$ runs in time
  $\textsc{np}^{\textsc{2exptime}} = \textsc{2exptime}$. \hfill $\Box$

%%%%%%%%%%%%%%%%%%%%%%%%%%%%%%%%%%%%%%%%%%%%%%%%%%%%%%%%%%%%%%%%%%%%%%%%%%%%%
%% FINAL COMPLEXITY RESULT

By combining Theorems~\ref{theo:hardness} and~\ref{theo:upperbound} we
immediately get the following characterization for the complexity of reasoning
under weakly guarded sets of TGDs.

\begin{theorem}\label{theo:mainresult}
  Let $\dep$ be a weakly guarded set of TGDs on a schema $\R$, $D$ a database
  for $\R$, and $Q$ a Boolean conjunctive query.  Determining whether $D\cup
  \dep \models Q$ or, equivalently, whether $\chase{D}{\dep} \models Q$ is
  \textsc{exptime}-complete in case of bounded predicate arities, even if
  $\dep$ is fixed and $Q$ is atomic.  In the general case of unbounded
  predicate arities, the same problem is \textsc{2exptime}-complete.  The same
  completeness results hold for the problem of query containment under weakly
  guarded sets of TGDs.
\end{theorem}

\medskip

\noindent\textbf{Generalization.}
The definition of WGTGDs can be generalized to classes of TGDs whose unguarded
positions are guaranteed to contain a controlled finite number of null-values
only. Let $f$ be a computable integer function in two variables.  Call a
predicate position $\pi$ of a TGD set $\dep$ \emph{$f$-bounded} if no more than
$f(|D|,|\dep|)$ null values appear in $\chase{D}{\dep}$ as arguments in
position $\pi$; otherwise call $\dep$ \emph{$f$-unbounded}.  A set $\dep$ of
TGDs is \emph{$f$-weakly guarded} if each each rule of $\dep$ contains an atom
in its body that covers all variables which occur within this rule in
$f$-unbounded positions only. By a very minor adaptation of the proof of
Theorem~\ref{the:wgtgds-decidable}, it can be seen that CQ-answering for the
class of $f$-weakly guarded TGDs is decidable. Moreover, by a simple
modification of the $\qcheck$ and $\tcheck$ procedures, allowing a polynomial
number of nulls to enter ``unguarded" positions, it can be shown that
CQ-answering for fixed sets $\Sigma$ of W$\small^*$\!GTGDs is
% tractable,
\textsc{exptime}-complete in the worst case, where the class of
W$\small^*\!$GTGD sets is defined as follows.  A set $\dep$ of TGDs belongs to
this class if $\dep$ is $f$-weakly guarded for some function $f$ for which
there exists a function $g$, such that $f(|D|,|\dep|)|\leq |D|^{g(|\dep|)}$.

%%% Local Variables: 
%%% mode: latex
%%% TeX-master: "main"
%%% End: 

\section{Guarded TGDs}
\label{sec:guarded}

We now turn our attention to GTGDs.  We first consider the case of a variable
database $D$ as input.  Later, we prove part of the complexity bounds under the
stronger condition of fixed database.

\subsection{Complexity---Variable Database}
\label{sec:gtgds-complexity}

\begin{theorem}\label{theo:GTDTbottom}
  Let $\dep$ be a set of GTGDs over a schema $\R$ and $D$ be a database for
  $\R$.  Let, as before, $w$ denote the maximum predicate arity in $\R$ and
  $|\R|$ the total number of predicate symbols in $\R$.  Then:
  \begin{plist} \itemsep-\parsep
    % 1:
  \item Computing $\chasehb{D}{\dep}$ can be done in polynomial time if both
    $w$ and $|\R|$ are bounded and, thus, also in case of a fixed set $\dep$.
    The same problem is in \textsc{exptime} (and thus
    \textsc{exptime}-complete) if $w$ is bounded, and in \textsc{2exptime}
    otherwise.
    % 2:
  \item If $Q$ is an atomic or fixed Boolean query then checking whether
    $\chase{D}{\dep} \models Q$ is \textsc{ptime}-complete when both $w$ and
    $|\R|$ are bounded.  The same problem remains \textsc{ptime}-complete even
    in case $\dep$ is fixed.  This problem is \textsc{exptime}-complete if $w$
    is bounded and \textsc{2exptime}-complete in general.  It remains
    \textsc{2exptime}-complete even when $|\R|$ is bounded.
    % 3
  \item If $Q$ is a general conjunctive query, checking
    $\chase{D}{\dep} \models Q$ is \textsc{np}-complete in case both $w$ and
    $|\R|$ are bounded and, thus, also in case of a fixed set $\dep$.  Checking
    $\chase{D}{\dep} \models Q$ is \textsc{exptime}-complete if $w$ is
    bounded and \textsc{2exptime}-complete in general.  It remains
    \textsc{2exptime}-complete even when $|\R|$ is bounded.
    % 4
  \item BCQ answering under GTGDs is \textsc{np}-complete if both $w$ and
    $|\R|$ are bounded, even in case the set $\dep$ of GTGDs is fixed.
    % 5
  \item BCQ answering under GTGDs is \textsc{exptime}-complete if $w$ is
    bounded and \textsc{2exptime}-complete in general.  It remains
    \textsc{2exptime}-complete even when $|\R|$ is bounded.
  \end{plist}
\end{theorem}

\begin{proof} 
  First, note that items \textit{(4)} and \textit{(5)} immediately follow from
  the first three items, given that $\chase{D}{\dep}$ is a universal model. We
  therefore just need to prove items \textit{(1)}-\textit{(3)}. We first
  explain how the hardness results are obtained, and then deal with the
  matching membership results.

\smallskip

\noindent \ourpar{Hardness Results}.
The \textsc{ptime}-hardness of checking $\chase{D}{\dep} \models Q$ for atomic
(and thus also fixed) queries $Q$ and for fixed $\dep$ follows from the fact
that ground atom inference from a fixed fully guarded Datalog program over
variable databases is \textsc{ptime}-hard.  In fact, in the proof of
Theorem~4.4 in~\cite{voronkov-csur-2001} it is shown that fact inference from a
single-rule Datalog program whose body has a guard atom that contains all
variables is \textsc{ptime}-hard.
The \textsc{np}-hardness in item \textit{(3)} is immediately derived from the
hardness of CQ containment (which in turn is polynomially equivalent to query
answering) without constraints~\cite{ChMe77}.
The hardness results for \textsc{exptime} and \textsc{2exptime} are all derived
via minor variations of the proof of Theorem~\ref{theo:hardness}.  For example,
when $|\R|$ is unbounded and $w$ is bounded, the tape cells of the polynomial
worktape are simulated by using polynomially many predicate symbols. For
example, the fact that in configuration $v$ cell $5$ contains symbol 1 can be
encoded as $S^1_5(v)$.  We omit further details, given that a much stronger
hardness result will be established via a full proof in
Theorem~\ref{theo:icdt}.

% \smallskip

\noindent \ourpar{Membership results}. 
The membership results are proved exactly as those for weakly guarded sets of
TGDs, except that instead of using the concept of cloud, we now use a similar
concept of \emph{restricted cloud}, which coincides with that of a \emph{type}
of an atom in~\cite{CaGL12}.  The restricted cloud $\rcloud{D}{\dep}{\atom{a}}$
of an atom $\atom{a} \in \chase{D}{\dep}$ is the set of all atoms $\atom{b} \in
\chase{D}{\dep}$ such that $\adom{\atom{b}} \subseteq \adom{\atom{a}}$. By a
proof that is almost identical to the one of Theorem~\ref{theo:dependcloud}, we
can show that if $D$ is a database, $\dep$ a set of GTGDs, and if $\atom{a} \in
\chase{D}{\dep}$, then $\rtrootc \atom{a} =
\ggcf{\atom{a}}{\rcloud{D}{\dep}{\atom{a}}}$, where $\rtrootc{\atom{a}}$ is
defined as $\rtrootc{\atom{a}} = \set{\troot{\atom{a}}} \cup
\rcloud{D}{\dep}{\atom{a}}$. It follows that, for the main computational tasks,
we can use algorithms $\racheck$, $\rqcheck$, and $\rtcheck$, which differ from
the already familiar $\acheck$, $\qcheck$, and $\tcheck$ only in that
restricted clouds instead of the ordinary clouds are used. However, unlike the
case when both $|\R|$ and $w$ are bounded and a cloud (or subcloud) can have
polynomial size in $|D \cup \dep|$, a restricted cloud
$\rcloud{D}{\dep}{\atom{a}}$ has a constant number of atoms, and storing its
canonical version $\can_{\atom{a}}(\rcloud{D}{\dep}{\atom{a}})$ thus requires
logarithmic space only. In total, in case both $|\R|$ and $w$ are bounded, due
to the use of restricted clouds (and subsets thereof) each configuration $c$ of
$\racheck$ and of $\rtcheck$ only requires logarithmic space. Since
\textsc{alogspace} = \textsc{ptime}, the \textsc{ptime}-results for atomic
queries in items \textit{(1)} and \textit{(2)} follow.
Moreover, if both $|\R|$ and $w$ are bounded, for general (non-atomic and
non-fixed) queries, the $\rqcheck$ algorithm decides if $\chase{D}{\dep}
\models Q$ in \textsc{np} by guessing a squid decomposition (in
nondeterministic polynomial time) and checking (in
\textsc{alogspace}=\textsc{ptime}) if there is a homomorphism from this squid
decomposition into $\chase{D}{\dep}$. Thus, in this case, $\rqcheck$ runs in
$\textsc{np}^\textsc{ptime} = \textsc{np}$, which proves the \textsc{np} upper
bound of Item~\textit{(3)}.
If, in addition, $Q$ is fixed, then $Q$ has only a constant number of squid
decompositions, and therefore $\rqcheck$ runs in $\textsc{ptime}^\textsc{ptime}
= \textsc{ptime}$, which proves the \textsc{ptime} upper bound for fixed
queries mentioned in item~\textit{(2)}.
The \textsc{exptime} and \textsc{2exptime} upper bounds are inherited from the
same upper bounds for WGTGDs.
\end{proof}

Note that one of the main results in~\cite{JoK84}, namely, that query
containment under inclusion dependencies of bounded arities is
\textsc{np}-complete, is a special case of Item~\textit{(3)} of
Theorem~\ref{theo:GTDTbottom}.

%%%%%%%%%%%%%%%%%%%%%%%%%%%%%%%%%%%%%%%%%%%%%%%%%%%%%%%%%%%%%%%%%%%%%%%%%%%%%

\subsection{Complexity---Fixed Database}

The next result tightens parts of Theorem~\ref{theo:GTDTbottom} by showing that
the above \textsc{exptime} and \textsc{2exptime}-completeness results hold even
in case of a fixed input database.

\begin{theorem}\label{theo:icdt}
  Let $\dep$ be set of GTGDs on a schema $\R$.  As before, let $w$ denote the
  maximum arity of predicate in $\R$ and $|\R|$ be the total number of
  predicate symbols.  Then, for fixed databases $D$, checking whether
  $\chase{D}{\dep} \models Q$ is \textsc{exptime}-complete if $w$ is bounded
  and \textsc{2exptime}-complete for unbounded $w$.  For unbounded $w$, this problem remains
  \textsc{2exptime}-complete even when $|\R|$ is bounded.
\end{theorem}

\begin{proof} % (Sketch)
  First, observe that the upper bounds (i.e., the membership results for
  \textsc{exptime} and \textsc{2exptime}) are inherited from
  Theorem~\ref{theo:GTDTbottom}, so it suffices to prove the hardness
  results for the cases where $Q$ is a fixed atomic query.

  We start by proving that checking $\chase{D}{\dep} \models Q$ is
  \textsc{exptime}-hard if $w$ is bounded.
  It is well-known that $\textsc{apspace}$ (alternating $\textsc{pspace}$)
  equals $\textsc{exptime}$.  % In particular, as already noted  

  As already noted in the proof of Theorem~\ref{theo:hardness}, it is
  sufficient to simulate an \textsc{linspace} alternating Turing machine (ATM)
  $\M$ that uses at most $n$ worktape cells on every input (bit string) $I$ of
  size $n$, where the input string is initially present on the worktape.  In
  particular, we will show that $\M$ accepts the input $I$ iff $\chase{D}{\dep}
  \models Q$.

  Without loss of generality, we assume that \textit{(i)} ATM $\M$ has exactly
  one accepting state, $a$, which is also a halting state; \textit{(ii)} the
  initial state of $\M$ is an existential state; \textit{(iii)} $\M$ alternates
  at each transition between existential and universal states; and
  \textit{(iv)} $\M$ never tries to read beyond its tape boundaries.  

  Let $\M$ be defined as
  \(
  % \M = (S, A, \blank, \delta, s_0, s_a)
  \M = (S, \Lambda, \delta, q_0, \set{s_a})
  \),
  where $S$ is the set of states, $\Lambda = \{ 0, 1, \blank \}$ is the tape
  alphabet, $\blank \in \Lambda$ is the blank tape symbol, $\delta: S \times
  \Lambda \ra (S \times \Lambda \times \{\ell,r,\bot\})^2$ is the transition
  function ($\bot$ denotes the ``stay'' head move, while $\ell$ and $r$ denote
  ``left'' and ``right'' respectively), $q_0 \in S$ is the initial state, and
  $\set{s_a}$ is the singleton set of final (accepting) states.
  Since $\M$ is an alternating TM, its set of states $S$ is
  \emph{partitioned} into two sets, $S_\forall$ and $S_\exists$---universal and
  existential states, respectively.
  The general idea of the encoding is that the different {\em configurations}
  of $\M$ on input $I$ of length $n$ will be represented by fresh nulls that
  are generated in the construction of the chase.

  Let us now describe the schema $\R$.  First, for
  each integer $1\leq i\leq n$,
  $\R$ contains the predicate $\whead_i/1$, such that
  $\whead_i(c)$ be true iff at configuration $c$ the head of $\M$ is over
  the tape
  cell $i$.  $\R$ also has the predicates $\zero_i/1$, $\one_i/1$, and
  $\blankp_i/1$, where $\zero_i(c)$, $\one_i(c)$, and $\blankp_i(c)$ are true
  if in configuration $c$ the tape cell $i$ contains the symbol 0, 1, or
  $\blank$, respectively.
  Furthermore, for each state $s\in S$, $\R$ has a predicate $\state_s/1$, such
  that $\state_s(c)$ is true iff the state of configuration $c$ is $s$.  $\R$
  also contains: the predicate $\start/1$, where $\start(c)$ is true iff $c$ is
  the starting configuration; the predicate $\config/1$, which is true iff its
  argument identifies a configuration; and the predicate $\next/3$, where
  $\next(c,c_1,c_2)$ is true if $c_1$ and $c_2$ are the two successor
  configurations of $c$.  There are also predicates $\universal/1$ and
  $\existential/1$, such that $\universal(c)$ and $\existential(c)$ are true if
  $c$ is a universal (respectively, existential) configuration.
  Finally, there is a predicate $\accepting/1$, where $\accepting(c)$ is true
  only for accepting configurations $c$, and a propositional symbol $\accept$,
  which is true iff the Turing Machine $\M$ accepts the input $I$.

\smallskip

  We now describe a set $\dep(\M,I)$ of GTGDs that simulates the behavior of
  $\M$ on input $I$.  The rules of $\dep(\M,I)$ are as follows.

  \begin{enumerate} \itemsep-\parsep
  \item \textit{Initial configuration generation rules.}  The following rule
    creates an initial state: $\ra \exists X\, \initc(X)$.  We also add a rule
    $\initc(X) \ra \config(X)$, which says that the initial configuration is,
    in fact, a configuration.
  \item \textit{Initial configuration rules.}  The following set of rules
    encodes the tape content of the initial configuration, that is,  the input string $I$.  For each $1\leq i\leq n$, if the $i$-th
    cell of the tape contains a $0$, we add the rule $\initc(X) \ra
    \zero_i(X)$; if it contains a $1$, we add $\initc(X) \ra \one_i(X)$.
%; if it
 %   contains a $\blank$, we add $\initc(X) \ra \blankp_i(X)$.  
We also add the
 rule $\initc(X) \ra \existential(X)$ in order to say, without loss of
    generality, that the initial configuration is an existential one.
    Moreover, we add the rules $\initc(X) \ra \whead_1(X)$ and $\initc(X)\ra
    \state_{s_0}(X)$ to define the initial values of the state and the head
    position of $\M$ on input $I$.
  \item \textit{Configuration generation rules.}  We add a rule that creates
    two successor configuration identifiers for each configuration identifier.
    Moreover, we add rules stating that these new configuration identifiers
    indeed identify configurations:
    \begin{eqnarray*}
      \config(X) & \ra & \exists X_1,\!X_2\, \next(X,X_1,X_2),\\
      \next(X,Y,Z) & \ra & \config(Y),\\ 
      \next(X,Y,Z) & \ra & \config(Z). 
    \end{eqnarray*}
  \item \textit{Transition rules.}  We show by example how
    transition rules are generated for each
    transition in the finite control.
    Assume, for instance, that the transition table contains a specific
    transition of the form: $(s,0) \ra (\, (s1,1,r)\,,\, (s2,0,\ell)\,) $.
    Then we assert the following rules, for $1 \leq i \leq n$:
    \begin{eqnarray*}
      \whead_i(X), \zero_i(X), \state_s(X), \next(X,X_1,X_2) & \ra 
      & \state_{s_1}(X_1)\\
      \whead_i(X), \zero_i(X), \state_s(X), \next(X,X_1,X_2) & \ra 
      & \state_{s_2}(X_2).
    \end{eqnarray*}
    Moreover, for each $1 \leq i < n$ we have these rules:
    \begin{eqnarray*}
      \whead_i(X), \zero_i(X), \state_s(X), \next(X,X_1,X_2)& \ra & 
      \one_i(X_1)\\
      \whead_i(X), \zero_i(X), \state_s(X), \next(X,X_1,X_2)& \ra & 
      \whead_{i+1}(X_1),
    \end{eqnarray*}
    and for each $1 < i \leq n$ we add these rules:
    \begin{eqnarray*}
      \whead_i(X), \zero_i(X), \state_s(X), \next(X,X_1,X_2)& \ra & 
      \zero_i(X_2)\\
      \whead_i(X), \zero_i(X), \state_s(X), \next(X,X_1,X_2)& \ra & 
      \whead_{i-1}(X_2)
    \end{eqnarray*}
    The other types of transition rules are constructed analogously.
    Note that the total number of rules added is $6n$ times the
    number of transition rules. Hence it is linearly bounded by the size $n$ of
    the input string $I$ to $\M$.
  \item \textit{Inertia rules.}  These rules state that tape cells in positions
    not under the head keep their values.  Thus, for each 
    $1 \leq i,j \leq n$  such that $i \neq j$ we add the rules:
    \begin{eqnarray*}
      \whead_i(X), \zero_j(X), \next(X,X_1,X_2)& \ra & 
      \zero_j(X_1)\\
      \whead_i(X), \one_j(X), \next(X,X_1,X_2)& \ra & 
      \one_{j}(X_1)\\
      \whead_i(X), \blankp_j(X), \next(X,X_1,X_2)& \ra & 
      \blankp_{j}(X_1),
    \end{eqnarray*}
  \item \textit{Configuration-type rules.}  These rules say that the
    immediate successor configurations of an existential configuration are
    universal, and vice-versa:
    \begin{eqnarray*}
      \existential(X), \next(X,X_1,X_2)& \ra & \universal(X_1)\\
      \existential(X), \next(X,X_1,X_2)& \ra & universal(X_2)\\
      \universal(X), \next(X,X_1,X_2)& \ra & \existential(X_1)\\
      \universal(X), \next(X,X_1,X_2)& \ra & \existential(X_2).
    \end{eqnarray*}
  \item \textit{Acceptance rules.}  These recursive rules state when a
    configuration is accepting:
    \begin{eqnarray*}
      \state_{s_a}(X) & \ra &  accepting(X) \\
      \existential(X), \next(X,X_1,X_2), \accepting(X_1) 
      & \ra & \accepting(X)\\
      \existential(X), \next(X,X_1,X_2), \accepting(X_2) 
      & \ra &  \accepting(X)\\
      \universal(X), \next(X,X_1,X_2), \accepting(X_1),
      \accepting(X_2) &\ra &  \accepting(X)\\
      \initc(X), \accepting(X) & \ra & \accept.
    \end{eqnarray*}
  \end{enumerate}
  This completes the description of the set of TGDs $\dep(\M,I)$.  Note that
  this set is guarded, has maximum predicate arity~3, can be obtained in
  logarithmic space from $I$ and the constant machine description of $\M$. It
  faithfully simulates the behavior of the alternating linear space machine
  $\M$ in input $I$.  It follows that $\dep(\M,I) \models accept$ iff $\M$
  accepts input $I$. Let $D_0$ denote the empty database, and let $Q_0$ be the
  ground-atom query $\accept$.  We then have that $\dep(\M,I)\cup D_0 \models
  Q_0$ iff $\M$ accepts input $I$.  This shows that answering ground atom
  queries on fixed databases constrained by bounded arity GTGDs is
  \textsc{exptime}-hard.

  \medskip

  Let us now illustrate how we obtain the \textsc{2exptime} hardness result for
  guarded TGDs when arities are unbounded, but when the number $|\R|$ of
  predicate symbols of the schema $\R$ is bounded by a constant. 
  Given that \textsc{aexpspace}$=$\textsc{2exptime} (\textsc{aexpspace} $\equiv$
  alternating \textsc{aexpspace}), our aim is now to simulate an
  \textsc{aexpspace} Turing machine. It is sufficient to simulate one that uses no more than 
$2^n$ worktape cells, since the acceptance problem for such machines is already \textsc{2exptime}-hard.
In fact, by trivial padding arguments, the acceptance problem for every  \textsc{aexpspace}
machine can be transformed in polynomial time into the acceptance problem for one using at most 
$2^n$ worktape cells.

The problem is, however,  that now we can no longer construct a polynomial number of
  rules that explicitly address each worktape cell $i$, or each pair of cells
  $i,j$, since now there is an \emph{exponential} number of worktape cells.
  The idea now is to encode tape cell indexes as \emph{vectors} of symbols
  $(v_1,\ldots,v_k)$ where $v_i \in \set{0,1}$.  As in the proof of
Theorem~\ref{theo:hardness}, we could define, with a
  polynomial number of rules, a successor relation $\succp$ that stores pairs
  of indexes as $\succp(v_1,\ldots,v_k,w_1,\ldots,w_k)$.  However, there is a
  further difficulty: we now  have two different
  types of variables: the variables $V_i,W_j$ that range over the bits
  $v_i,w_i$ in the above-described bit vectors, and the variables $X,Y,Z$ that
  range over configurations.  A major difficulty is that, given that our rules
  are all guarded, we must make sure that these two types of variables,
  whenever they occur elsewhere in a rule body, also occur in some guard.
  To this end, we will use a fixed database $D_{01}$ that contains the single
  fact $\zeroone(0,1)$, and we will construct a ``guard'' relation $g$ such
  that for each vector $\vv$ of $n$ bits and its binary successor $\ww$, and
  for each configuration $x$ with its two successor configurations $y$ and $z$,
  the relation $g$ contains a tuple $g(\vv,\ww,x,y,z)$.  We will use several
  auxiliary relations to construct $g$.
  % A particular feature of these auxiliary relations is that each of them will
  % have, two extra arguments, represented by the variables $S_0$ and $S_1$,
  % which will be forced to take the values $0$ and $1$, respectively.  We do
  % this in order to have the constants $0$ and $1$ always at hand in rules
  % where such predicates appear.

  For technical reasons, the first two arguments of some atoms below will be
  dummy variables $T_0$ and $T_1$ that will always be forced to take the values
  $0$ and $1$, respectively. This way, where convenient, we will have the
  values $0$ and $1$ available implicitly in form of variables, and we will not need to
  use these constants explicitly in our rules.

  Given that our database is now non-empty, we do not need to create the
  initial configuration identifier via an existential rule as before.  We can
  simply take $0$ as the identifier of this initial configuration:
  $\zeroone(T_0,T_1)\, \ra\, \initc(T_0,T_1,T_0)$. (Here, the first two
  arguments of $\initc(T_0,T_1,T_0)$ just serve, as explained, to carry the
  values $0$ and $1$ along.)  We also add: $\initc(T_0,T_1,T_0)\, \ra\,
  \config(T_0,T_1,T_0)$ to assert that $0$ is the identifier of the initial
  configuration.
  Next we present the new configuration generation rules.
  \begin{eqnarray*}
    \config(T_0,T_1,X) & \ra & \exists Y,\exists Z\, \next(T_0,T_1,X,Y,Z),\\
    \next(T_0,T_1,X,Y,Z) & \ra & \config(T_0,T_1,Y),\\ 
    \next(T_0,T_1,X,Y,Z) & \ra & \config(T_0,T_1,Z). 
  \end{eqnarray*}

  We use further rules to create a relation $b$ such that each atom
  $b(0,1,\vv,x,y,z)$ contains a tuple for each vector $\vv$ of $n$ bits, and for
  each configuration $x$.  For better readability, whenever useful, we will use 
  superscripts for indicating the arity of vector
  variables: % or constants (e.g.,
  for instance, $\VV^{(n)}$ denotes $\dd{V}{n}$.
% and $\vv^{(n)}$, respectively).  
  Moreover, $\vz^{(j)}$ denotes the vector of $j$ zeros and $\vo^{(j)}$ the
  vector of $j$ ones.
  % $\ss_0^{(0)}$ and $\ss_1^{(0)}$ are empty lists.
%
  We start with the rule
\(
\next(T_0,T_1,X,Y,Z) \ra b(T_0,T_1,{\bf T_0}^{(n)},X,Y,Z),
\)
which defines an atom $b(0,1,\vz^{(n)},x,y,z)$, for each configuration $x$ and its
$\mathit next$-successors $y$ and $z$.  
%Notice that in this TGD, as well as in
%others below, the addition of the atom $\zeroone(0,1)$ in the body is %a
%technical trick that guarantees that all constants appearing in the %head also
%appear in the body.
  
The following $n$ rules, for $1 \leq i \leq n$, generate an exponential number
of new atoms, for each triple $X,Y,Z$, by swapping $0$s to $1$s in all possible
ways.  Eventually, the chase will generate all possible prefixes of $n$ bits.
\begin{multline*} 
%\[
b(T_0,T_1,U_1,\ldots, U_{i-1},T_0, U_{i+1},\ldots, U_n, X,Y,Z) \ra\\
b(T_0,T_1,U_1,\ldots,U_{i-1}, T_1,U_{i+1},\ldots, U_n, X,Y,Z).
% \]
\end{multline*}
  We are now ready to define the guard-relation $g$ through another
  group of guarded rules.
%
  % \noteac{In the following paragraph, $n$ was originally $k$, and has been
  % fixed.}
%
%  \marginpar{\tiny MK: should it be $0\leq r<n$??}
  For each $0 \leq r < n$, we add:
  \[b(T_0,T_1,\UU^{(r)},T_0,{\bf T_1}^{(n-r-1)},X,Y,Z) \ra
  g(\UU^{(r)},T_0,{\bf T_1}^{(n-r-1)},\UU^{(r)},T_1,{\bf T_0}^{(n-r-1)},X,Y,Z).
  \]
  The above $n$ rules define an exponential number of cell-successor pairs for
  each triple of configuration identifiers $X,Y,Z$, where $Y$ and $Z$ are the
  ``next'' configurations following $X$.  In particular, the relation $g$
  contains precisely all tuples $g(\vv,\ww,x,y,z)$, such that $\vv$ is an
  $n$-ary bit vector, % smaller than $1^{(k)}$,
  $\ww$ is its binary successor, $x$ is a configuration identifier, $y$ is its
  first successor via the $\mathit next$ relation, and $z$ is its second
  successor via the $\mathit next$ relation.

  We are now ready to simulate an \textsc{aexpspace} Turing machine $\M'$ over
  an input string $I$ by a set of GTGDs $\dep(\M',I)$.  Since this simulation
  is similar to the one presented in the first part of this proof, we just
  sketch it and point out the main differences.

  For the simulation, we use (in addition to the aforementioned auxiliary
  predicates) predicates similar to the ones used earlier for the simulation of
  the \textsc{exptime} Turing machine $\M$.  However, we only use a constant
  number of predicates.  So, rather than using, atoms $\whead_i(x)$,
  $\zero_i(x)$ and so on, we use their vectorized versions $\whead(\vv,x)$,
  $\zero(\vv,x)$ and so on, where $\vv$ is a bit vector of length $n$ that
  takes the role of an exponential index.  Thus, for example, the equivalent of
  the earlier rule
  \[
  \whead_i(X), \zero_i(X), \state_s(X), \next(X,X_1,X_2) \ra \one_i(X_1)
  \]
  is 
  \(
  g(\VV,\WW,X,X_1,X_2), \whead(\VV,X), \zero(\VV,X), \state(X,s) \ra
  \one(\VV,X_1).\) \ The earlier rule $$\whead_i(X), \zero_i(X), \state_s(X), \next(X,X_1,X_2) \ra \whead_{i-1}(X_2)$$
  becomes
  \(
  g(\VV,\WW,X,X_1,X_2), \whead(\WW,X), \zero(\WW,X), \state(X,s) \ra
  \whead(\VV,X_2).
  \)
  It is now straightforward to see how the initialization rules can be written.
  Informally, for copying the input string $I$ to the worktape, we place the
  $n$ input bits of $I$ on the tape by writing a rule for each such bit.  We
  then add rules that fill all positions from $n+1$ to $2^n$
  with blanks.  As this can be done in a similar way as in the second part of the proof of 
Theorem~\ref{theo:hardness}, we omit the details.

  The only remaining issue is the specification of the inertia rules.  These
  rules deal with \emph{pairs} $i,j$ of different, not necessary adjacent, tape
  cell positions in our earlier simulation.  Here we have only adjacent cell
  positions available so far.  The problem can be solved in different ways.
  One possibility is described below.
%
  % \noteac{Suggestion for better presentation: maybe it would be better to
  % explain one solution only, and then give the other more briefly as an
  % alternative.}

  We can simply modify the definition of the predicate $b$
  by adding a second vector of $n$ bits to the $b$-atoms so that
  $b$-atoms actually have the form $b(T_0,T_1,\vv,\uu,x,y,z)$, where $\vv$ and $\uu$
  range over all possible \emph{distinct pairs} of bit vectors of length $n$.
  This $\uu$ vector is then carried over to the $g$-atoms.  We can thus assume
  that the $g$-atoms now have the form $g(\vv,\ww,\uu,x,y,z)$.  The former
  inertia rule
  \(
  \whead_i(X), \zero_j(X), \next(X,X_1,X_2) \ra \zero_j(X_1)
  \)
  would then become
  \(
  g(\VV,\WW,\UU,X,X_1,X_2), \whead(\WW,X), \zero(\UU,X) \ra \zero(\UU,X_1).
  \)

%% ******** Why burden the already long proof with uninteresting alternatives???
%%  I propose to delete the second solution. It serves no purpose except to
%%  distract the reader. ******************

%%  \textsl{Second solution.}  Another way to represent the inertia rules
%%  is to define two new predicates $\whead^-(\vv,x)$ and $\whead^+(\vv,x)$, and
%%  add recursive rules that, starting at the neighbors of the actual head
%%  position, assert $\whead^-(\vv,x)$ for each position to the left of the head
%%  position of each configuration $x$, and $\whead^+(\vv,x)$ for each position
%%  to the right of the head position of configuration of each configuration $x$.
%%  We leave this as a simple TGD-programming exercise to the reader.  The
%%  inertia rules then look as follows:
%%%
%%  \[
%%  \begin{array}{c}
%%    g(\VV,\WW,X,X_1,X_2), \whead^-(\VV,X), \zero(\VV,X) \ra 
%%    \zero(\VV,X_1)\\
%%    g(\VV,\WW,X,X_1,X_2), \whead^+(\WW,X), \zero(\WW,X) \ra \zero(\WW,X_1),
%%  \end{array}
%%  \]
%%  and so on.

  What remains to be defined are the configuration and the acceptance rules.
  The configuration rules are very similar to the ones used in the previous
  reduction, hence we leave them as an exercise.  The acceptance
  rules are as follows:
  \begin{eqnarray*}
    \state(X,s_a) & \ra &  \accepting(X) \\
    \existential(X), g(\VV,\WW,X,X_1,X_2), \accepting(X_1) 
    & \ra &  \accepting(X)\\
    \existential(X), g(\VV,\WW,X,X_1,X_2), \accepting(X_2) 
    & \ra &  \accepting(X)\\
    \universal(X), g(\VV,\WW,X,X_1,X_2), \accepting(X_1),
    \accepting(X_2) &\ra &  \accepting(X)\\
    \zeroone(T_0,T_1),\accepting(T_0) & \ra & \accept.
  \end{eqnarray*}

  This completes the description of the set of TGDs $\dep(\M',I)$.  Note that
  this set is guarded and has a constant number of predicates.  It can be
  obtained in~\textsc{logspace} from $I$ and the constant machine description
  of $\M$. It also faithfully simulates the behavior of the alternating
  exponential space machine $\M'$ on input $I$.  It follows that $\dep(\M',I)
  \models \accept$ iff $\M'$ accepts input $I$.  Let $Q_0$ be the BCQ defined
  as $Q_0 = \set{\accept}$.  We then have $D_{01} \cup \dep(\M',I) \models Q_0$
  iff $\M'$ accepts input $I$.  This shows that answering ground atomic queries
  on fixed databases under guarded TGDs with a fixed number of predicate
  symbols (but unbounded arity) is \textsc{2exptime}-hard.
\end{proof}

%%% Local Variables: 
%%% mode: latex
%%% TeX-master: "main"
%%% End: 

\section{Polynomial Clouds Criterion}
\label{sec:polycloud}

In the previous section we have seen that, in case of bounded arity, query
answering under weakly guarded sets of TGDs is \textsc{exptime}-complete, while
query answering under GTGDs is \textsc{np}-complete.  Note that, for
unrestricted queries and databases, \textsc{np}-completeness is the best we can
obtain. In fact, even in the absence of constraints, the BCQ answering problem
is \textsc{np}-complete~\cite{ChMe77}.

In this section, we establish a criterion that can be used as a tool for
recognizing relevant cases where query answering is in \textsc{np} even for
weakly guarded sets of TGDs that are not fully guarded.  Note that we consider
both a setting where the weakly guarded set $\dep$ of TGDs is fixed and a
setting where \emph{classes} of TGD sets are considered. For these classes, we
require \emph{uniform} polynomial bounds.

\begin{definition}[Polynomial Clouds Criterion]\label{def:pcc}
  A fixed weakly guarded set $\dep$ of TGDs % of size $n$
  satisfies the \emph{Polynomial Clouds Criterion} (\emph{PCC}) if both of the
  following conditions hold:
  \begin{enumerate}\itemsep-\parsep
  \item There exists a polynomial $\pi(\cdot)$ such that for each database $D$,
    $|\clouds{D}{\dep}/_\simeq| \leq \pi(|D|)$. In other words, up to an
    isomorphism, there are only polynomially many clouds.
  \item There is a polynomial $\pi'(\cdot)$ such that, for each database
    $D$ and for each atom $\atom{a}$:
    \begin{itemize}\itemsep-\parsep
    \item if $\atom{a} \in D$ then $\cloud{D}{\dep}{\atom{a}}$ can be computed
      in time $\pi'(|D|, |\dep|)$, and
    \item if $\atom{a} \not\in D$ then $\cloud{D}{\dep}{\atom{a}}$ can be
      computed in time $\pi'(|D|, |\dep|)$ starting with $D$, $\atom{a}$,
      and $\cloud{D}{\dep}{\atom{b}}$, where $\atom{b}$ is the predecessor of
      $\atom{a}$ in $\gcf{D}{\dep}$.
    \end{itemize}
    We also say that $\dep$ satisfies the PCC with respect to $\pi$ and $\pi'$.
    Note that in the above, $|\dep|$ is constant and can be omitted.  However,
    the use of $|\dep|$ is justified by the following.  A class $\C$ of weakly
    guarded TGD sets satisfies the PCC if there are \emph{fixed} polynomials
    $\pi$ and $\pi'$ such that each TGD set in $\C$ satisfies the PCC
    \emph{uniformly} with respect to $\pi$ and $\pi'$ (i.e., each TGD set in
    this class has $\pi$, $\pi'$ as a bound).
    % ************ What does ``uniformly'' mean here?? That each TGD set in
    % $\C$ satisfies such PCC? Then say so directly and do not introduce a
    % new, useless notion.  *****************************
  \end{enumerate}
\end{definition}

\begin{theorem}\label{theo:pcc}
  Let $\dep$ be a fixed weakly guarded set of TGDs over a schema $\R$, such
  that $\dep$ enjoys the Polynomial Clouds Criterion.  Then:
  \begin{itemize}\itemsep-\parsep
  \item Deciding for a database $D$ and an atomic or fixed Boolean conjunctive
    query $Q$ whether $D \cup \dep\models Q$ (equivalently, whether
    $\chase{D}{\dep} \models Q$) is in \textsc{ptime}.
  \item Deciding for a database $D$ and a general Boolean conjunctive query $Q$
    whether $D \cup \dep \models Q$ (equivalently, $\chase{D}{\dep} \models Q$)
    is in \textsc{np}.
  \end{itemize}
\end{theorem}
%
% A proof sketch of this theorem is given in Appendix~\ref{sec:app-polycloud}.
\begin{proof}
  A polynomial algorithm $\acheckbis$ for atomic queries $Q$ works as follows.
  We start to produce the chase forest $\gcf{D}{\dep}$ using the standard
  chase.  In addition, immediately after generating any node $\atom{a}$ and its
  cloud $\cloud{D}{\dep}{\atom{a}}$ (in polynomial time), we will store
  $\can_{\atom{a}}(\atom{a},\cloud{D}{\dep}{\atom{a}})$ in a buffer, which we
  call \emph{cloud-store}.  Whenever a branch of the forest reaches a vertex
  $\atom{b}$ such that $\can_{\atom{b}}(\cloud{D}{\dep}{\atom{b}})$ is already
  in the cloud-store, further expansion of that branch $\atom{b}$ is blocked.
  Since there can be only a polynomial number of pairs
  $\can_{\atom{a}}(\atom{a},\cloud{D}{\dep}{\atom{a}})$, the algorithm stops
  after a polynomial number of chase steps, each step requiring only polynomial
  time.  Now, by Corollary~\ref{col:sim}, the cloud-store already contains all
  possible atoms of $\chase{D}{\dep}$ and their clouds, up to isomorphism.  To
  check whether $\chase{D}{\dep} \models Q$ holds for an atomic query $Q$, it
  is thus sufficient to test whether every atom $\atom{c}$ that occurs in the
  cloud-store matches $Q$.  In summary, $\acheckbis$ runs in \textsc{ptime}.

  The algorithm $\qcheckbis$ for conjunctive queries works just like $\qcheck$,
  except that it calls the algorithm $\tcheckbis$ as a subroutine instead of
  $\tcheck$.  The input to $\tcheckbis$ is $D$, $Q$, and also the cloud-store
  computed by $\acheckbis$.  We further assume that this cloud-store identifies
  each entry $e = \can_{\atom{a}}(\atom{a},\cloud{D}{\dep}{\atom{a}})$ by a
  unique integer $e_\#$ using $O(\log n)$ bits only.  $\tcheckbis$ is an
  alternating algorithm that works essentially like $\tcheck$, except for the
  following modifications:
  \begin{itemize}\itemsep-\parsep
  \item $\tcheck$ always guesses the full cloud $S =
    \cloud{D}{\dep}{\atom{a}}$, instead of possibly guessing a subcloud.
  In contrast,
    $\tcheckbis$ just guesses the entry number $e_\#$ of the
    corresponding entry $\can_{\atom{a}}(\atom{a}, \cloud{D}{\dep}{\atom{a}})$
    of the cloud-store.
  \item $\tcheckbis$ verifies correctness of
    the cloud guess in \textsc{alogspace} using
    $D$, $\atom{a}$, $e_\#$, as well as
    $\atom{b}$ and $e'_\#$, where $\atom{b}$ is the main atom of the predecessor
    configuration and $e'$ is the entry in the cloud-store featuring
    $\can_{\atom{b}}(\atom{b},\cloud{D}{\dep}{\atom{b}})$.  Note that
    such verification is effectively possible due to condition \textit{(2)}
    of Definition~\ref{def:pcc}.
  \item $\tcheckbis$ only needs to compute the main configuration tree---the
    one whose configurations contain $\star$. The algorithm does \emph{not}
    compute the auxiliary branches, since they are no longer necessary, as the
    correctness check $S$ is done in a different way.
  \item The configurations of $\tcheckbis$ do not need to guess or memorize
    linear orders $\prec$ and the set $S^+$.
  \end{itemize}
  Given that $\tcheckbis$ is an \textsc{alogspace} algorithm, $\qcheckbis$ is
  an $\textsc{np}^\textsc{alogspace}$ procedure.  Since
  $\textsc{np}^\textsc{alogspace} = \textsc{np}^\textsc{ptime} = \textsc{np}$,
  query answering is in \textsc{np}.  In case of a fixed conjunctive query $Q$,
  since $Q$ has a constant number of squid decompositions, $\qcheckbis$ runs in
  $\textsc{ptime}^\textsc{ptime} = \textsc{ptime}$.
\end{proof}

Note that the Polynomial Clouds Criterion is not syntactic.  Nevertheless,
it is useful for proving that query
answering for some weakly guarded sets TGDs is in \textsc{np}, or
even in polynomial time for atomic queries.  An application of this criterion
is illustrated in Section~\ref{sec:applications}.

%%%%%%%%%%%%%%%%%%%%%%%%%%%%%%%%%%%%%%%%%%%%%%%%%%%%%%%%%%%%%%%%%%%%%%%%%%%%%

The following is a direct corollary of
Theorem~\ref{theo:GTDTbottom}.

\begin{theorem} (1) Every set $\dep$ of GTGDs satisfies the PCC. (2) For
  any constant $k$, the class of all GTGD sets of arity bounded by $k$ satisfies
  the PCC.
\end{theorem}

The following result can be obtained by a minor adaptation of the proof of
Theorem~\ref{theo:pcc}.
\begin{theorem}\label{theo:restrictedQ}
  Let $\dep$ be a fixed weakly guarded set of TGDs that enjoys the Polynomial
  Clouds Criterion, and let $k$ be a constant. Then:
  \begin{plist}\itemsep-\parsep
  \item For a database $D$ and a Boolean conjunctive query of treewidth $\leq
    k$, deciding whether $D \cup \dep\models Q$ (equivalently, $\chase{D}{\dep}
    \models Q$) is in \textsc{ptime}.
  \item The same tractability result holds for acyclic Boolean conjunctive
    queries.
  \end{plist}
\end{theorem}

\smallskip

By analogy to the PCC, one may define various other criteria based on other
bounds.
% As an example, consider the Exponential Clouds Criterion defined for
% \emph{classes} of TGDs as follows.
In particular, we can define the \emph{Exponential Clouds Criterion (ECC)}
for classes of TGD sets, which we will use in the next section, as follows:

\begin{definition}[Exponential Clouds Criterion]\label{def:ecc}
  Let $\C$ be a class of weakly guarded TGD sets.  $\C$ satisfies the
  Exponential Clouds Criterion (ECC) if both of the following conditions are
  satisfied:
  \begin{enumerate}\itemsep-\parsep
  \item There is a polynomial $\pi(\cdot)$ such that for every database $D$ and
    any set of TGDs $\dep$ in $\C$ of size $n$,
    $|\clouds{D}{\dep}/_\simeq| \leq 2^{\pi(|D|+n)}$.
  \item There exists a polynomial $\pi'(\cdot)$ such that for every database
    $D$, any set of TGDs $\dep$ in $C$ of size $n$, and any atom $\atom{a}$:
    \begin{itemize}\itemsep-\parsep
    \item if $\atom{a} \in D$, then $\cloud{D}{\dep}{\atom{a}}$ can be computed
      in time $2^{\pi'(|D|+n)}$, and
    \item if $\atom{a} \not\in D$, then $\cloud{D}{\dep}{\atom{a}}$ can be
      computed in time $2^{\pi'(|D|+n)}$ from $D$, $\atom{a}$, % and $b$,
      and $\cloud{D}{\dep}{\atom{b}}$, where $\atom{b}$ is the predecessor of
      $\atom{a}$ in $\gcf{D}{\dep}$.
    \end{itemize}
  \end{enumerate}
\end{definition}

We have the following result on sets of TGDs enjoying the ECC:

\begin{theorem}\label{theo:ecc}
  If $\dep$ is a weakly guarded set of TGDs from a class $\C$ that enjoys the
  Exponential Clouds Criterion, then deciding for a database $D$ and a Boolean
  conjunctive query $Q$ (atomic or not) whether $D \cup \dep\models Q$ is in
  \textsc{exptime}.
\end{theorem}

\begin{proofsk} The proof is very similar to that for
  Theorem~\ref{theo:pcc}.
  The main difference is that \textsc{ptime} and \textsc{alogspace} are
  replaced by \textsc{exptime} and \textsc{apspace}, respectively.  We then get
  that query answering for atomic queries is in $\textsc{apspace} =
  \textsc{exptime}$, and that answering non-atomic queries is in
  $\textsc{np}^\textsc{apspace} =\textsc{np}^\textsc{exptime} =
  \textsc{exptime}$.  Thus, in this case, there is no difference between atomic
  and non-atomic query answering:  both are in \textsc{exptime}.
\end{proofsk}

%%% Local Variables: 
%%% mode: latex
%%% TeX-master: "main"
%%% End: 
 
\section{TGDs with Multiple-Atom Heads}
\label{sec:multiple-atoms}

% Before presenting our main results, we prove the following lemma, which will
% allow us to simplify the proofs in the rest of the paper. The lemma states
% that we can restrict our attention on TGDs that have singleton atoms in the
% head.

As mentioned in Section~\ref{sec:preliminaries}, all complexity results
proved so far for single-headed TGDs also carry over to the
general case, where multiple atoms may appear in rule heads.  We make this
claim  more
formal here.

\begin{theorem} \label{the:multiple-atoms} All complexity results derived in
  this paper for sets TGDs whose heads are single-atoms are equally valid for
  sets of multi-atom head TGDs.
\end{theorem}

\begin{proofsk} It suffices to show that the upper bounds carry
  over to the setting of TGDs with multiple-atom heads.  We exhibit a
  transformation from an arbitrary set of TGDs $\dep$ over a schema $\R$ to a
  set of single-headed TGDs $\dep'$ over a schema $\R'$ that extends $\R$ with
  some auxiliary predicate symbols.
  
  The TGD set $\dep'$ is obtained from $\dep$ by replacing each rule of the
  form $r: \bodyp(\vett{X}) \ra \headp_1(\vett{Y}), \headp_2(\vett{Y}), \ldots,
  \headp_k(\vett{Y})$, where $k > 1$ and $\vett{Y}$ is the set of all the
  variables that appear in the head, with
  the following set of rules:
  \begin{eqnarray*}
    \bodyp(\vett{X}) &\ra& V(\vett{Y})\\
    V(\vett{Y}) &\ra& \headp_1(\vett{Y})\\
    V(\vett{Y}) &\ra& \headp_2(\vett{Y})\\
    &\vdots\\
    V(\vett{Y}) &\ra& \headp_k(\vett{Y}),
  \end{eqnarray*}
  where $V$ is a fresh predicate symbol, having the same arity as the number of
  variables in $\vett{Y}$. Note that, in general, neither
  $\vett{Y}$ is contained in $\vett{X}$ not the other way around.
  It is easy to see that, except for the atoms of the form
  $V(\vett{Y})$, $\chase{D}{\dep}$ and $\chase{\dep'}{D}$ coincide.  The atoms
  of the form $V(\vett{Y})$ have completely new predicates and thus
  do not match any predicate symbol in the conjunctive query $Q$. Therefore, $\chase{D}{\dep} \models Q$
  iff $\chase{\dep'}{D} \models Q$.

  Obviously, $\dep'$ can be constructed in \textsc{logspace} from $\dep$.
  Therefore, the extension of our complexity results to the general case is
  immediate, \emph{except for the case of bounded arity}.  Notice that the
  arity of each auxiliary predicate in the above construction depends on the
  number of head-variables of the corresponding transformed TGD, which, in
  general, is not bounded.

  In case of bounded-arity WGTGDs, the \textsc{exptime} upper bound can still
  be derived by the above transformation by showing that the class of TGD sets
  $\dep'$ obtained by that transformation satisfies the Exponential Clouds
  Criterion of Section~\ref{sec:polycloud}.  To see that for each database $D$
  and each such $\dep'$ there is only an exponential number of clouds, notice
  that every ``large'' atom $V(\vett{Y})$ is derived by a rule with a ``small''
  weak guard $\atom{g}$ in its body, i.e., a weak guard $\atom{g}$ of bounded
  arity.  The cloud $\cloud{D}{\dep'}{\atom{g}}$ of this weak guard $\atom{g}$
  clearly determines everything below $\atom{g}$ in the guarded chase forest;
  in particular, the cloud of $V(\vett{Y})$.  Thus the set $\clouds{D}{\dep'}$
  of all clouds of all atoms is only determined by the clouds of atoms of
  bounded arity.  For immediately verifiable combinatorial reasons, there can
  be only singly-exponentially many such clouds.  This shows that
  $|\clouds{D}{\dep'}/_\simeq|$ is singly-exponentially bounded.  Therefore,
  the first condition of Definition~\ref{def:ecc} is satisfied.  It is not too
  hard to verify the second condition of Definition~\ref{def:ecc}, too.  Thus,
  query-answering based on bounded-arity WGTGDs is in \textsc{exptime}.  Given
  that GTGDs are a subclass of WGTGDs, the same \textsc{exptime} bound holds
  for bounded-arity GTGDs, as well.
 \nop{
    Let us now show that the second condition of Definition~\ref{def:ecc} is
    satisfied, too.  Towards this aim, let us transform the initial set of
    WGTGDs $\dep$ to yet another set $\dep"$ of single-headed TGDs.  $\dep"$ is
    obtained from $\dep$ as follows:.  For each rule $body(\vett{X}) \ra
    \exists \vett{Y} head(\vett{X'},\vett{Y})$ of $\dep$ (where the elements of
    $\vett{X'}$ are among those of $\vett{X}$) and for each atom
    $\atom{a}(\vett{X_i},\vett{Y_i},\vett{Y_{ij}})$ in
    $head(\vett{X'},\vett{Y})$, where both $\vett{Y_i}$ and $\vett{Y_{ij}}$ are
    made of elements of $\vett{Y}$, do the following:
    \begin{itemize}
    \item insert into $\dep"$ the rule $body(\vett{X}) \ra
      \atom{a}(\vett{X_i},\vett{Y_i},\vett{Y_{ij}})$, and
    \item for each other atom $\atom{b}(\vett{X_j},\vett{Y_j}, \vett{Y_{ij}})$
      in the same rule head, such that the common $\vett{Y}$-variables of the
      two atoms are precisely the variables in $\vett{Y_{ij}}$, insert into
      $\dep"$ the rule $body(\vett{X}) \ra \exists \vett{U_j},\exists
      \vett{Y_{ij}} \atom{b}(\vett{X_j},\vett{U_j}, \vett{Y_{ij}})$, where
      $\vett{U_j}$ are fresh variables replacing the variables $\vett{Y_j}$.
    \end{itemize}
    It is not hard to see that for each database $D$, and for rule $body\ra
    head$ of $\Sigma$, and for each atom $\atom{a}$ in $head$, }
\end{proofsk}

A completely different proof of the above theorem follows directly from the
results in~\cite{GoMP13} for the class of GTGDs, and from those
in~\cite{GoMP13a} for the class of WGTGDs.

%%%%%%%%%%%%%%%%%%%%%%%%%%%%%%%%%%%%%%%%%%%%%%%%%%%%%%%%%%%%%%%%%%%%%%%%%%%%%

%%% Local Variables: 
%%% mode: latex
%%% TeX-master: "main"
%%% End: 

\section{EGDs}
\label{sec:egds}

In this section we deal with equality generating dependencies (EGDs),
a generalization of \emph{functional dependencies}, which, in turn,
generalize \emph{key dependencies}~\cite{AbHV95}.

\begin{definition}\label{def:egd}\rm
  Given a relational schema $\R$, an EGD is a first-order formula of the form
  $\forall \vett{X} \Phi(\vett{X}) \ra X_\ell = X_k$, where $\Phi(\vett{X})$ is
  a conjunction of atoms over $\R$, and $X_\ell,X_k \in \vett{X}$.  Such a
  dependency is satisfied in an instance $B$ if, whenever there is a
  homomorphism $h$ that maps the atoms of $\Phi(\vett{X})$ to atoms of $B$, we
  have $h(X_\ell) = h(X_k)$.
\end{definition}

It is possible to ``repair'', or chase, an instance according to EGDs by
analogy with the chase based on TGDs.  We start by defining the EGD chase rule.

\smallskip

%\textsc{EGD Chase Rule.}  Consider an instance $B$, and an EGD $\eta$ of the
%form $\Phi(\vett{X}) \ra X_\ell = X_k$, where $X_\ell,X_k\in \vett{X}$.  The
%EGD $\eta$ is \emph{applicable} to $B$ if there is a homomorphism $h$ that maps
%the atoms of $\Phi(\vett{X},\vett{Y})$ to atoms of $B$ and $h(X_\ell) \neq
%h(X_k)$.  If $\eta$ is applicable and $X_\ell, X_k$ are two distinct elements
%in the set of constants $\dom$, then the application of the EGD yields a hard
%constraint violation, which in turn causes the \emph{failure} of the chase, and
%the halting of its computation.  In such a case, the result of the chase is an
%inconsistent theory.
%%
%If $\eta$ is applicable and its application does not cause a chase failure as
%above, the result of its application is the replacement of all occurrences of
%$h(X_\ell)$ in $B$ with $h(X_k)$, if $h(X_k)$ precedes $h(X_\ell)$ in the
%lexicographical order of constants and variables given in
%Section~\ref{sec:preliminaries}.  If $h(X_\ell)$ precedes $h(X_k)$, we replace
%all occurrences of $h(X_k)$ with $h(X_\ell)$. \markfull

\begin{definition}[EGD Applicability]\label{def:egd-applicability}
  Consider an instance $B$ of a schema $\R$, and an EGD $\eta$ of the form
  $\Phi(\vett{X})
  \ra X_i = X_j$ over $\R$.  We say that $\eta$ is \emph{applicable} to $B$ if
  there is a homomorphism $h$ such that $h(\Phi(\vett{X})) \subseteq B$ and
  $h(X_i) \neq h(X_j)$.
\end{definition}

\begin{definition}[EGD Chase Rule]
  Let $\eta$ be an EGD of the form $\Phi(\vett{X}) \ra X_i = X_j$ and suppose
  that it is applicable to an instance $B$ via a homomorphism $h$.  The result
  of the application of $\eta$ on $B$ with $h$ is a \emph{failure} if
  $\{h(X_i),h(X_j)\} \subset \dom$ (because of the unique name assumption).
  Otherwise, the result of this application is the instance $B'$ obtained from
  $B$ by replacing each occurrence of $h(X_j)$ with $h(X_i)$ if $h(X_i)$
  precedes $h(X_j)$ in lexicographical order.  If $h(X_j)$ precedes $h(X_i)$
  then the occurrences of $h(X_i)$ are replaced with $h(X_j)$ instead.  We
  write $B \stackrel{\eta,h}{\longrightarrow} B'$ to say that $B'$ is obtained
  from $B$ via a single EGD chase step.
\end{definition}

\begin{definition}[Chase sequence \wrt~TGDs and EGDs]
  Let $D$ be a database and $\dep = \tdep \cup \edep$, where $\tdep$ is a set
  of TGDs and $\edep$ is a set of EGDs.  A (possibly infinite) \emph{chase
    sequence} of $D$ with respect to $\dep$ is a sequence of instances $B_0,
  B_1, \ldots$ such that $B_i \stackrel{\sigma_i,h_i}{\lora} B_{i+1}$, where
  $B_0 = D$ and $\sigma_i \in \tdep \cup \edep$ for all $i \geq 0$.  A chase
  sequence is said to be \emph{failing} if its last step is a failure.  A chase
  sequence is said to be \emph{fair} if every TGD or EGD that is applicable at
  a certain step is eventually applied.
\end{definition}

In case a fair chase sequence happens to be finite, $B_0, \ldots, B_m$, and no
further rule application can change $B_m$, then the chase is well defined as
$B_m$, and is denoted by $\chase{D}{\dep}$.
% We will also use $\pchase{D}{\dep}{i}$ to denote $B_i$.
%
For our purposes, the order of application of TGDs and EGDs is irrelevant.  
In the following therefore, when saying ``the fair chase sequence'', we will
refer to any fair chase sequence, chosen according to some order of application
of the dependencies.  
% If the fair chase sequence of $D$ w.r.t.~$\dep$ does not
% fail, we say that $D \cup \dep$ is satisfiable; otherwise, we say that $D \cup
% \dep$ is unsatisfiable.

% dependencies, and that is is fair (that is, every TGD or EGD that is
% applicable at a certain step is eventually applied).

\smallskip

%%%%%%%%%%%%%%%%%%%%%%%%%%%%%%%%%%%%%%%%%%%%%%%%%%%%%%%%%%%%%%%%%%%%%%%%%%%%%

It is well-known (see~\cite{JoK84}) that EGDs cause problems when combined with
TGDs, because even for very simple types of EGDs, such as plain key
constraints, the implication problem for EGDs plus TGDs and the query answering
problem are undecidable.  This remains true even for EGDs together with GTGDs.
In fact, even though inclusion dependencies are fully guarded TGDs, the
implication problem, query answering, and query containment are undecidable
when keys are used as EGDs and inclusion dependencies as
TGDs~\cite{ChV85,Mitch83,CaLR03}.

Moreover, while the result of an infinite chase using TGDs is well-defined as
the limit of an infinite, monotonically increasing sequence (or, equivalently,
as the least fixed-point of a monotonic operator), the sequence of sets
obtained in the infinite chase of a database under TGDs and EGDs is, in
general, neither monotonic nor convergent.
Thus, even though we can define the chase procedure for TGDs plus EGDs, it is
not clear how the \emph{result} of an infinite chase involving both TGDs and
EGDs should be defined.  
%
% However, if the infinite chase converges, and it is fair (that is, every TGD
% or EGD that is applicable at a certain step is eventually applied), then the
% result is a universal solution, as shown in~\cite{NaDR06}.

For the above reasons, we cannot hope to extend the positive results for weakly
guarded sets of TGDs, or even GTGDs, from the previous sections to include
arbitrary EGDs.  Therefore, we are looking for suitable restrictions on EGDs,
which would allow us to: \textit{(i)} use the (possibly infinite) chase
procedure to obtain a query-answering algorithm, and \textit{(ii)} transfer the
decidability results and upper complexity bounds derived in the previous
sections to the extended formalism.

%%%%%%%%%%%%%%%%%%%%%%%%%%%%%%%%%%%%%%%%%%%%%%%%%%%%%%%%%%%%%%%%%%%%%%%%%%%%%

A class that fulfills both desiderata is a subclass of EGDs, which we call
\emph{innocuous} relative to a set of TGDs.  These EGDs enjoy the property that
query answering is % basically
insensitive to them, provided that the chase does not fail.  In other words,
when $\dep = \tdep \cup \edep$, where $\tdep$ is a set of TGDs, $\edep$ a set
of EGDs, and $\edep$ is innocuous relative to $\tdep$, we can simply ignore
these EGDs in a non-failing chase sequence.  This is possible because,
intuitively, such a non-failing sequence does not generate any atom that is not
entailed by $\chase{D}{\tdep}$.

More specifically, we start from the notion of \emph{innocuous application} of
an EGD.  Intuitively, when making two symbols equal, an innocuous EGD
application makes some atom $\atom{a}$ equal to some other existing atom
$\atom{a}_0$; this way, as the only consequence of the EGD application, the
original atom $\atom{a}$ is \emph{lost}, but no new atom whatsoever is
introduced.  The concept of innocuous EGD application is formally defined as
follows.

%% Previous version of the above paragraph follows.
%%
% More specifically, we start from the notion of \emph{innocuous application}
% of an EGD.  An innocuous application of an EGDs on an instance $B$ (with
% nulls), when making symbols equal, it makes some \emph{atoms} equal; this
% way, being an instance a set of atoms, $B$ \emph{loses} some atoms as the
% only consequence of the application.  This is formally defined as follows.
%
\begin{definition}[Innocuous EGD application]
  \label {def:smooth-egd}
  Consider a (possibly infinite) non-failing chase sequence $D = B_0, B_1,
  \ldots$, starting with a database $D$, with respect to a set $\dep = \tdep
  \cup \edep$, where $\tdep$ is a set of TGDs and $\edep$ is a set of EGDs.  We
  say that the EGD application $B_i \stackrel{\eta,h}{\lora} B_{i+1}$, where
  $\eta \in \edep$ and $i \geq 0$, is \emph{innocuous} if $B_{i+1} \subset
  B_i$.
  % Suppose that for a particular value $i \geq 0$, $B_i
  % \stackrel{\eta,h}{\lora}
  % B_{i+1}$, where $\eta \in \edep$.
  % We say that this EGD application is \emph{innocuous} if $B_{i+1} \subset
  % B_i$.
  % $\pchase{D}{\dep}{i+1} \subset \pchase{D}{\dep}{i}$.
\end{definition}

Notice that innocuousness is a \emph{semantic}, not syntactic, property.  It is
desirable to have innocuous EGD applications because such applications cannot
trigger new TGD applications, i.e., TGD applications that were not possible
before the EGD was applied.  
% Thus, such EGDs cannot be responsible for perpetuating an infinite chase
% process.

Given that it might be undecidable whether a set of dependencies from a certain
class guarantees innocuousness of all EGD applications, one can either give a
direct proof of innocuousness for a concrete set of dependencies, as we will do
in Section~\ref{sec:fll}, or define sufficient syntactic conditions that
guarantee innocuousness of EGD applications for an entire class of
dependencies, as done, e.g., in~\cite{CaGL12}.

\begin{definition}\label{def:smooth-deps}
  Let $\dep=\tdep \cup \edep$, where $\tdep$ is a set of TGDs and $\edep$ a set
  of EGDs, where $\dep = \tdep \cup \edep$.  $\edep$ is \emph{innocuous} for
  $\tdep$ if, for every database $D$ such that the fair chase sequence of $D$
  \wrt~$\dep$ is non-failing,
  % $D \cup \dep$ is satisfiable,
  each application of an EGD in such sequence of $D$ \wrt~$\dep$ is innocuous.
\end{definition}

\begin{theorem}\label{theo:egds}
  Let $\dep = \tdep \cup \edep$, where $\tdep$ is a set of TGDs and $\edep$ a
  set of EGDs that is innocuous for $\tdep$.  Let $D$ be a database such that
  the fair chase sequence of $D$ \wrt~$\dep$ is non-failing.
  % $D \cup \dep$ is satisfiable.  no fair chase sequence of $D$ with respect
  % to $\dep$ fails.
  Then $D \cup \dep \models Q$ iff $\chase{D}{\tdep} \models Q$.
\end{theorem}

\begin{proof}
  Consider the fair chase sequence $B_0, B_1, \ldots$ of $D = B_0$ in the
  presence of $\dep$, where $B_i \stackrel{\sigma_i,h_i}{\lora} B_{i+1}$ for $i
  \geq 0$ and $\sigma \in \tdep \cup \edep$.
  %% We say that we \emph{apply} $\tup{\sigma_i,h_i}$ to get $B_{i+1}$ from
  %% $B_i$.
  Let us define a modified chase procedure which we call the \emph{blocking
    chase}, denoted by $\blockchase{D}{\dep}$.  The blocking chase uses two
  sets: a set $C$ of blocked atoms and a set of (unblocked) atoms $A$.  When
  started on a database $D$ such that $D \models \edep$ (the case $D
  \not\models \edep$ is not possible as this implies an immediate chase
  failure), $C$ is initialized to the empty set ($C = \emptyset$) and $A$ is
  initialized to $D$.  After the initialization, the blocking chase attempts to
  apply the dependencies in $\tdep \cup \edep$ exactly in the same way as in
  the standard fair chase sequence, with the following caveats.  While trying
  an application of $\tup{\sigma_i,h_i}$:
  \begin{itemize}\itemsep-\parsep
  \item If $\sigma_i$ is a TGD, and if $h_i(\body{\sigma_i}) \cap C =
    \emptyset$, then apply $\tup{\sigma_i,h_i}$ and add the new atom generated
    by this application to $A$.
  \item If $\sigma_i$ is a TGD and $h_i(\body{\sigma_i}) \cap C \neq
    \emptyset$, then the application of $\tup{\sigma_i,h_i}$ is blocked, and
    nothing is done.
  \item If $\sigma_i$ is an EGD, then the application of $\tup{\sigma_i,h_i}$
    proceeds as follows.  Add to $C$ all the facts that in the standard chase
    disappear in that step (because $B_i \subseteq B_{i-1}$, due to the
    innocuousness), i.e., add to $C$ the set $B_i - B_{i-1}$.
    % $\pchase{D}{\dep}{i} - \pchase{D}{\dep}{i-1}$.
    Thus, instead of eliminating tuples from $A$, the blocking chase simply
    bans them from being used by putting them in $C$.
  \end{itemize}
  Note that, by the construction of $\blockchase{D}{\dep}$, whenever the block
  chase encounters an EGD $\sigma_i$, $\tup{\sigma_i,h_i}$ is actually
  applicable, so $\blockchase{D}{\dep}$ is well-defined.  Let us use $C_i$ and
  $A_i$ to denote the values of $C$ and $A$ at step $i$, respectively.
  Initially, $C_0 = \emptyset$ and $A_0 = D$ as explained before.  Observe that
  $\emptyset = C_0 \subseteq C_1 \subseteq C_2 \subseteq \cdots$ and $D = A_0
  \subseteq A_1 \subseteq A_2 \subseteq \cdots$ are monotonically increasing
  sequences that have least upper bounds $C^* = \cup_i C_i$ and $A^* = \cup_i
  A_i$, respectively.
  % $C^*=\bigcup_{i \geq 0} B_i$ and $A^* = \bigcup_{i \geq 0} A_i$,
  % respectively.
  Clearly, $(C^*, A^*)$ is the least fixpoint of the transformation performed
  by $\blockchase{D}{\dep}$ (with respect to component-wise set inclusion).

  Now, let $S$ be defined as $S = A^* - C^*$.  By the definition of $S$, we
  have: $S \models \dep$.  Moreover, there is a homomorphism $h$ that maps
  $\chase{D}{\tdep}$ to $S$.
  Note that $h$ is the limit homomorphism of the sequence $h_1, h_2, h_3,
  \ldots$ (these $h_i$s are the very homomorphisms used while computing the
  block chase), and can be defined as the set of all pairs $(x,y)$ such that
  there exists an $i \geq 0$ such that $h_i(h_{i-1}(\cdots h_1(x))) = y$ and
  $y$ is not altered by any homomorphism $h_j$ for $j>i$.  Note that for every
  instance $B$ that contains $D$, we have $B \models D$.  In particular,
  $S\models D$.  Putting everything together, we conclude that $S \models D
  \cup \dep$.

  It is also well-known (see~\cite{NaDR06}) that for any set of atoms $M$ such
  that $M \models S \cup \tdep$, there is a homomorphism $h_M$ such that
  $h_M(\chase{D}{\tdep}) \subseteq M$.
  Now assume $D \cup \dep \models Q$.  Then $S \models Q$ and, because $S
  \subseteq \chase{D}{\tdep}$, we also have that $\chase{D}{\tdep} \models Q$.
  Conversely, if $\chase{D}{\tdep}\models Q$, then there is a homomorphism $g$,
  such that $g(Q) \subseteq \chase{D}{\tdep}$.  Therefore, for any set of atoms
  $M$ such that $M \models D \cup \dep$, since $h_M(\chase{D}{\tdep}) \subseteq
  M$, we have $h_M(g(Q)) \subseteq M$.  The latter means that $M \models Q$.
\end{proof}

\smallskip

\begin{andrea}
  We now come to the problem of checking, given a database $D$ and a set $\dep
  = \tdep \cup \edep$, where $\tdep$ is set of WGTGDs and $\edep$ are EGDs
  innocuous for $\tdep$, whether the fair chase sequence (denoted $B_0, B_1,
  \ldots$) of $D$ \wrt~$\dep$
  fails. % that is, whether $D \cup \dep$ is unsatisfiable.
%
  % Let $B_0, B_1, \ldots$ be the fair chase sequence of $D$ \wrt~$\dep$.
  Consider an application $B_i \stackrel{\eta,h}{\lora} B_{i+1}$, with $\eta
  \in \edep$ of the form $\Phi(\vett{X}) \ra X_\ell = X_k$.  When this
  application causes the chase to fail, we have that $h(X_\ell)$ and $h(X_k)$
  are distinct values in $\adom{D}$.
  %%In this case we write $B_i\not\models^f \eta$.
  Notice that $B_j$ exists for $j \leq i$, while it does not exist for any $j >
  i$.

\begin{lemma}\label{lem:failure}
  Consider a database $D$ and a set of dependencies $\dep = \tdep \cup \edep$,
  where $\tdep$ is a weakly guarded set of TGDs and $\edep$ are EGDs that are
  innocuous for $\tdep$.  Then the fair chase sequence of $D$ \wrt~$\dep$ fails
  iff there is an EGD $\eta \in \edep$ of the form $\Phi(\vett{X}) \ra X_\ell =
  X_k$ and a homomorphism $h$ such that $h(\Phi(\vett{X})) \subseteq
  \chase{D}{\tdep}$, $h(X_\ell) \neq h(X_k)$, and $\set{h(X_\ell), h(X_k)}
  \subseteq \adom{D}$.
\end{lemma}

  \begin{proofsk}

%    \noindent\\
    \ifdirection % \textsl{(If).}
    Let $B_0, B_1, \ldots$ be the fair chase sequence of $D$ \wrt~$\dep$.
    First, it is not difficult to show that, since $\edep$ is innocuous
    relative to~$\tdep$, if the failure occurs at step $\ell$ then all EGD
    applications $B_i \stackrel{\sigma_i,h_i}{\lora} B_{i+1}$, such that
    $\sigma_i \in \edep$ and $i < \ell -1$, are innocuous (see a similar proof
    in~\cite{CaCF13}) in the sequence $B_0, \ldots, B_{\ell-1}$.  From this,
    the ``if'' direction follows straightforwardly.

    \onlyifdirection % \textsl{(Only if).}
    By assumption, $\eta$ fails at some $B_k$, $k \geq 1$.  Since applications
    of innocuous EGDs can only remove tuples from the chase, it is easily seen
    that, if $\eta$ is applicable to $B_k$ via an homomorphism $h$, then it is
    also applicable to $\chase{D}{\tdep}$ via the same homomorphism $h$, which
    settles the ``only-if'' part.
\end{proofsk}

%     \noindent
%     \textsl{(If).}  Let $\Xi \subseteq \edep$ be the set of EGDs $\xi$ of the
%     form $\Phi_\xi(\vett{X}) \ra X_\ell = X_k$ such that there is an
%     homomorphism $h_\xi$ such that $h(\Phi_\xi(\vett{X})) \subseteq
%     \chase{D}{\tdep}$, $h_\xi(X_\ell) \neq h_\xi(X_k)$ and
%     $\set{h_\xi(X_\ell), h_\xi(X_k)} \subseteq \adom{D}$.  For each of such
%     $\xi$, let $k_\xi$ be the lowest chase step such that the above condition
%     is satisfied. We therefore have $h(\Phi_\xi(\vett{X})) \subseteq
%     \pchase{D}{\tdep}{k_\xi}$, $h(X_\ell) \neq h(X_k)$ and $\set{h(X_\ell),
%     h(X_k)} \subseteq \adom{D}$.  Let $k_\eta$ be the lowest among all the
%     $k_\xi$s and $\eta\in \Xi$ be the corresponding EGD (if there is more
%     than one such EGD, take the lexicographically first one). Assume $\eta$
%     has the form $\Phi_\eta(\vett{X}) \ra X_\ell = X_k$.  Let us now consider
%     the chase sequence $D=B_0, B_1, \ldots$.  Note that every EGD application
%     equates two nulls or a null and a constant, so if $B_{m-1} \models
%     \Phi_\eta(\vett{X})$ then $B_m \models \Phi_\eta(\vett{X})$, for any $m$.
% %
%     Therefore there must be some $k \geq k_\eta$ such that $B_k$ exists,
%     $h(\Phi_\eta(\vett{X})) \subseteq B_k$, $h(X_\ell) \neq h(X_k)$, and
%     $\set{h(X_\ell), h(X_k)} \subseteq \adom{D}$, which settles the if-part.

%   ******** In the last sentence of the proof: why is all this crap needed???
%   It follows directly from the previous that $h_\eta(\Phi_\eta(\vett{X}))
%   \subseteq B_{k_\eta}$, etc.  *********

\begin{theorem} \label{the:failure-check} Consider a database $D$ and a set of
  dependencies $\dep = \tdep \cup \edep$, where $\tdep$ are GTGDs (resp.,
  WGTGDs) and $\edep$ are EGDs that are innocuous for $\tdep$.  Checking
  whether 
  % $D \cup \dep$ is satisfiable
  the fair chase sequence of $D$ \wrt~$\dep$ fails is decidable, and has the
  same complexity as query answering for GTGDs (resp., WGTGDs) alone.
\end{theorem}

\begin{proofsk}
  Let $\mathit{neq}$ be a new binary predicate, which will serve as inequality.
  The extension of $\mathit{neq}$ is defined as $\adom{D} \times \adom{D} -
  \{(d,d) ~\mid~ d\in\adom{D}\}$ and can be constructed in time quadratic in
  $|\adom{D}|$.  Now, for every EGD $\eta$ of the form $\Phi(\vett{X}) \ra X_1
  = X_2$, where $X_1, X_2 \in \vett{X}$, we define the following Boolean
  conjunctive query (expressed as a set of atoms): \( Q_\eta = \Phi(\vett{X})
  \cup \set{\mathit{neq}(X_1,X_2)} \).  Since, by construction, no new facts of
  the form $\mathit{neq}(\sigma_1,\sigma_2)$ are introduced in the chase, it is
  immediate to see, from Lemma~\ref{lem:failure}, that at least one of the
  above $Q_\eta$ has a positive answer if and only if the fair chase sequence
  of $D$ with respect to $\dep$ fails.
  By Theorem~\ref{theo:egds}, answering the query $Q_\eta$ can be done with
  respect to the chase by $\tdep$ alone, which is decidable.
\end{proofsk}

\smallskip

Let $\dep=\tdep \cup \edep$ be as in the above theorem, $D$ be a database, and
let $Q$ be a query.  By the above theorem, we can check $\dep \cup D \models Q$
with the help of the following algorithm:

\begin{enumerate}\itemsep-\parsep
\item check whether % $D \cup \dep$ is satisfiable
  the fair chase sequence of $D$ \wrt~$\dep$ fails with the algorithm described
  in Theorem~\ref{the:failure-check};
\item if % $D \cup \dep$ is unsatisfiable (
  the fair chase sequence of $D$ \wrt~$\dep$ fails, then return \emph{``true''}
  and halt;
\item if $D \cup \tdep \models Q$ then return \emph{``true''}; otherwise return
  \emph{``false''}.
\end{enumerate}

This gives us the following corollary:

\begin{corollary}\label{cor:egds}
  Answering general conjunctive queries under weakly guarded sets of TGDs and
  innocuous EGDs is \textsc{ptime} reducible to answering queries of the same
  class under a weakly guarded sets of TGDs alone, and thus has the same
  complexity.
\end{corollary}
\end{andrea}

%%%%%%%%%%%%%%%%%%%%%%%%%%%%%%%%%%%%%%%%%%%%%%%%%%%%%%%%%%%%%%%%%%%%%%%%%%%%%

%%% Local Variables:
%%% mode: latex
%%% TeX-master: "main"
%%% End:

\section{Applications}
\label{sec:applications}

In this section we discuss applications of our results on weakly guarded sets of
TGDs to Description Logic languages and object-oriented logic languages.

\subsection{DL-Lite}
\label{sec:dl-lite}

DL-Lite~\cite{CDLL*07,ACKZ09} is a prominent family of ontology languages that
has tractable query answering.  Interestingly, a restriction of GTGDs called
\emph{linear TGDs} (which have exactly one body-atom and one head-atom)
properly extends most DL-Lite languages, as shown in~\cite{CaGL12}.  The
complexity of query answering under linear TGDs is lower than that of GTGDs,
and we refer the reader to~\cite{CaGL12} for more details.

Furthermore, in~\cite{CaGL12} it is also shown that the language of GTGDs
properly extends the description logic $\E\L$ as well as its extension
$\E\L^f$, which allows inverse and functional roles.  The fact that TGDs
capture important DL-based ontology languages confirms that TGDs are useful
tools for ontology modeling and querying.

\subsection{\fll}
\label{sec:fll}

\fll{} is an expressive subset of F-logic~\cite{flogic-new}, a well-known
formalism introduced for object-oriented deductive languages.  We refer the
reader to~\cite{cali-kifer-06} for details about \fll.  Roughly speaking,
compared to full F-Logic, \fll{} excludes negation and default inheritance, and
allows only a limited form of cardinality constraints.
\fll{} can be encoded by a set of twelve TGDs and EGDs, below, which
we denote by
$\flldep$:
\begin{plist}\itemsep-\parsep
\item[$\rho_{1}$:] $\type(O,A,T), \data(O,A,V) \ra \member(V,T)$.
\item[$\rho_{2}$:] $\fsub(C_1,C_3), \fsub(C_3,C_2) \ra \fsub(C_1,C_2)$.
\item[$\rho_{3}$:] $\member(O,C), \fsub(C,C_1) \ra \member(O,C_1)$.
\item[$\rho_{4}$:]
  $\data(O,A,V), \data(O,A,W), \funct(A,O) \ra V=W$.\\
  Note that this is the only EGD in this axiomatization.
\item[$\rho_{5}$:]
  % $\forall O,A~\exists V~\data(O,A,V) \la \mandatory(A,O)$.\\
  $\mandatory(A,O) \ra \exists V\,\data(O,A,V)$.\\
  Note that this TGD has an existentially quantified variable in the head.
\item[$\rho_{6}$:] $\member(O,C), \type(C,A,T) \ra \type(O,A,T) $.
\item[$\rho_{7}$:] $\fsub(C,C_1), \type(C_1,A,T) \ra \type(C,A,T)$.
\item[$\rho_{8}$:] $\type(C,A,T_1), \fsub(T_1,T) \ra \type(C,A,T)$.
\item[$\rho_{9}$:] $\fsub(C,C_1), \mandatory(A,C_1) \ra \mandatory(A,C)$.
\item[$\rho_{10}$:] $\member(O,C), \mandatory(A,C) \ra \mandatory(A,O)$.
\item[$\rho_{11}$:] $\fsub(C,C_1), \funct(A,C_1) \ra \funct(A,C)$.
\item[$\rho_{12}$:] $\member(O,C), \funct(A,C) \ra \funct(A,O)$.
\end{plist}
The results of this paper apply to the above set of constraints,
since $\flldep$ is a weakly guarded set, and the single
EGD~$\rho_4$ is innocuous.  The innocuousness of~$\rho_4$ is shown by
observing that, whenever the EGD is applied, it turns one atom into another;
moreover, all new $\data$ atoms created in the chase (see rule $\rho_5$)
have new labeled nulls exactly in the position $\data[3]$, where the symbols to
be equated also reside.

We now prove the relevant complexity results.  We start by showing that BCQ
answering under \fll{} is \textsc{np}-complete.  

% Notice that we cannot rely on the fact that BCQ answering \emph{without}
% constraints is already \textsc{np}-complete, because here we have a fixed
% schema. In principle, this can cause the complexity of BCQ answering to be
% higher or even lower (take, for instance, the case of a schema with a single
% unary predicate) than that of the corresponding constraint-free problem.

% ****** The last sentence about makes no sense. How can a fixed schema be
% harder than variable schema? ******

%
\begin{theorem}\label{the:fll-upper}
  Conjunctive query answering under \fll{} rules is \textsc{np}-hard.
\end{theorem}

\begin{proofsk}
  The proof is by reduction from the $3$-\textsc{colorability} problem.  Encode
  a graph $G = (V,E)$ as a conjunctive query $Q$ which, for each edge $(v_i,
  v_j)$ in $E$, has two atoms $\mathsf{data}(X,V_i,V_j)$ and
  $\mathsf{data}(X,V_j,V_i)$, where $X$ is a unique variable.  Let $D$ be the
  database $D= \{ \data(o,r,g)$, $\data(o,g,r)$, $\data(o,r,b)$,
  $\data(o,b,r)$, $\data(o,g,b)$, $\data(o,b,g) \}$.  Then, $G$ is
  three-colorable iff $D \models Q$, which is the case iff $D \cup \flldep
  \models Q$.  The transformation from $G$ to $(Q,D)$ is obviously polynomial,
  which proves the claim.
\end{proofsk}

\begin{theorem}\label{the:fll-lower}
  Conjunctive query answering under \fll{} rules is in \textsc{np}.
\end{theorem}

\begin{proofsk}
  As mentioned before, we can ignore the only EGD in $\flldep$ since, being
  innocuous, it does not interfere with query answering.
  % details are found in the full version \cite{CaGK08}.
  Let $\flldep'$ denote the set of TGDs resulting from $\flldep$ by
  eliminating rule $\rho_4$, i.e., let $\flldep' = \flldep - \set{\rho_4}$.
  To establish membership in \textsc{np}, it is sufficient to show that:
  % \begin{plist}\itemsep-\parsep
  %\item 
  \textit{(1)} $\flldep'$ is weakly guarded;
  % \item
  \textit{(2)} $\flldep'$ enjoys the PCC (see Definition~\ref{def:pcc}).
%  \end{plist}
Under the above condition, the membership in \textsc{np} can be proved by
exhibiting the following. \textit{(i)} An algorithm, analogous to $\acheck$,
that constructs \emph{all} ``canonical'' versions of the atoms of the chase and
their clouds (which are stored in a ``cloud store''), in polynomial time. Then
the algorithm should check whether an atomic (Boolean) query is satisfied by an
atom in the cloud store.  \textit{(ii)} An algorithm, analogous to $\qcheck$,
that guesses (by calling an analogous version of $\tcheck$) entire clouds
through guessing the cloud index (a unique integer) in the cloud store. Then
the algorithm should check, in alternating logarithmic space
(\textsc{alogspace}), the correctness of the cloud guess. In that check, it can
use only the cloud of the main atom of the predecessor configuration.  The
complexity of running this algorithm is shown to be
$\textsc{np}^{\textsc{alogspace}} = \textsc{np}$.

  \textit{(1)} is easy: the affected positions are $\data[3]$, $\member[1]$,
  $\type[1]$, $\mandatory[2]$, $\funct[2]$ and $\data[1]$.  It is easy to see
  that every rule of $\flldep'$ is weakly guarded, and thus $\flldep$ is weakly
  guarded.

  Now let us sketch \textit{(2}).  We need to show that $\flldep'$ satisfies
  the two conditions of Definition~\ref{def:pcc}.  We prove that the first
  condition holds for $\flldep'$ as follows.  Let $\flldepfull = \flldep' -
  \{\rho_5\}$.  These are all full TGDs (no existentially-quantified variables)
  and their application does not alter the domain.  We have
  $\chase{D}{\flldep'} = \chase{\chase{D}{\flldepfull}}{\flldep'}$.
  Let us now have a closer look at $D^+ = \chase{D}{\flldepfull}$.  Clearly,
  $\adom{D^+} = \adom{D}$.  For each predicate symbol $p$, let $\Rel(p)$ denote
  the relation consisting of all the $p$-atoms in $D^+$.  Let $\Omega$ be the
  family of all the relations that can be obtained from any of the relations
  $\Rel(p)$ by performing an arbitrary selection followed by some projection
  (we forbid disjunctions in the selection predicate).  For example, let $c,d
  \in \adom{D}$. Then $\Omega$ will contain the relations
  $\pi_{1,2}(\select{1=c} \Rel(\data))$, $\pi_{2}(\select{1=d\wedge 3=c}
  \Rel(\data))$, and so on, where the numbers represent the attributes to which
  selection is applied.  Given that $D^+$ is of size polynomial in $D$ and that
  the maximum arity of any relation $\Rel(p)$ is~$3$, the set $\Omega$ is of
  size polynomial in $D^+$ and thus polynomial in $D$.
  % \marginpar{\tiny MK: what does preservation mean here??}
  It can now be shown that $\Omega$ is preserved in a precise sense, when going
  to the final result $\chase{D^+}{\flldep'}$: for each relation $\Rel'(p)$
  corresponding to predicate $p$ in the final chase result, when performing a
  selection on values outside of $\adom{D}$ and projecting on the columns not
  used in the selection, the set of all tuples of $\adom{D}$-elements in the
  result is a relation in $\Omega$.  For example, if $v_5$ is a labeled null,
  then the set of all $T \in \adom{D}$, such that $\mathsf{member}(v_5,T)$ is
  an element of the final result, is a relation in $\Omega$.  Similarly, if
  $v_7$ and $v_8$ are new values, the set of all values $A$, such that
  $\mathsf{data}(v_7,A,v_8)$ is in the chase, is a relation in $\Omega$.  From
  this it follows that $\flldep'$ satisfies~\textit{(2)}.  In fact, all
  possible clouds are determined by the polynomially many ways of choosing at
  most three elements of $\Omega$ for each predicate.  The proof of the
  preservation property can be done by induction on the $i$-th new labeled null
  added.  Roughly, for each such labeled null, created by rule $\rho_5$, we
  just analyze which sets of values (or tuples) are attached to it via rules
  $\rho_4$, then $\rho_6$, $\rho_7$, $\rho_8$, $\rho_{10}$, and so on, and
  conclude that these sets were already present at the next lower level, and
  thus, by induction hypothesis, are in $\Omega$.

  The second condition of Definition~\ref{def:pcc} is proved
  by similar arguments.
\end{proofsk}

From Theorems~\ref{the:fll-upper} and~\ref{the:fll-lower} we immediately get
the following result.

\begin{corollary}\label{cor:fll-complexity}
  Conjunctive query answering under \fll{} rules is \textsc{np}-complete for
  general conjunctive queries, and in \textsc{ptime} for fixed-size or atomic
  conjunctive queries.
\end{corollary}

%%% Local Variables: 
%%% mode: latex
%%% TeX-master: "main"
%%% End: 

\section{Conclusions and Related Work}
\label{sec:conclusions}

In this paper we identified a large and non-trivial class of
\emph{tuple-generating} and \emph{equality-generating} dependencies for which
the problems of conjunctive query containment and answering are decidable, and
provided the relevant complexity results.  Applications of our
results % include
span databases and knowledge representation.  In particular, we have shown that
this class of constraints subsumes the classical work by Johnson and
Klug~\cite{JoK84} as well as more recent results~\cite{cali-kifer-06}.
Moreover, we are able to capture relevant ontology formalisms in the
Description Logics (DL) family, in particular DL-Lite and $\E\L$.

% \paragraph{Related work.}  %
The problem of query containment for non-terminating chase was addressed in the
database context in~\cite{JoK84}, where the ontological theory contains
inclusion dependencies and key dependencies of a particular form.  The
introduction of the \emph{DL-Lite} family of description logics
in~\cite{CDLL*07,ACKZ09} was a significant leap forward in ontological query
answering due to the expressiveness of DL-Lite languages and their tractable
data complexity.  Conjunctive query answering in DL-Lite has the advantage of
being \emph{first-order rewritable}, i.e., any pair $\tup{Q,\dep}$, where $Q$
is a CQ and $\dep$ is a DL-Lite ontology (TBox), can be rewritten as a
first-order query $Q_\dep$ such that, for every database (ABox) $D$, the answer
to $Q$ against the logical theory $D \cup \dep$ coincides with the answer to
$Q_\dep$ against $D$.  Since each first-order query can be written in SQL, in
practical terms this means that a pair $\tup{Q,\dep}$ can be rewritten as an
SQL query over the original database $D$.

Rewritability is widely adopted in ontology querying.  The works
\cite{CCDL01,CaLR03a} present query rewriting techniques that deal with
Entity-Relationship schemata and inclusion dependencies, respectively.
The work~\cite{PeMH10} presents a Datalog rewriting algorithm for the
expressive DL $\E\L\H\I\O^\neg$ , which comprises a limited form of concept and
role negation, role inclusion, inverse roles, and nominals, i.e., concepts that
are interpreted as singletons. Conjunctive query answering in $\E\L\H\I\O^\neg$
is \textsc{ptime}-complete in data complexity, and the proposed algorithm is
also optimal for other ontology languages such as DL-Lite.
Optimizations of rewriting under linear TGDs (TGDs with exactly one atom in the
body) are presented in~\cite{GoOP11,OrPi11}.  In~\cite{GoSc12} it is shown that
the rewriting of a conjunctive query under a set of linear TGDs can be of
polynomial size in the query and the TGD set.

Other rewriting techniques for \textsc{ptime}-complete languages (in data
complexity) have been proposed for the description logic
$\E\L$~\cite{Rosa07,LuTW09,KrRu07}.
Another approach worth mentioning is a combination of rewriting 
% according to the ontology 
and of the chase % according to the data
(see~\cite{KLTW*10}); this technique was introduced for DL-Lite in order to
tackle the performance problems that arise when the rewriting according to the
ontology is too large.

Recent works concentrate on semantic characterization of sets of TGDs under
which query answering is decidable~\cite{BLMS11}.  The notion of first-order
rewritability is tightly connected to that of \emph{finite unification set}
(FUS).  A FUS is semantically characterized as a set of TGDs that enjoy the
following property: for every conjunctive query $Q$, the rewriting $Q_\dep$ of
$Q$ obtained by backward-chaining through unification, according to the rules
in $\dep$, terminates.  Another semantic characterization of TGDs is that of
\emph{bounded treewidth set} (BTS), i.e., a set of TGDs such that the chase
under such TGDs has bounded treewidth.  As seen in
Section~\ref{sec:decidability}, every weakly guarded set of TGDs is a BTS.  A
\emph{finite expansion set} (FES) is a set of TGDs that guarantees, for every
database, the termination of the \emph{restricted} chase, and therefore the
decidability of query answering.

% The interesting classes of \emph{frontier guarded (FGTGDs)} and \emph{weakly
% frontier-guarded TGDs (WFGTGDs)} were considered and studied
% in~\cite{BLMS11,BMRT11,KrRu11}.  The idea underlying these classes is that,
% to obtain decidability, it is sufficient to guard only \emph{frontier
% variables}, that is, variables that occur both in the body and in the head of
% a rule\footnote{FGTGDs were idependently discovered by Mantas \v{S}imkus
% while working on his doctoral thesis.}.  WFGTGDs are syntactically more
% liberal and more succinct than WGTGDs, but conjunctive query answering under
% WFGTGDs is more complex in case of bounded arities.  It can be seen that
% querying under WFGTGDs is no more expressive than querying under WGTGDs.  In
% fact, for every WFGTGD set $\dep$ and CQ $Q$, there exists a WGTGD set
% $\dep'$ and a CQ $Q'$ such that for every database $D$, $D\cup\dep\models Q$
% iff $D\cup \dep'\models Q'$.

The \datalogpm~family~\cite{CaGP11} has been proposed with the purpose of
providing tractable query answering algorithms for more general ontology
languages.  In \datalogpm, the fundamental constraints are TGDs and EGDs.
Clearly, TGDs are an extension of Datalog rules. The absence of value invention
(existential quantification in the head), thoroughly discussed
in~\cite{PaHo07}, is the main shortcoming of plain Datalog in modeling
ontologies
% ontological reasoning, 
and even conceptual data formalisms such as the Entity-Relationship
model~\cite{er-chen-tods76}.  Sets of GTGDs or WGTGDs are \datalogpm{}
ontologies.  \datalogpm{} languages easily extend the most common tractable
ontology languages; in particular, the main DL-Lite languages
(see~\cite{CaGL12}).
The fundamental decidability paradigms in the \datalogpm{} family are the
following:

\begin{itemize} \itemsep-\parsep
\item \textit{Chase termination.}  When the chase terminates, a finite instance
  is produced; obviously, by Theorem~\ref{the:ans-by-chase}, query answering in
  such a case is decidable.  The most notable syntactic restriction
  guaranteeing chase termination is \emph{weak acyclicity} of TGDs, for which
  we refer the reader to the milestone paper~\cite{FKMP05}.  More general
  syntactic restrictions are studied
  in~\cite{DeNR08,Marn09,GrST11,BLMS11,CHKM*12}.  A semantic property of TGDs,
  called \emph{parsimony}, is introduced in~\cite{LMTV12}.  Parsimony ensures
  decidability of query answering by termination of a special version of chase,
  called \emph{parsimonious} chase.
\item \textit{Guardedness.}  This is the paradigm we studied in this paper.  A
  thorough study of the data complexity of query answering under GTGDs and
  \emph{linear TGDs}, a subset of the guarded class, is found in~\cite{CaGL12}.
  %% Moved from above in the conclusions:
  The interesting classes of \emph{frontier guarded (FGTGDs)} and \emph{weakly
    frontier-guarded TGDs (WFGTGDs)} were considered and studied
  in~\cite{BLMS11,BMRT11,KrRu11}.  The idea underlying these classes is that,
  to obtain decidability, it is sufficient to guard only \emph{frontier
    variables}, that is, variables that occur both in the body and in the head
  of a rule.\footnote{FGTGDs were independently discovered by Mantas \v{S}imkus
    while working on his doctoral thesis.}  WFGTGDs are syntactically more
  liberal and more succinct than WGTGDs, but conjunctive query answering under
  WFGTGDs is computationally more expensive in case of bounded arities.  It can
  be seen that querying under WFGTGDs is no more expressive than querying under
  WGTGDs.  In fact, for every WFGTGD set $\dep$ and CQ $Q$, there exists a
  WGTGD set $\dep'$ and a CQ $Q'$ such that for every database $D$,
  $D\cup\dep\models Q$ iff $D\cup \dep'\models Q'$.
  %% 
  % Extensions of WGTGDs have been proposed in~\cite{BLMS11}.  In that paper,
  % the authors propose the class of \emph{weakly frontier-guarded TGDs}, which
  % is defined by making the guard-atom to contain all body-variables that
  % appear only in affected positions and in the head of the rule.
  % Interestingly, our proofs and techniques, developed for WGTGDs, can be
  % straightforwardly extended to capture weakly frontier-guarded TGDs with
  % minimal changes.
  A generalization of WFGTGDs, called \emph{greedy bounded-treewidth TGDs}, was
  proposed in~\cite{BMRT11}, together with a complexity analysis.  The
  guardedness paradigm has been combined with acyclicity in~\cite{KrRu11},
  where a generalization of both WFGTGDs and weakly acyclic TGDs is proposed.
\item \textit{Stickiness.}  The class of \emph{sticky sets of TGDs} (or
  \emph{sticky \datalogpm}, see~\cite{CaGP12}) is defined by means of syntactic
  restriction on the rule bodies, which ensure that each sticky set of TGDs is
  first-order rewritable, being a FUS, according to~\cite{BLMS11}.  The
  work~\cite{CiRo12} has proposed an extension of sticky sets of TGDs.
\end{itemize}

The interaction between equality generating dependencies and TGDs has been the
subject of several works, starting from~\cite{JoK84}, which deals with
functional and inclusion dependencies, proposing a class of inclusion
dependencies called \emph{key-based}, which, intuitively, has no interaction
with key dependencies thanks to syntactic restrictions.  The absence of
interaction between EGDs and TGDs is captured by the notion of
\emph{separability}, first introduced in~\cite{CaLR03} for key and inclusion
dependencies, and also adopted, though sometimes not explicitly stated, for
instance, in~\cite{CaGP12a,ACKZ09,CDLL*07}---see~\cite{CGOP12} for a survey on
the topic.

As shown in~\cite{CaGL12}, stratified negation can be added straightforwardly
to \datalogpm{}.  More recently, guarded \datalogpm{} was extended by two
versions of well-founded negation
(see~\cite{gottlob-hernich-etal-12,hernich-etal-13}).

In ontological query answering, normally both finite and infinite models of
theories are considered.  In some cases, restricting the attention to finite
solutions (models) only is not always equivalent to the general approach.
% However, when dealing with databases, which always
% have finite size, it is customary to define query answering (see
% Definition~\ref{def:answers}) only on finite instances.  
%%
The property of equivalence between query answering under finite
models % is equivalent to
and query answering under arbitrary models (finite and infinite) is called
\emph{finite controllability}, and it was proved for restricted classes of
functional and inclusion dependencies in~\cite{JoK84}.  Finite controllability
was proved for the class of arbitrary inclusion dependencies in a pioneering
work by Rosati~\cite{Rosa11}.  An even more general result appears in the
work~\cite{BaGO10}, where it is shown that finite controllability holds for
guarded theories.

A related previous approach to guarded logic programming is \emph{guarded open
  answer set programming}~\cite{HeNV05}.  It is easy to see that a set of GTGDs
can be interpreted as a guarded answer set program, as defined
in~\cite{HeNV05}, but guarded answer set programs are more expressive than
GTGDs because they allow negation.

Implementations of ontology-based data access systems take advantage of query
answering techniques for tractable ontologies; in particular, we mention
DLV$^\exists$~\cite{LMTV12}, Mastro~\cite{SLLP*10} and NYAYA~\cite{DOTT12}.

\medskip

\noindent
\textbf{Acknowledgments}\\[.45em]  
This is the extended version of results by the same authors, published in the
KR~2008 Conference and in the DL~2008 Workshop.  
Andrea Cal\`\i{} acknowledges support by the EPSRC project ``Logic-based
Integration and Querying of Unindexed Data'' (EP/E010865/1).
Georg Gottlob acknowledges funding from the
European Research Council under the European Community's Seventh Framework
Program (FP7/2007-2013) / ERC grant agreement DIADEM no.~246858.
Michael Kifer was partially supported by the NSF grant 0964196.
The authors are grateful to Andreas Pieris, Marco Manna and Michael Morak for
their valuable comments and suggestions to improve the paper.

%%% Local Variables: 
%%% mode: latex
%%% TeX-master: "main"
%%% End: 

\vskip 0.2in
\bibliographystyle{plain}
\bibliography{new-main-bib}

\end{document}

%%% Local Variables: 
%%% mode: latex
%%% TeX-master: "main"
%%% End: 